\documentclass[12pt,reqno]{amsart}

\usepackage{amsmath,amsxtra,amssymb,amsthm,amsfonts,eufrak,bm}
\usepackage{hyperref}
\usepackage{cleveref}
\usepackage{amssymb}
\usepackage{bbm}
\usepackage{mathtools}

\usepackage{tikz-cd}
\makeatletter
\newcommand*{\rom}[1]{\expandafter\@slowromancap\romannumeral #1@}
\makeatother
\usetikzlibrary{matrix,arrows,decorations.pathmorphing}

\newcommand{\cH}{{\mathcal H}}

\newcommand{\bphi}{{\bm \phi}}
\newcommand{\bpsi}{{\bm \psi}}

\newcommand{\lel}{\pl = \pl}

\newcommand{\R}{{\mathbb R}}

\newcommand{\C}{{\mathbb C}}
\newcommand{\ten}{\otimes}

\newcommand{\pl}{\hspace{.1cm}}

\newcommand{\ran}{\rangle}
\newcommand{\lan}{\langle}
\newcommand{\al}{\alpha}
\newcommand{\si}{\sigma}

\newcommand{\la}{\lambda}
\newcommand{\eps}{\varepsilon}

\newcommand{\id}{\iota_{\infty,2}^n}

\newcommand{\E}{{\mathcal E}}

\newcommand{\A}{{\mathcal A}}

\newcommand{\M}{{\mathcal M}}
\newcommand{\Ra}{{\mathcal R}}

\renewcommand{\S}{{\mathcal S}}

\newcommand{\Ha}{{\mathcal H}}

\newcommand{\tr}{\operatorname{tr}}

\newcommand{\N}{{\mathcal N}}
\newcommand{\G}{\mathcal{G}}

\newcommand{\norm}[2]{\parallel \! #1 \! \parallel_{#2}}

\newtheorem{lemma}{Lemma}[section]
\newtheorem{prop}[lemma]{Proposition}
\newtheorem{theorem}[lemma]{Theorem}

\newtheorem{rem}[lemma]{Remark}

\newcommand{\re}{\begin{rem}\rm}
\newcommand{\mar}{\end{rem}}

\newcommand{\bra}[1]{\langle{#1}|}
\newcommand{\ket}[1]{|{#1}\rangle}
\newcommand{\ketbra}[1]{|{#1}\rangle\langle{#1}|}

\newcommand{\qd}{\end{proof}\vspace{0.5ex}}
\newcommand{\prf}{\begin{proof}[\bf Proof:]}

\newcommand{\xspace}{\hbox{\kern-2.5pt}}

\renewcommand{\id}{\operatorname{id}}
\renewcommand{\eps}{\varepsilon}
\renewcommand{\al}{\alpha}
\renewcommand{\E}{{\mathcal{E}}}

\newcommand{\TD}{{D^{\textup{ref}}_\alpha}}

\newcommand{\cM}{{\mathcal M}}
\newcommand{\cN}{{\mathcal N}}

\newcommand{\cR}{{\mathcal R}}

\newcommand{\cX}{\mathcal{X}}

\newcommand{\bD}{{\mathbb D}}

\newcommand{\bN}{{\mathbb N}}

\newcommand{\bR}{{\mathbb R}}

\newcommand{\bZ}{\mathbb{Z}}

\newcommand{\spec}{{\text{spec}}}

\renewcommand{\id}{{\text{id}}}

\renewcommand{\id} {\operatorname{id}}

\allowdisplaybreaks
\newtheorem{defi}[lemma]{Definition}
\begin{document}
\title{Asymptotic Equipartition Theorems in von Neumann algebras}

\author[O.~Fawzi]{Omar Fawzi}
\address{Univ Lyon, Inria, ENS Lyon, UCBL, LIP, F-69342, Lyon Cedex 07, France}
%\affiliation{Department of Computer Science, University of Oxford, Wolfson Building, Parks Road, Oxford, UK}
%\affiliation{Perimeter Institute for Theoretical Physics, 31 Caroline Street North, Waterloo,  Ontario, Canada}
%\address{Department of Computer Science, University of Oxford, Wolfson Building, Parks Road, Oxford, UK}
\email{omar.fawzi@ens-lyon.fr}

\author[L.~Gao]{Li  Gao}
\address{University of Houston, Department of Mathematics,
                  Houston, TX 77204-3008}
%\affiliation{Department of Computer Science, University of Oxford, Wolfson Building, Parks Road, Oxford, UK}
%\affiliation{Perimeter Institute for Theoretical Physics, 31 Caroline Street North, Waterloo,  Ontario, Canada}
%\address{Department of Computer Science, University of Oxford, Wolfson Building, Parks Road, Oxford, UK}
\email{lgao12@uh.edu}

\author[M.~Rahaman]{Mizanur Rahaman*}

\address{Univ Lyon, Inria, ENS Lyon, UCBL, LIP, F-69342, Lyon Cedex 07, France}
%\affiliation{Department of Computer Science, University of Oxford, Wolfson Building, Parks Road, Oxford, UK}
%\affiliation{Perimeter Institute for Theoretical Physics, 31 Caroline Street North, Waterloo,  Ontario, Canada}
%\address{Department of Computer Science, University of Oxford, Wolfson Building, Parks Road, Oxford, UK}
\email{mizanur.rahaman@ens-lyon.fr}
\thanks{* corresponding author.}
%\author{Omar Fawzi, Li Gao, and Mizanur Rahaman}
\maketitle
%\author{Omar Fawzi, Li Gao, and Mizanur Rahaman}
%\maketitle
\begin{abstract}
 The Asymptotic Equipartition Property (AEP) in information theory shows that independent and identically distributed (i.i.d.) states behave as uniform states on its typical subspace. In particular, such a phenomenon can be expressed as that asymptotically, the min and the max relative entropy under appropriate smoothing  coincide with the relative entropy. In this paper, we generalize several such equipartition properties to states on general von Neumann algebras.

First, we show that the smooth max relative entropy of i.i.d. states on a von Neumann algebra has an asymptotic rate given by the quantum relative entropy.
In fact, our AEP not only applies to states, but also to quantum channels with appropriate restrictions. In addition, going beyond the i.i.d. assumption, we show that for states that are produced by a sequential process of quantum channels, the smooth max relative entropy
can be upper bounded by the sum of appropriate channel relative entropies. Our main technical contributions are to extend to the context of general von Neumann algebras a chain rule for quantum channels, as well as an additivity result for the channel relative entropy with a replacer channel.
\end{abstract}

\section{Introduction}

The Asymptotic Equipartition Property (AEP) is a fundamental result in information theory. In the finite-dimensional classical setting, the simplest version states that the distribution $P^{\otimes n}$ of independent and identically distributed (i.i.d.) samples of a probability distribution $P$ is close to having properties of a uniform distribution on its typical subset of size approximately $2^{n H(P)}$, where $H(P)$ denotes the Shannon entropy of $P$ (c.f. \cite[Chapter 3]{thoma-cover}). It is convenient to express the AEP in terms of entropic quantities. Recall that the $\eps$-smooth min and max divergences between a probability distribution $P$ and a nonnegative function $Q$ on the finite set $\cX$ are defined as
\begin{align*}
D^{\eps}_{\min}(P \| Q) &= \sup_{P' \sim^{\eps} P} D_{\min}(P' \| Q) \qquad \text{ with } \qquad D_{\min}(P \| Q) = -\log \sum_{x \in \cX : P(x) > 0} Q(x) \\
D^{\eps}_{\max}(P \| Q) &= \inf_{P' \sim^{\eps} P} D_{\max}(P' \| Q) \qquad \text{ with } \qquad D_{\max}(P \| Q) = \log \max_{x \in \cX} \frac{P(x)}{Q(x)} \ ,
\end{align*}
where we refer to Section~\ref{vN} below for the precise definition of the distance used to define $P' \sim^{\eps} P$. With this notation, $2^{-D_{\min}^{\eps}(P^{\otimes n} \| Q^{\otimes n})}$ is the approximate size of the support of $P^{\otimes n}$ under the measure $Q^{\otimes n}$, and $2^{D_{\max}^{\eps}(P^{\otimes n} \| Q^{\otimes n})}$ is the approximate maximum value taken by the function $\frac{dP^n}{dQ^n}$. Recall that the relative entropy is defined as $D(P \| Q) = \sum_{x \in \cX} P(x) \log \frac{P(x)}{Q(x)}$, and the Shannon entropy is $H(P)=-D(P||1)$ with $1$ being the constant function on $\cX$. Then the AEP discussed above can be stated as follows: for any $\eps \in (0,1)$,
\begin{align}
D_{\min}^{\eps}(P^{\otimes n} \| Q^{\otimes n}) &= n D(P \| Q) + O(\sqrt{n}) \label{eq:dmin_aep} \\
D_{\max}^{\eps}(P^{\otimes n} \| Q^{\otimes n}) &= n D(P \| Q) + O(\sqrt{n}) \label{eq:dmax_aep} \ .
\end{align}
where $O(\cdot)$ notation hides a dependence on $P,Q$ and $\eps$.

%Given its importance, the AEP has been established in many different settings. For continuous distributions, by replacing the counting measure by the Lebesgue measure, the Shannon entropy becomes the differential entropy and the corresponding AEP is presented in~\cite[Chapter 8]{cover1999elements}.
In quantum information theory, Schumacher's compression theorem~\cite{schumacher1995quantum} can be interpreted as an AEP for
finite-dimensional quantum system, where the Shannon entropy is replaced by the von Neumann entropy. The AEP has also been established in the classical and quantum case, with the second argument $Q$ in~\eqref{eq:dmax_aep} and~\eqref{eq:dmin_aep} not necessarily being the constant function but an arbitrary nonnegative function~\cite{renner2008security,tomamichel,furrer}. In fact, for an appropriately chosen $Q$,  the quantities $D^{\eps}_{\max}(P \| Q)$ and $D^{\eps}_{\min}(P \| Q)$ are directly related to Renner's smooth min and max conditional entropies~\cite{renner2008security}. We also note that the constant factors in the $O(\sqrt{n})$ term in~\eqref{eq:dmax_aep} and~\eqref{eq:dmin_aep} are well understood: this was established by~\cite{toma-haya, ke-Li-1} and in the infinite-dimensional case in~\cite{datta2016second,wilde-2019,ke-Li}.

The AEP has applications in many areas of quantum information. A notable example, known as the Entropy Accumulation Theorem, has been used to obtain optimal security bounds for general attacks in the context of device-independent cryptography
~\cite{DIQKD, DFR}. Recently there has been significant interests in developing various information-theoretic concepts in finite-dimensional quantum theory to general infinite dimensional quantum systems, including the context of von Neumann algebras modelling quantum field theory and quantum gravity. Indeed, entropic quantities and related entanglement measures, approximate recovery maps and error correction have recently been studied in the abstract von Neumann algebra setting~\cite{HS, FHSW, FLO,jencova-1,hiai-book,  hiai-mosonyi, berta-smooth}.

%In this paper, we establish AEP and a generalization of AEP in the setting of abstract von Neumann algebras.

\subsection{Our results}

%Our starting point in this paper is to establish a general AEP that encompasses all the setups discussed in the previous paragraph and includes new settings that were not yet covered. The framework for doing this is that of states on a von Neumann algebra. A von Neumann algebra $\cM$ is a subalgebra of bounded operators on a Hilbert space $\cH$ satisfying some properties (see Section~\ref{vN} for details). A positive linear functional on $\cM$ is a linear functional $\sigma : \cM \to \C$ such that $\sigma(x^*x) \geq 0$ for all $x \in \cM$ and that satisfies the weak-* continuity property. A state $\rho$ on $\cM$ is a positive linear functional satisfying $\rho(1) = 1$. For example, if $\cH$ is a finite dimensional Hilbert space, then states can be described by $\rho(x) = \tr(\hat{\rho} x)$, where $\hat{\rho}$ is a density operator. Another example is for a measure space $(\cX, \mu)$, the commutative von Neumann algebra $\cM = L^{\infty}(\cX, \mu)$ which is the set of functions $f : \cX \to \R$ that are bounded almost everywhere. In this setting, a state can be described by $\rho(f) = \int \hat{\rho}(x) f(x) d\mu(x)$ where $\hat{\rho}$ is a density function. In this general setting, the statement~\eqref{eq:dmin_aep} was established in~\cite{ke-Li}.

The starting point of this paper is to establish a general AEP that encompasses all the setups discussed above, as well as quantum systems in the algebraic quantum field theory \cite{petz-ohaya}. One particular motivation is to include a general classical-quantum setting, for both discrete and continuous (classical and quantum) variables. The framework for doing so is to consider quantum states on a von Neumann algebra. A von Neumann algebra $\cM$ is a weak$^*$-closed $*$-subalgebra of bounded operators on a Hilbert space $\Ha$ and a state $\rho$ is a positive and unital ($\rho(1)=1$) element in $\M_*$, where the predual $\M_*$ is the set of all normal linear functionals $\M \to \C$. For example, for finite-dimensional quantum systems we let
$\M = B(\Ha)$ for a finite-dimensional Hilbert space $\Ha$ and a quantum state can be described by $\rho(x) = \tr(\hat{\rho} x)$ for all $x \in B(\Ha)$, where $\hat{\rho}$ is a density operator. Another example is a classical system modelled by a measure space $(\cX, \mu)$, then $\M$ is the commutative von Neumann algebra $L_{\infty}(\cX, \mu)$ of functions $f : \cX \to \R$ that are bounded almost everywhere. In this setting, a state can be described by $\rho(f) = \int \hat{\rho}(x) f(x) d\mu(x)$ where $\hat{\rho}$ is a density function.

In the setting of  general  von Neumann algebras, the AEP for the quantum hypothesis relative entropy, which can be interpreted as a generalization of~\eqref{eq:dmin_aep}, was established in~\cite{ke-Li}. Our focus in this paper is on the max relative entropy, i.e., statements generalizing~\eqref{eq:dmax_aep}. Our first result is to establish the AEP for the max relative entropy for general von Neumann algebras. We refer to Section~\ref{vN} for precise definitions of the quantities appearing in the statement below.

%For example, if $\cM=B(\cH)$ for a Hilbert space $\cH$, then states can be described by $\rho(x) = \tr(\hat{\rho} x)$, where $\hat{\rho}$ is a density operator. For a classical system, we choose $\cM = L_{\infty}(\cX, \mu)$ for a measure space $(\cX, \mu)$, a normal state is given by $\rho(f) = \int \hat{\rho}(x) f(x) d\mu(x)$ where $\hat{\rho}$ is a probability density function.
%In the general setting, the AEP for quantum hypothesis relative entropy was established in~\cite{ke-Li} via generalizing a key lemma in \cite{ke-Li-1}.
%Our first result is to establish AEP for $D_{\max}$ divergence~\eqref{eq:dmax_aep} for general von Neumann algebras. We refer to Section~\ref{vN} for precise definitions of the quantities appearing below.

\begin{theorem}[$D_{\max}$ AEP for states]
\label{eq:aep-states}
Let $\cM$ be a von Neumann algebra, $\rho$ be a normal state on $\cM$ and $\sigma$ a normal positive linear functional. Assume that $D(\rho \| \sigma) < \infty, V(\rho \| \sigma) < \infty$, and $T(\rho \| \sigma) < \infty$ are all finite. Then for any $\eps \in (0,1)$ and $n \geq 1$,
\begin{align}
\label{eq:aep-vn}
\frac{1}{n}D_{\max}^{\sqrt{\eps}}(\rho^{\otimes n} \| \sigma^{\otimes n}) =  D( \rho \| \sigma) -\sqrt{\frac{V(\rho \| \sigma)}{n}}\Phi^{-1}(\eps) + O(\frac{\log n}{n}) \ ,
\end{align}
where the $O(\cdot)$ hides constants that only depend on $\rho$ and $\sigma$ and $\eps$.
\end{theorem}

To show the statement above, we extend a lemma of~\cite{Anshu,wilde-2019} relating $D_{\max}^{\eps}$ with the hypothesis testing relative entropy $D_H^\eps$ and combine it with the AEP of $D_H^\eps$ established by  Pautrat and Wang in \cite{ke-Li}. This, in particular, recovers the AEP on $B(\cH)$ for separable Hilbert spaces $\cH$ in~\cite{wilde-2019}. Theorem \ref{eq:aep-states} applies to more general cases: for example, continuous classical-quantum system, whose corresponding von Neumann algebra is $L_{\infty}(\cX, \mu)\ten B(\cH)$. It also applies to Type III von Neumann algebras that do not admit a nontrivial trace functional: for instance, all the local algebras of open bounded regions in relativistic quantum field theory are of Type III \cite{yngvason2005role}.

In fact, we obtain a significantly more general AEP for quantum channels.
Recall that given two von Neumann algebras $\cM$ and $\cN$, a quantum channel $\Phi:\cM_*\to \cN_*$ is the pre-adjoint of a normal unital completely positive map $\Phi ^*:\cN\to \cM$, and the channel divergence for two channels $\Phi$ and $\Psi$ is defined as
\[ D(\Phi \| \Psi)=\sup_{\rho} D(\Phi(\rho)\|\Psi(\rho))\pl,  D^{\text{ref}}(\Phi \| \Psi)=\sup_{\cR}D(\id_{\cR}\ten \Phi \| \id_{\cR}\ten\Psi )\]
where the first supremum is over all quantum state $\rho\in \cM_*$ and the second supremum is over identity channels $\id_{\cR}:\cR_{*}\to \cR_{*}$ of all von Neumann algebras $\cR$. The same definitions apply to other divergences.

\begin{theorem} [$D_{\max}$ AEP for channels]
\label{thm:aep-channel}
Let $\cM, \cN$ be two von Neumann algebras, $\Phi:\cM_*\to \cN_*$ be a quantum channel and $\Psi_{\sigma}$ be a replacer channel with output $\sigma$, a state on $\N$ (i.e. $\Psi_{\sigma}(\rho)=\rho(1)\sigma$). Assume that for some $\alpha > 1$, we have $\TD(\Phi\|\Psi_\sigma)< \infty$, where $D_{\alpha}$ is the sandwiched R\'enyi relative entropy. Then for any $\eps \in (0,1)$ and $n \geq 1$,
\begin{align}
\label{eq:aep-channel}
D_{\max}^{\eps, \textup{ref}}(\Phi^{\otimes n} \| \Psi_{\sigma}^{\otimes n}) = n D^{\textup{ref}}( \Phi \| \Psi_{
\sigma}) + O(\sqrt{n}) \ ,
\end{align}
where the $O(\cdot)$ hides constants that only depend on $\Phi$, $\sigma$ and $\eps$.
\end{theorem}
Note that in the special case where $\Phi$ is a replacer channel with output $\rho$, the channel divergences become state divergences and we recover the statement of Theorem~\ref{eq:aep-states}, though without the precise second order term.
Theorem~\ref{thm:aep-channel} for finite-dimensional quantum system was established in \cite{CMW}. Our proof here makes use of Pisier's non-commutative vector valued $L_p$-space theory \cite{pisier}. As a key ingredient, we obtain that the channel divergence (with reference) of two channels is additive under tensor products when one of the channels is a replacer channel---a fact discovered in \cite{CMW, amortized} in finite dimensions. Finally, for the proof of Theorem \ref{thm:aep-channel}, we show that the max-relative entropy can be upper bounded by the R\'enyi relative entropy in general von Neumann algebras. This fact was proved for finite-dimensional systems in~\cite{datta-renner,tomamichel,datta-hyp}.

Our third result is motivated by recent generalizations of the AEP, called entropy accumulation theorems, that relax the i.i.d. assumption~\cite{DFR,gen-EAT}. A crucial ingredient for entropy accumulation is a chain rule for quantum channels that bounds the increase in relative entropy that can occur when applying different channels to two states. We establish the following chain rule for general von Neumann algebras.
%is the entropy accumulation theorem for sequential applications of channels. The key step is the following chain rule . Recall that the regularized channel divergence is defined as
%and similar definition applies to other type of quantum divergences.
 \begin{theorem}[Chain rule for quantum channels] \label{thm:chain rule1}
Let $\M,\N$ be two von Neumann algebras and $\Phi,\Psi:\M_*\to \N_*$ be two quantum channels. Then for any $\alpha\in (1,\infty]$ and $\rho,\sigma\in \S(\M)$,
\begin{align*}{D}_{\al}(\Phi(\rho) \| \Psi(\si))\leq {D}_{\al}(\rho \| \sigma)+{D}_{\al}^{\textup{reg}}(\Phi\|\Psi),\end{align*}
where ${D}_{\al}$ is the sandwiched R\'enyi relative entropy and

\[ {D}^{\textup{reg}}_{\alpha}(\Phi \| \Psi):=\lim_{n\to\infty}\frac{1}{n}{D}_{\alpha}(\Phi^{\ten n}\| \Psi^{\ten n})\pl.\]
%and $D_{\alpha}(\Phi \| \Psi) = \sup_{\omega \in \S(\M)} D_{\alpha}(\Phi(\rho) \| \Psi(\rho))$.
\end{theorem}
This chain rule was first established by \cite{fawzi2021defining} and recently a short proof was found by \cite{berta-toma}. Our approach here is to use the spectrum truncation technique from Hayashi and Tomamichel \cite{toma-haya} to first approximate the semi-finite case (when the algebra admits a trace). For Type \text{III} algebras, we apply Haagerup reduction \cite{haagerup2010reduction}, which is a method to reduce results of Type \text{III} algebras to the case of finite von Neumann algebras. Our argument shows that Haagerup reduction is well compatible with the entropic inequality. With the chain rule above, we can easily obtain a relative entropy accumulation result in general von Neumann algebras analogus to the one obtained in~\cite{gen-EAT} in finite dimensions.
\begin{theorem} [Relative EAT] \label{thm:REAT}
Let $\eps \in (0,1), n \geq 1$. Let $\M_1,\cdots,\M_n$ and $\M_{n+1}$ be a sequence of von Neumann algebras and $\Phi_i,\Psi_i:(\M_{i})_* \to (\M_{{i+1}})_*, 1\le i\le n$ be two sequences of quantum channels. Then
\[D_{\max}^\eps (\Phi_n\circ\cdots\circ\Phi_1 \| \Psi_n\circ\cdots\circ\Psi_1)\leq \sum_{i=1}^n D_\al^{\textup{reg}}(\Phi_i \| \Psi_i) +  \frac{1}{\al-1} \log \frac{2}{\eps^2} -\log ( \sqrt{1-\eps^2}).\]
\end{theorem}
The sequential setup of Theorem~\ref{thm:REAT} generalizes the i.i.d. setting. In fact, considering $\M_1 = \M^{\otimes n}, \M_2 = \N \otimes \M^{\otimes (n-1)}, \dots, \M_{n+1} = \N^{\otimes n}$ and $\Phi_1 = \Phi \otimes \id_{\M}^{\otimes (n-1)}, \dots, \Phi_n = \id_{\N}^{\otimes (n-1)} \otimes \Phi$ and similarly for $\Psi$, we get
$D_{\max}^\eps (\Phi_n\circ\cdots\circ\Phi_1 \| \Psi_n\circ\cdots\circ\Psi_1) = D_{\max}^\eps (\Phi^{\otimes n} \| \Psi^{\otimes n})$. As a result, Theorem~\ref{thm:REAT} together with the additivity result for replacer channels implies the upper bound in Theorem~\ref{thm:aep-channel}.

In~\cite{gen-EAT}, such a statement is established in the finite-dimensional setting and applied to classes of channels corresponding to a conditional entropy. In this case, it is shown that the regularization in $D_\al^{\textup{reg}}$ is not needed and the resulting statement is applied to obtain security proofs for quantum cryptography.

\section{Preliminary on von Neumann algebras}\label{vN}
We briefly recall the basic definitions for von Neumann algebras and refer the readers to the book \cite{takesaki-1} for more information on this topic. Let $\Ha$ be a (possibly infinite-dimensional) separable Hilbert space and denote $B(\Ha)$ the algebra of bounded operators on $\Ha$. We denote by $1$ the identity operator and $\id$ as the identity map from $B(\Ha)$ to itself.

The predual of $B(\Ha)$ is given by the trace-class operators $T(\Ha)=\{ y\in B(\Ha)\pl |\pl \tr(|y|)
<\infty\}$, where $\tr(x)=\sum_{i}\bra{e_i} x\ket{e_i}$ is the standard matrix trace with $\{\ket{e_i}\}_{i\in I}$ being any orthonormal basis.
 We say a net of operators $\{x_\al\} \in B(\Ha)$ converges in weak$^*$ topology to $x\in B(\Ha)$ if $\tr(yx_\al)\rightarrow \tr(yx)$ for any trace-class operator $y\in T(\Ha)$. A von Neumann algebra $\M$ is a $*$-subalgebra of $B(\Ha)$ of some $\Ha$ that is closed under weak$^*$ topology. An example of a von Neumann algebra is the set $L_\infty(\Omega,\mu)$ of bounded measurable functions for some measure space $(\Omega, \mu)$, which is a subalgebra of $B(L_2(\Omega,\mu))$ as multiplier. %Every commutative von Neumann algebra falls into this class.
%Von Neumann algebras are noncommutative analogs of (bounded measurable) function space $L_\infty(\Omega,\mu)$.
Indeed, whenever (the multiplication in) $\M$ is commutative, $\M \cong L_\infty(\Omega,\mu)$ for some measure space $(\Omega,\mu)$ (see \cite[Theorem 1.18]{takesaki-1}) where $\cong$ denotes $*-$isomorphic von Neumann algebras.
Recall that a linear functional $\phi:\M\rightarrow \C$ is called \begin{enumerate}
\item[i)] normal if it is weak$^*$ continuous;
\item[ii)] positive if $\phi(x^*x)\ge 0$ for any $x\in \M$; and faithful if $\phi(x^*x)= 0$ implies $x=0$;
\item[iii)] unital if $\phi(1)=1$.
\end{enumerate}
We denote $\M_*$ the set of normal linear functionals on $\M$, and $\M_*^+$ be the subset of positive ones.
The functional $\phi$ is a called state if $\phi$ is positive and unital. Throughout the paper, we will only consider normal states. We denote $\S(\M)$ for the set of
 all functionals in $\M_*^+$ that are unital, that is, $\S(\M)=\{\phi\in \M_*^+: \phi(1)=1\}$. The set $\S_{\le}(\M)$ is defined as  $\S_{\le}(\M)=\{\phi\in \M_*^+: \phi(1)\leq 1\}$.

Von Neumann algebras are classified by the existence of traces. Denote $\M_+=\{x^*x \pl | \pl x\in \M\}$ as the positive cone of $\M$.
A trace is a map $\tau: \M_+ \rightarrow [0, +\infty]$ that is positive linear
$\tau(x+\lambda y)=\tau(x)+\lambda\tau(y)\pl, \pl \forall \lambda \ge 0$ and satisfying the tracial property $\tau(xx^*)=\tau(x^*x)\pl, \ \forall \pl x\in \M$. A von Neumann algebra $\M$ is called:\begin{enumerate}
\item[i)] \textbf{finite} if $ \M$ admits a normal faithful trace $\tau$ that is finite,  i.e.  $\tau(1)<\infty$. More precisely, a von Neumann algebra is finite if it admits a sufficient (i.e. separating) family of normal finite traces (see Theorem 2.4 in \cite{takesaki-1}). When $ \M$ is $\sigma$-finite, that is it admits at most  countably many  orthogonal projections, $ \M$ is finite if and only if it admits a faithful normal trace. In this paper we will always assume that $ \M$ is a $\sigma$-finite von Neumann algebras.

%\\
%Examples of finite von Neumann algebras are: a) $L_\infty(\Omega,\mu)$ for some probability measure $\mu$ where the trace is given by the integral $\tau(f)=\int f d\mu$; b) matrix algebra $M_n$ with matrix trace $\tr(x)=\sum_{i=1}^n\bra{i}x\ket{i}$; and, c) a tensor product of above two as the classical-quantum (c-q) systems, $L_\infty(\Omega,\mu)\ten M_n\cong L_\infty(\Omega, M_n)$, which are matrix valued function on $\Omega$ and the natural trace is $\tau(f)=\int_{\Omega}\tr(f(\omega))d\mu$.
\item[ii)] \textbf{semi-finite} if $\M$ admits a normal faithful trace $\tau$ that is semi-finite, i.e. for any $x\in\M_+$, there exists a nonzero $y\in \M_+$ with $y\le x$ such that $\tau(y)<\infty$.
% \\ Examples of semi-finite von Neumann algebras are: a) $L_\infty(\Omega,\mu)$ for some semi-finite measure $\mu(\Omega)=\infty$; b) $B(\Ha)$ for infinite dimensional $\Ha$; and, c) the infinite dimensional c-q system, $L_\infty(\Omega,B(\Ha))$.
\item[iii)] \textbf{type III} if $\M$ is not semi-finite, i.e. does not admit any normal faithful semi-finite trace. Type III algebras model the observables localized in open bounded regions of space-time relativistic quantum field theory.
\end{enumerate}
%One of the main purpose of this work is to extend Entropy Accumulation Theorem \ref{thm:REAT}, which have been obtained in finite dimensional quantum system $B(\Ha)$, to general von Neumann algebras, including the infinite dimensional classical-quantum system $L_\infty(\Omega,B(\Ha))$ and also the Type III algebras without a proper trace.

\subsection{Examples}\label{sec: examples}
In this section, we mention some natural examples of von Neumann algebras appearing in quantum theory.
%Below we mention the common quantum systems that fall in the category where the Entropy Accumulation Theorem holds beyond the finite dimensional systems. These spaces are of interests in quantum information theoretic viewpoint because in the CV-QKD and other cryptographic scenarios, these spaces arise naturally.

\subsubsection{Classical systems:}
%For a measurable space $X$ and a sigma finite measure $\mu$, the von Neumann algebra describing this space is denoted by $\M=L^\infty(X, \mu)$ consisting of all essentially bounded functions on $X$.
Let $\M = L_\infty(X, \mu)$ for some measurable space $X$ and sigma finite measure $\mu$.
A classical state on this system is defined as a normal positive functional on $L_\infty(X, \mu)$ and thus can be identified as an element of the set of all positive integrable functions $L_1(X, \mu)$.
These are the probability densities on $X$. In this case, the trace $\tau$ is just the integration  $f\rightarrow \int_X f d\mu$. As an example, for the position and momentum of a particle, the classical system is obtained for $X=\R$ with the usual Lebesgue measure $\mu$. Any abelian von Neumann algebra is finite (see \cite{takesaki-1}, Definitions 1.15-1.17). Indeed, note that if $\mu$ is a $\sigma$-finite measure, then it is equivalent to a finite measure $\mu_0$, so $L_\infty(X, \mu)\cong L_\infty(X, \mu_0)$ admits a faithful normal finite trace.

%It is simple to see that in this case $L_\infty(\R, \mu)$ is semi-finite. If $\mu$ is a finite measure, then $L_{\infty}(X, \mu)$ is a finite von Neumann algebra.

When $X$ is countable it is natural to take $\mu$ to be the counting measure and in this case, one writes $L_\infty(X, \mu) = l_{\infty}(X)$.
%consider the counting measure  the corresponding von Neumann algebra is $l^{\infty}(X)$ which is defined as
%\[l^{\infty}(X)=\{f: X\rightarrow \C: \sup_x f(x)<\infty\}.\]
A classical state is a normalized linear functional on this space and hence one can identify a state $\rho$ on $l_{\infty}(X)$ as a probability distribution on $X$.
%, that is $\rho \in l^{1}(X)$.
%Once we make this identification, the entropic quantities we are interested now are simply given by the corresponding formula for two measures.

\subsubsection{Quantum systems:}
A von Neumann algebra $\M$ should be seen as the algebra of observables of a quantum system. When $\M$ is semi-finite with a trace $\tau$, we have for any state $\phi$ on $\M$, there exists a positive $\tau$-measurable operator $\rho$ affiliated with $\M$ with $\tau({\rho}) = 1$ such that $\phi(x) = \tau(\rho x)$ for all $x \in \M$.
%$\M=B(\Ha)$, where $\Ha$ is a Hilbert space. For the finite-dimensional $\Ha$, it represents the familiar quantum systems defined on some matrix algebras. For infinite-dimensional $\Ha$, the normal semi-finite trace is the canonical trace $\tr$. In this case, $(L^p(\M), \tau)$ coincide with the Schatten $p$-class of operators. By taking an increasing sequence of finite dimensional projections strongly converging to the identity element, one can realize $B(\Ha)$ as an injective von Neumann algebra. The algebra of observables is identified with $B(\Ha)$ where the physical states are in one-to-one correspondence with the density operators $\rho$. The pure physical states are in one-to-one correspondence with the rank-1 projectors on $\Ha$.

\subsubsection{Composite systems and tensor product:}
If we have two systems described by von Neumann algebras $\M$ and $\N$, then the composite system is described by the tensor product $\M \overline{\otimes} \N$ (see chapter IV in \cite{takesaki-1} for more detail). Let $\M\subseteq B(\Ha)$ and $\N\subseteq B(\mathcal{K})$, we have a natural containment of $\M\otimes \N\subseteq B(\Ha\otimes\mathcal{K})$ of the algebraic tensor product of $\M$ and $\N$ inside
$B(\Ha\otimes\mathcal{K})$. The tensor product of $\M$ and $\N$ (in the von Neumann sense) is denoted by $\M\overline{\otimes}\N$ and defined as the weak$^*$ closure of $\M\otimes \N$ inside $B(\Ha\otimes\mathcal{K})$. In particular, we have $B(\Ha)\overline{\otimes}B(\mathcal{K})=B(\Ha\otimes\mathcal{K})$. Naturally, one can define the tensor product of two states $\rho \in \M$ and $\sigma \in \N$ as the state
%, if $\rho\in \M_*$ and $\sigma\in \N_*$ are two states, then
$\rho\otimes\sigma$ on $\M\bar{\otimes}\N$ which acts on product elements as follows
\[(\rho\otimes\sigma)(a\otimes b)=\rho(a)\sigma(b),\]
for $a \in \M$ and $b\in \N$.
 Since in this paper we will only talk about von Neumann algebras, we may often drop the $\overline{\otimes}$ notations and work simply with ${\otimes}$ with no confusion.

For another example, if $\cM = L_\infty(X, \mu)$ describes a classical system and $B(\Ha)$ describes  a quantum system, then the von Neumann algebra for describing the composite system is $L_\infty(X, \mu)\overline{\otimes} B(\Ha)$. Such algebras are called Type I. In this case, observables can be seen as operator-valued functions.

We note that the tensor product of semi-finite von Neumann algebras is semi-finite.
%It holds that tensor product of two  semi-finite and injective von Neumann algebras is again semi-finite and injective. For a bipartite system where one part $A$ is modelled by a classical system and the other part $B$ is fully quantum, the system is described by the von Neumann algebra $\M=L^\infty(X, \mu)\overline{\otimes} B(H)\simeq L^\infty(X, B(\Ha))$. The states on $\M$ are identified with elements that are integrable functions $f_{AB}$ on $X$ with values in $\S(B(\Ha))$-the state space of $B(\Ha)$ satisfying $\int_X f_{AB}({1} )d\mu (x) = 1$.

\subsubsection{Fermionic systems:}
Here we provide an example of a von Neumann algebra that frequently appears in quantum theory that is finite but not Type I. In fact it is an infinite-dimensional von Neumann algebra that admits a finite tracial state $\tau$, that is, $\tau(1)=1$.
This is called the hyperfinite $\mathrm{II_1}$-factor and it is physically relevant as it often appears in the study of the fermionic algebra that arises through a limiting process
involving finite quantum systems (\cite{kadison-pnas}).

% which is the C$^*$-algebra representing the infinite canonical anticommutation relations arising from the quantum mechanical study of fermions.
%tensor products of the 2 × 2 matrix algebra

Consider the increasing sequence of tensor products of the $2\times 2$ matrix algebra over complex numbers:
\[M_2(\C)\hookrightarrow M_2(\C)^{\otimes 2}\hookrightarrow M_2(\C)^{\otimes 3}\hookrightarrow \cdots, \]
where the inclusion $M_2(\C)^{\otimes n}\ \hookrightarrow M_2(\C)^{\otimes {n+1}}$ is given by $x\rightarrow x\otimes 1$, for each $n\in \mathbb{N}$. Let $\A$ be the norm completion of the inductive limit, that is, $\A=\overline{\bigcup_{n\in \mathbb{N}} M_2(\C)^{\otimes n}}$. Now if $\tr$ is the normalized trace on $M_2(\C)$, then
on $M_2(\C)^{\otimes n}=M_{2^n}(\C)$ has a tracial state $\tau_n$ defined as
\[\tau_n(x_1\otimes \cdots \otimes x_n)=\tr(x_1) \tr(x_2)\cdots \tr(x_n); x_1, x_2, \cdots, x_n\in M_2(\C) .\]
Hence the sequence $\{\tau_n\}_{n=1}^\infty$ defines a unique state $\tau$ on the C$^*$-algebra $\A$. And if $(\Ha_\tau, \xi_\tau, \pi_\tau)$ is the GNS tuple for $\tau$, then we define
\[\Ra:= (\pi_\tau(\A))'',\]
and $\tau$ extends to a normal tracial state on $\Ra$ by $\tau(x)=\langle \xi_\tau, x\xi_\tau\rangle.$
The algebra $(\Ra, \tau)$ is called the hyperfinite $\mathrm{II_1}$ factor.

\subsubsection{Type III algebras}
Let $0<\lambda<1$,  and define the density matrix
\[\omega_\lambda=\Big(\begin{array}{cc}
\frac{1}{(1+\lambda)} & 0\\
0 & \frac{\lambda}{(1+\lambda)}
\end{array}\Big).\]
This defines a state $\psi_\lambda (\cdot):=\mathrm{Tr}(\omega_\lambda \cdot)$ on $M_2(\C)$. Now as in the previous example define the tensor product state $\phi_\lambda$ on the UHF C$^*$-algebra $\A$ as $\phi_\lambda(x_1\otimes \cdots \otimes x_n)=\psi_\lambda(x_1)\cdots \psi_\lambda(x_n)$, $x_i\in M_2(\C)$. Let $\pi_{\phi_{\lambda}}$ be the GNS representation of $\A$ with respect to $\phi_\lambda$. Then the von Neumann algebra $(\pi_{\phi_{\lambda}}(\A)'')$ obtained by taking the weak closure of $\pi_{\phi_{\lambda}}(\A)$ is a Type III algebra (see \cite{witten}). These are known as the Powers factors.

%Now we proceed similarly to the previous example and construct an algebra $\A_0$ out of infinitely many tensor products of $M_2(\C)$, consisting of elements $a=a_1\otimes a_2\otimes \cdots$ such that all but finitely many $a_i$'s are equal to $K_\lambda$. Then the closure of this algebra, denoted by $\A_\lambda$, is a von Neumann algebra which does not admit a non-trivial trace functional and hence a Type III algebra (see \cite{witten}).

\subsection{Relative Entropies}\label{entropies and channels}
We first consider a semi-finite von Neumann algebra $\M$ equipped with a normal faithful semi-finite trace $\tau$.  Recall that (\cite{terp, nelson}) the noncommutative $L_p$-space $L_p(\M)$ (for $1<p<\infty$)  on $(\M, \tau)$ is defined as the space of all $\tau$-measurable operators $a$ affiliated with $\M$ such that $\tau(|a|^p)<\infty$ (where $\tau$ is naturally extended to the positive $\tau$-measurable operators).

We will often write $\norm{\cdot}{p}$ if the algebra $\M$ and the trace $\tau$ is clear. We identify the $L_1$-space $L_1(\M):=\M_*$ with the pre-dual space $\M_*$ via the correspondence
\[a\in L_1(\M)\longleftrightarrow \phi_a\in \M_*,\pl  \phi_a(x)=\tau(ax)\pl.\]
Under this identification, a state $\phi$ is associated with a density operator $\rho \in L_1(\M)$ such that $\rho\ge 0$ and $\tau(\rho)=1$. In the beginning of this section we defined
normal states and substates as linear functionals. In the case where the von Neumann algebra has a tracial state $\tau$ we can identify $L_1(\M):=\M_*$ and we can define normal states and substates in terms of the density operators as follows:
%Throughout the paper, we will often identify normal states with their density operators if no confusion. We denote the (normal) state space of $\M$ as well as sub-states and bounded invertible states as follows
\begin{align*}
&\S(\M)=\{ \rho \in L_1(\M)\pl | \pl \rho \ge 0,\tau(\rho)= 1\},\\
&\S_\le (\M)=\{ \rho \in L_1(\M)\pl  | \pl  \rho \ge 0,\tau(\rho)\le 1\pl\}\pl, \\
&\S_b (\M)=\{ \rho \in \S(\M)\pl  | \pl  m s(\rho)\le \rho \le Ms(\rho) \pl  \text{ for some } 0<m<M<\infty \}\pl,
\end{align*}
where $s(\rho)$ denotes the support of state $\rho$, which is the minimal projection $e$ such that $\rho=e\rho e$. Note that for $\rho \in \S_b(\M)$, $\tau(m s(\rho))\le \tau(\rho)=1 $ implies $\tau(s(\rho))<\infty $. For any $e$ with $\tau(e)<\infty$, $\tau$ restricted on $e\M e$ is a finite trace. In many cases, this allows us to restrict the discussion to finite von Neumann algebras.

We now recall the definition of various relative entropies for two quantum (sub)states. These definitions on finite-dimensional
systems are available in any book on quantum information theory (for example \cite{tomamichel-book}) and their extension to infinite-dimensional systems can be found in documents such as \cite{hiai-book, bst, jencova-1}.

Let $\rho \in \S_\le(\M)$ and $\sigma \in \S(\M)$. In the following unless otherwise mentioned, we always assume the support condition $s(\rho)\le s(\sigma)$.
Recall the the following definitions:
\begin{enumerate}
\item[i)] Relative entropy: \[D(\rho \| \sigma)=\tau(\rho\log \rho -\rho\log\sigma)\pl  \]
provided $\rho(\log \rho-\log\sigma)\in L_1(\M)$.
\item[ii)] Sanwiched-R\'enyi relative entropy:  for $\al\in [1/2,1)\cup (1,\infty), \frac{1}{\al}+\frac{1}{\al'}=1$. \begin{align*}D_{\al}(\rho \| \sigma)=&\frac{\al}{\al-1}\log\norm{\sigma^{-\frac{1}{2\al'}}\rho\sigma^{-\frac{1}{2\al'}}}{\al}
    \\ =&\frac{1}{\al-1}\log\tau(|\sigma^{-\frac{1}{2\al'}}\rho\sigma^{-\frac{1}{2\al'}}|^{\al})
    \end{align*}
\item[iii)] Petz-R\'enyi relative entropy:   for $\al\in (0,1)\cup (1,2)$,
\[\tilde{D}_{\al}(\rho\|\sigma)=\frac{1}{\al-1}\log\tau(\rho^\al\sigma^{1-\al}).\]
\item[iv)] Geometric R\'enyi relative entropy:   for $\al\in (0,1)\cup (1,2)$,
\[\hat{D}_{\al}(\rho\|\sigma)=\frac{1}{\al-1} \log \tau[\sigma^{1/2}(\sigma^{-1/2}\rho\sigma^{-1/2})^\al\sigma^{1/2} ].\]
\end{enumerate}
All the three R\'enyi relative entropies $D_{\al},\tilde{D}_{\al}$ and $\hat{D}_{\al}$ are finite given the corresponding trace terms are well-defined, and infinite otherwise. Moreover, they coincide when $\rho$ and $\sigma$ commute.
When $\al$ goes to $1$, both $D_{\al}$ and $\tilde{D}_{\al}$ recovers the standard relative entropy
\[ \lim_{\al \to 1}D_{\al}(\rho\|\sigma)=\lim_{\al \to 1} \tilde{D}_{\al}(\rho\|\sigma)=D(\rho\|\sigma),\]
provided that the quantities $D_{\al}(\rho\|\sigma)<\infty $ and $ \tilde{D}_{\al}(\rho\|\sigma)<\infty$ for some $\alpha>1$.

For $\al=\infty$, the Sandwiched-R\'enyi relative entropy is the max relative entropy.
\begin{enumerate}
\item[v)] Max-relative entropy: \[D_{\max}(\rho\|\sigma)=\log \inf\{\lambda>0 |\rho\le \lambda \sigma \}\pl .\]
\end{enumerate}
%For $\al=1/2$, $D_{\min}:=D_{1/2}$ is given by the fidelity (note that the $D_{\min}$ here is different with the classical definition in the introduction)
\begin{enumerate}
%\item[vi)] Min-relative entropy: \[D_{\min}(\rho\|\sigma)=-2\log F(\rho,\sigma) \pl, \pl F(\rho,\sigma)=\norm{\sqrt{\rho}\sqrt{\sigma}}{1} .\]
\item[vii)] Information divergence: \[V(\rho\|\sigma)=\tau(\rho(\log \rho -\log\sigma)^2)\] provided $\rho(\log \rho -\log\sigma)^2\in L_1(\M)$. Third order term

    \[ T(\rho\|\sigma)=\tau(\rho(\log \rho -\log\sigma)^3)\]
    provided $\rho(\log \rho -\log\sigma)^3\in L_1(\M)$.
\item[viii)] Information spectrum: \[D_{s}^\eps(\rho\|\sigma)=\sup\{  \log \lambda \ | \ \tau(\rho\{\rho\le \lambda \sigma\}) \le \eps\} \]
where $\{\rho\le \lambda\sigma\}$ is the support projection of the positive part $(\lambda\sigma-\rho)_+$.
\item[ix)] \label{eq:hyp}  Hypothesis testing relative entropy: \[D_{H}^\eps(\rho\|\sigma)=-\log\inf\{ \tau(Q\sigma) \pl | \pl 0\le Q\le 1\pl, \tau(Q\rho)\ge 1-\eps \}. \]

\end{enumerate}
All the logarithms above and in the following are in base $2$. Recall the purified distance between two substates $\rho $ and $\sigma$ is
\[ d(\rho,\sigma)= \sqrt{1-F(\rho,\sigma)^2}\pl, \pl F(\rho,\sigma)=\tau(|\sqrt{\rho}\sqrt{\sigma}|)+\sqrt{(1-\tau(\rho))(1-\tau(\sigma))}\pl,\]
where $F(\rho,\sigma)$ is the generalized fidelity for $\rho,\sigma\in S_\le(\M)$. Following \cite{toma-haya}, we say two substates $\rho$ and $\rho'$ are $\eps$-close, denoted as $\rho'\sim^{\eps} \rho$, if $d(\rho',\rho)\le \eps$.
 For $0\le \eps\le 1$, we consider the smoothed version of max relative entropy,
\begin{enumerate}
\item[x)] Smooth max-relative entropy: \[D_{\max}^\eps(\rho\|\sigma)=\inf_{\rho'\sim^{\eps}\rho}D_{\max}(\rho'\|\sigma),\]
where the infimum is over all substates $\rho' \in S_{\le}(\M)$ that are $\eps$-close to $\rho$.
\end{enumerate}

\subsection{Relative entropies on general von Neumann algebras}\label{subs-ent. in gen vN}
Most often, the algebras considered in quantum information theory all admit a proper trace. Algebras of type $\text{III}$ do not admit any non-trivial trace functional, not even semi-finite. However, the type $\text{III}$ algebras are relevant models in quantum field theory (see \cite{witten}). Smooth (relative) entropies were first considered in general von Neumann algebras in \cite{berta-smooth}.
We follow the monograph \cite{hiai-book} to use the standard form to define relative entropy. Recall that the standard form $(\M,\cH,J,P)$ of a von Neumann algebra $\M$ is given by an injective $*$-homomorphism~$\pi:\M \to B(\cH)$,  an anti-linear isometry $J$ on $\cH$, and a self-dual cone $P$ such that
\begin{itemize}
\item[i)] $J^2=1$, $J\M J=\M'$,
\item[ii)] $JaJ=a^*$ for $a\in \M\cap \M'$,
\item[iii)] $J\xi=\xi$ for $\xi\in P$,
\item[iv)] $aJaJP=P$ for $a \in \M$,
\end{itemize}
where $\M' \coloneqq \{x\in B(\cH)\pl |\pl xa=ax \quad \forall  a\in \M\}$ is the commutant of $\M$.
Such standard form is unique up to unitary equivalence. For example, when $\M$ is $\sigma$-finite equipped with a normal faithful state $\psi$, the standard form is basically given on the Hilbert space $L_2(\M,\psi)$
\[ \lan a,b\ran_\psi=\psi(a^*b)\pl,\qquad \norm{a}{L_2(\M,\psi)}^2=\lan a, a\ran_\psi\pl, \]
and $\M$ acts on $L_2(\M,\psi)$ via the GNS representation $\pi:\M\to L_2(\M,\psi)$
\[
\pi(x) \ket{a}= \ket{xa} \pl , \pl \forall x\in \M\pl, \pl \ket{a}\in L_2(\M,\psi)\pl.
\]

For each positive linear functional $\phi\in \M_*^+$, there exists a unique vector $\xi_\phi\in P$ implementing $\phi(x)=\bra{\xi_\phi}x\ket{\xi_\phi}$.  Recall (\cite{takesaki-1}) that for a $\phi\in\M_*^+$, the support projection of $\phi$ is the smallest projection $e\in \M$ such that $\phi(x)=\phi(e x e)$, for any $x\in \M$. The projection $e$ is often known as the support projection of $\phi$ and denoted as $s(\phi)$.

The support projection $s(\phi)$ is then the projection onto $\overline{\M'\xi_\phi}$.
Given two positive linear functional $\phi,\psi\in \M_*^+$, we define the operator $S_{\psi,\phi}$ as follows:
\begin{equation}
\label{eq:S-op-def}
S_{\psi,\phi}(a\xi_\phi+\eta)=s(\phi)a^*\xi_\psi\pl, \pl a\in \M\pl,
\end{equation}
where $a\xi_\phi\in \overline{\M\xi_\phi}, \eta\in \M\xi_\phi^\perp$. Then $S_{\bpsi,\bphi}$ is a closable anti-linear operator, and the relative modular operator is the positive self-adjoint operator defined as
\begin{equation}
\label{eq:rel-mod-op-def}
\Delta(\psi,\phi) \coloneqq (S_{\psi,\phi})^*\bar{S}_{\psi,\phi}\pl,
\end{equation}
where $\bar{S}_{\psi,\phi}$ is the closure of $S_{\psi,\phi}$.
%By the symmetric role of $\M$ and $\M'$, we have
%\begin{align}
%\label{inv}\Delta(\phi,\psi)=J\Delta(\psi,\phi)^{-1}J.
%\end{align}
In particular, the modular operator of $\phi$ is $\Delta_\phi:=\Delta(\phi,\phi)$
 and the modular automorphism group $\al_t^{\phi}:\M\to \M\pl$ is as follows:
 \[
 \al_t^{\phi}(x)=\Delta_\phi^{-it} x \Delta_\phi^{it}.
 \]
Let $\rho\in S_{\le}(\M)$ and $\sigma \in \M_*^+$. Again, we always assume $s(\rho) \le s(\sigma)$.  Many quantum divergences between $\rho$ and $\sigma $ are defined using the relative modular operator $\Delta(\rho,\sigma)$ (see \cite{hiai-book} \cite{gao2021recoverability}). Given the spectral decomposition of $\Delta(\rho,\sigma)$ as
$\Delta(\rho,\sigma)=\int_{[0,\infty)} t dE_{\rho,\sigma}(t)$, then for any real function $f$ on the interval  $(0, \infty)$, we have the integral expression
$$\lan \xi_\rho| f(\Delta(\rho,\sigma))\xi_\rho \ran=\int_{0}^\infty  f(t) d\|E_{\rho,\sigma}(t) \xi_\rho\|^2.$$

\begin{enumerate}
\item[i)] Relative entropy: \[D(\rho\|\sigma)=\lan \xi_\rho| \log\Delta(\rho,\sigma) | \xi_\rho \ran \pl .\]
where $\xi_\rho\in P$ is the vector implementing $\rho$.

\noindent This definition is due to H. Araki (\cite{Araki}). A useful variational expression of the relative entropy was put forward by H. Kosaki (\cite{Kosaki}).

We also define
Information variance:
\[V(\rho\|\sigma)=\lan \xi_\rho| (\log\Delta(\rho,\sigma))^2 | \xi_\rho \ran \pl .\]
The third-order quantity:
\[T(\rho\|\sigma)=\lan \xi_\rho| (\log\Delta(\rho,\sigma))^3 | \xi_\rho \ran \pl .\]
\item[ii)] Sandwiched-R\'enyi relative entropy:  for $\al\in (1,\infty), \frac{1}{\al}+\frac{1}{\al'}=1$. \[D_{\al}(\rho\|\sigma)=\al'\log \sup_{\omega\in \S(\M)} \bra{\xi_\rho} \Delta(\omega,\si)^{1/\al'}\ket{\xi_\rho}\pl.\]
\noindent Note that the Sandwiched-R\'enyi relative entropy in the context of von Neumann algebras was introduced for $\al\in [1/2, 1)\cup (1,\infty)$ in \cite{bst} and a different
but equivalent definition was given by Jen{\v{c}}ov\'a (\cite{jencova-1, jencova-2}).
\item[iii)] Petz-R\'enyi relative entropy:  for $\al\in (0,1)\cup (1,2)$, \[\tilde{D}_{\al}(\rho\|\sigma)=\frac{1}{\al-1}\log \bra{\xi_\rho} \Delta(\omega,\rho)^{1-\al}\ket{\xi_\rho}\pl.\]
\noindent For $\al>1$, it is proved in \cite{jencova-1} that \begin{align}\tilde{D}_{2-\frac{1}{\al}}(\rho\|\sigma)\le D_\al(\rho||\sigma)\le \tilde{D}_{\al}(\rho\|\sigma).\label{eq:bound}\end{align}
\item[iv)] Max relative entropy:   \[D_{\max}(\rho\|\sigma)=\log \inf\{ \lambda \pl |\pl \lambda\sigma-\rho\in\M_*^+\}\pl.\]
\item[v)] Fidelity: \[F(\rho,\sigma)= \inf_{\omega \in \mathcal{S}(\mathcal{M})} \bra{\xi_\rho} \Delta(\omega,\si)^{-1}\ket{\xi_\rho}\pl.\]

\noindent The fidelity function was first introduced by A. Uhlmann \cite{uhlmann} and a useful variational expression
\begin{align}\label{eq:fe}F(\rho, \sigma)= \frac{1}{2}\inf_{x\in \M^{++}}\{\rho(x)+\sigma (x^{-1})\},\end{align}
where $\M^{++}$ is the set of all positive invertible elements of $\M$, was provided by P. Alberti and Uhlmann \cite{Alberti-Uhlmann}. It follows that $F$ is jointly upper semi-continuous in the weak-topology. Also, note that $F(\rho, \sigma)=2^{-\frac{1}{2} D_{1/2}(\rho\| \sigma)}$, although we will not use the Sandwiched-R\'enyi relative entropy $D_{\alpha}$ for $\alpha<1$ in this article.

    \item[vi)] Smooth max relative entropy:   \[D_{\max}^\eps(\rho\|\sigma)= \inf_{\rho'\sim^\eps \rho} D_{\max}(\rho'\|\sigma)\]
    where $\rho'\sim^\eps \rho$ means the purified distance $d(\rho,\rho')=\sqrt{1-F(\rho,\rho')^2}\le \eps$.
    \item[vii)] $\eps$-Hypothesis relative entropy: \[D_H^\eps(\rho\|\sigma)=-\log\inf\{\sigma(Q)| 0\leq Q\leq 1, \rho(Q)\geq 1-\eps\}.\]
\end{enumerate}

\subsection{Quantum channels} \label{channels}
A common property shared by all above relative entropies is the data processing inequality, also called monotonicity under quantum channels. Let $\M$ and $\N$ be two von Neumann algebras. A linear map $\Phi^*: \M\to \N$ that is unital, normal and completely positive is called a quantum Markov map. The pre-adjoint map $\Phi: \N_*\rightarrow \M_*$ is called a quantum channel, which sends normal states to normal states as follows
\[\Phi(\rho)(x)=\rho(\Phi^*(x))\pl , \pl \forall \rho\in \N_*\pl, \pl  x\in \M\pl.\]
When the von Neumann algebra admits a trace $\tau$, at the level of preduals $L_1(\M)\cong \M_*, L_1(\N)\cong \N_*$, the map $\Phi$ is given by
\[\tau(\Phi(\rho) x)=\tau(\rho\Phi^*(x))\pl , \pl \forall \rho\in L_1(\N)\pl, \pl  x\in \M\pl.\]
In this case, $\Phi$ is completely positive and trace preserving (CPTP). For $\bD=D,D_\al, \tilde{D}_\al, \hat{D}_\al, D_H^\eps$ or $ D_{\max}^\eps$, the \textbf{data processing inequality} states that for any states $\rho,\sigma\in \S(\M),$ and every quantum channel $\Phi:\M_*\to\N_*$ we have
\[ \bD(\rho\|\sigma)\ge \bD(\Phi(\rho)\|\Phi(\sigma)).\]

\subsection{Haagerup Reduction}\label{Haagerup reduction} It turns out that the relative entropies in the Type III algebras can always be approximated by semifinite cases via Haagerup's reduction.
Here we briefly review the basics of Haagerup's reduction and refer to \cite{haagerup2010reduction} for more details. The key idea is to consider the additive subgroup $G=\bigcup_{n\in \bN} 2^{-n}\bZ \subset \mathbb{R}$. Let $\M\subset B(\cH)$ be a von Neumann algebra and $\phi$ be a normal faithful state. Recall  the modular automorphism group  $\al^\phi_t:\M\to\M, t\in \bR$
\[ \al^\phi_t(x)=\Delta_\phi^{it}x\Delta_\phi^{-it} \pl \forall\pl  x\in \M,\]
where $\Delta_\phi=\Delta(\phi, \phi)$ is the modular operator of $\phi$.
One can define the crossed product by the action $\al^\phi:G\curvearrowright \M$
\[ \hat{\M}\lel \M\rtimes_{\al^{\phi}}G \pl. \]
$\hat{\M}$ is the von Neumann subalgebra  $ \hat{\M}=\{\pi(\M),\la(G)\}'' \subset \M\overline\ten B(\ell_2(G))$ generated by the embeddings
\begin{align}\label{eq:cross}
  \pi:\M\to \M\rtimes_{\al^{\phi}}G\pl, \pl &\pi(a)\lel \sum_{g} \al_{g^{-1}}(a) \ten \ketbra{g}\nonumber\\
  \lambda:G\to \M\rtimes_{\al^{\phi}}G,\pl  &\la(g)(\ket{x}\ten \ket{h})=\ket{x}\ten \ket{gh}, \text{ and } \pl \forall \pl \ket{x}\in \cH\pl, \ket{h}\in \ell_2(G) \pl.
  \end{align}
Here $\pi$ is the transference homomorphism $\M\to \ell_\infty(G,\M)$, and $\lambda$ is the left regular representation on $\ell_2(G)$.
 The set of finite sums
 $\{\sum_g a_g \la(g)\pl | \pl  a_g \in \M \}\subset \hat{\M}$ forms a  $w^*$-dense subalgebra of $\hat{\M}$.
 In the following, we identify $\M$ with $\pi(\M)$ (resp. $a$ with $\pi(a)$) and view $\M\subset \hat{\M}$ as a subalgebra. There is a canonical normal conditional expectation
 \[E_\M:\hat{\M}\to \M\pl, \pl E_\M(\sum_g a_g\la(g))=a_0\pl.\]
 For a normal state $\psi\in \S(\M)$, $\hat{\psi}=\psi\circ E_\M$ gives a natural extension on $\hat{\M}$ such that
 \[ \hat{\psi}(\sum_g a_g \la(g)) \lel \psi(a_0) \pl.\]

The main object in Haagerup's construction is an increasing family of subalgebras $\M_n, n\ge 1$, whose key properties are summarized in \cite[Theorem 2.1 \& Lemma 2.7]{haagerup2010reduction}
\begin{theorem}\label{thm:hr}
There exists an increasing family of von Neumann subalgebras $\M_n,n\ge 1$ satisfying the following properties
\begin{enumerate}
\item Each $(\M_n,\tau_n)$ is a finite von Neumann algebra equipped with some normal faithful tracial state $\tau_n$.
\item $\bigcup_{n\ge 1} \M_n$ is weak$^*$-dense in $\hat{\M}$.
\item There exists $\hat{\phi}$-preserving normal faithful conditional expectation $E_{\M_n}:\hat{\M}\to \M_n$ such that
\[ \hat{\phi}\circ E_{\M_n}=\hat{\phi}\pl, \pl \al_t^{\hat{\phi}}\circ E_{\M_n}= E_{\M_n}\circ \al_t^{\hat{\phi}}\pl. \]
Moreover, $E_{\M_n}(x)\to x$ in $\sigma$-strong topology for any $x\in \hat{\M}$.
\end{enumerate}
\end{theorem}
 For a state $\rho\in \S(\M)$,  $\hat{\rho}=\rho\circ E_\M$ is the canonical extension on $\hat{\M}$ such that
 \[ \hat{\rho}(\sum_g a_g \la(g)) \lel \rho(a_0) \pl.\]
 We denote $\rho_n:=\hat{\rho}|_{\M_n}\in S(\M_{n})$ as the restriction state of $\hat{\rho}$ on the subalgebra $\M_n \subset\hat{\M}$. Note that the predual $\M_{n,*}$ can be viewed as a subspace of $\hat{\M}_*$ via the embedding
\[ \iota_{n,*}: {\M}_{n, *}\to \hat{\M}_*\pl,  \iota_{n,*}(\omega)=\omega\circ E_{\M_n}\pl. \]
 Via this identification, $\rho_n=\hat{\rho}|_{\M_n}\circ E_{\M_n}=\hat{\rho}\circ E_{\M_n}=E_{\M_n,*}(\hat{\rho})\in \hat{\M}_*$. Moreover, it was proved in \cite[Theorem 3.1]{haagerup2010reduction}, $\rho_n\to \hat{\rho}$ converges in the norm topology. One immediate consequence is the following approximation of relative entropies.
\begin{prop}
\label{lemma:ultra}
Let $\rho$ and $\sigma$ be two normal states of $\M$. Fix $\eps\in (0,1)$.  Then for $\bD=D,D_\al,D_{\max}$, $ D_{\max}^\eps$, or $D_H^\eps$ with
 \[ \bD(\rho\|\si) \lel \bD(\hat{\rho}\|\hat{\si})=\lim_{n\to \infty} \bD(\rho_n\|\sigma_n) \pl .\]
\end{prop}
\begin{proof} The proof is identical to the argument of $D$ in \cite[Lemma 4.4]{gao2022complete} by using data processing inequality and lower semi-continuity.
The lower semi-continuity of $D,D_\al,D_{\max}$ with respect to weak$^*$-topology were proved in \cite{jencova-1}. The next lemma verifies lower semi-continuity for $D_{\max}^\eps$ and $D_H^\eps$ of a fixed $\eps$.
\end{proof}
\begin{lemma}For any $\eps\in(0,1)$,
$D_{\max}^\eps$ and $D_H^\eps$ are lower semi-continuous with respect to the weak$^*$-topology.
\end{lemma}
\begin{proof}For simplicity of notations, we discuss sequences but the argument works for nets.
Let $\M$ be a von Neumann algebra and $\rho_n,\sigma_n\in \S(\M)$ be two sequences of normal states that converge to $\rho$ and $\sigma$ respectively in $\S(\M)$ in norm topology. We show
\[ \liminf_{n}D_{\max}^\eps(\rho_n\|\sigma_n)\ge D_{\max}^\eps(\rho\|\sigma)\pl.\]
We can assume that $ \liminf_{n}D_{\max}^\eps(\rho_n\|\sigma_n)<\infty$. For any $n$ and $\delta>0$, take $\rho_n' \sim ^\eps\rho_n$ such that
 \[\rho_n' \leq \lambda_n \sigma_n \ \text{with} \ \lambda_n\leq 2^{D_{\max}^{\epsilon}(\rho_n\|\sigma_n)+\delta}.\]
 By the weak$^*$-compactness of the unit ball in $\M^*$, $\rho_n'$ admits a weak$^*$-limit point $\rho'\in (\M^*)_+$. Hence for some subsequence $\{n_k\}$, we have $\rho'_{n_k}\to \rho'$ in the weak$^*$-topology, where we can further assume that $D_{\max}^\eps(\rho_{n_k}\|\sigma_{n_k})\to  \liminf_{n}D_{\max}^\eps(\rho_n\|\sigma_n)$. Now we have to show that the weak$^*$-limit  $\rho'\in \M^*$ is in $\S(\M)\subset \M_*$. Note that
 \[\lambda :=  \liminf_{k} \lambda_{n_k}\leq 2^{\lim_k D_{\max}^{\epsilon}(\rho_{n_k}\|\sigma_{n_k})+\delta}=2^{\liminf_{n}D_{\max}^\eps(\rho_n\|\sigma_n)+\delta}<\infty. \]

Since $\sigma_n\rightarrow \sigma$ weakly, we have for every $x\in \M_+$
\[\rho'(x)=\lim _{k} \rho_{n_k}'(x)\leq \liminf_{k} \lambda_{n_k} \sigma_{n_k} (x)=\lambda \sigma(x).\]
Hence from  \cite[pp. 127]{takesaki-1}), it follows that $\rho'\in  \M_*$, so that $\rho'\in \S(\M)$.

%For all $n$ and $\delta>0$, take $\rho_n' \sim ^\eps\rho_n$ such that
 %\[\rho_n' \leq \lambda_n \sigma_n \ \text{with} \ \lambda_n\leq 2^{D_{\max}^{\epsilon}(\rho_n\|\sigma_n)+\delta}.\]
 %By the weak$^*$-compactness of the unit ball in $\M^*$, $\rho_n'$ admits a weak$^*$-limit point $\rho'\in (\M^*)_+$. Since $\sigma_n\rightarrow \sigma$ weakly, we have for every positive $x\in \M_+$
 %\[\rho'(x)=\lim _{n_k} \rho_{n_k}'(x)\leq \liminf_{n_k} \lambda_{n_k} \sigma_{n_k} (x)=(\lim\inf_k \lambda_{n_k}) \sigma(x).\]
% Now we have to show that the weak$^*$ limit $\rho'\in\M^*$ is in $\S(\M)\subset \M_*$).  We can assume that $\lambda :=\lim\inf_n D_{\max}^{\epsilon}(\rho_n\|\sigma_n)<\infty$ and so $\rho'\leq \lambda \sigma$. Hence by the orthogonal decomposition of normal and singular part (see \cite[pp. 127]{takesaki-1}), it follows that $\rho'\in  \M_*$.

% By the decomposition of $\M^*=\M_*\oplus\M_*^{\perp}$, we have $\rho'=\rho'_{no}\oplus \rho'_s$, where $\rho'_s$ decomposes as a normal part $\rho_{no}'$ and a singular part $%\rho'_s$. It follows that the singular part vanishes because $\rho'_s$ is controlled by $\sigma$ which is a normal linear functional. Indeed, suppose $\rho'_s\neq 0$, there exists a %increasing net of projection $e_\al$ with $\sup_{\al}e_\al=e$ such that
 %\[ \rho'_s(e)>\rho'_s(e_\al)\pl.\]
 %On the other hand, by normality of $\sigma$, we have
 %\[ \lim_{\al}\rho'(e- e_\al)\le \lim_{\al}\lambda \sigma(e-e_\al)=0\pl,\]
 %which is contradiction. Hence $\rho'=\rho'_{no}\in \M_*$.

 Furthermore, as the fidelity $F(\rho'_n,\rho_n)$ is upper semi-continuous for $\rho_n'\to \rho'$ and $\rho_n\to \rho$ weakly (see \eqref{eq:fe}), we have $\rho'\sim^\eps \rho$. Therefore,
 \[ D_{\max}^\eps(\rho\|\sigma)\le D_{\max}(\rho'\|\sigma)\le \liminf_{n}D_{\max}^\eps(\rho_n\|\sigma_n)+\delta\]
and taking $\delta\to 0$ yields the conclusion.

For $D_H^\eps$, let us define the quantity $\beta_H^\eps(\rho\|\sigma):=\inf\{ \sigma(Q) \pl | \pl 0\le Q\le 1\pl, \rho(Q)\ge 1-\eps \}$. Then $D_H^\eps(\rho\|\sigma)=-\log \beta_H^\eps(\rho\|\sigma)$. Since the unit ball in $\M$ is compact in the weak operator topology,  there exists an $X\in \M$, $0\leq X\leq 1$ such that $$\beta_H^\eps(\rho\|\sigma)=\sigma(X), \ \text{and} \  \rho(X)\geq 1-\eps.$$
For any $\delta\in (0,1)$  setting $X'=(1-\delta)X+\delta 1$, we have
\[\sigma(X')=(1-\delta)\sigma(X)+\delta \sigma(1)< \sigma (X)+ 2\delta \sigma(1),\]
\[\rho(X')=(1-\delta)\rho(X)+\delta \geq (1-\delta)(1-\eps)+\delta> 1-\eps.\]

Now since $\rho_n$ (resp. $\sigma_n$) converges to $\rho$ (resp. $\sigma$) in the weak*-sense,  we can find an $N\in \mathbb{N}$ such that $\forall n\geq N$ we have
$\sigma_n(X')< \sigma (X)+ 2\delta \sigma(1)$ and $\rho_n(X')> 1-\eps$. Hence by definition

\[\beta_H^\eps(\rho_n\|\sigma_n)\leq \sigma_n(X')\leq \beta_H^\eps(\rho\|\sigma) +2\delta \sigma(1),\]
so that  $\liminf_n D_H^\eps(\rho_n\|\sigma_n)\geq -\log (\beta_H^\eps(\rho\|\sigma) +2\delta \sigma(1))$. Letting ${\delta\searrow 0}$ gives the assertion.

\qd

\section{$D_{\max}$ AEP for states} \label{sec: AEP for states}

In this section we prove the following theorem which generalizes the quantum asymptotic equipartition property of \cite{tomamichel} (see also \cite{tomamichel-book}) in the setting of von Neumann algebras. Recall that the cumulative standard Gaussian distribution function is given by  \[\Phi(a)=\int_{-\infty}^a \frac{1}{\sqrt{2\pi}} e^{\frac{-x^2}{2}} dx,\] and the inverse function $\Phi^{-1}(\eps)=\sup\{a\in \R: \Phi(a)\leq \eps\}$. %For the definition of the quantities $D(\rho \| \sigma), V(\rho\|\sigma), T(\rho\|\sigma)$, see section \ref{subs-ent. in gen vN}.
\begin{theorem}[AEP for states] \label{AEP for states sf}
Let $\cM$ be a von Neumann algebra and $\rho$ be a state on $\cM$ and $\sigma$ be a normal positive linear functional on $\cM$. Assume that $D(\rho \| \sigma) < \infty, V(\rho \| \sigma) < \infty, T(\rho \| \sigma) < \infty$. Then for any $\eps \in (0,1)$ and $n \geq 1$,
\begin{align}
\label{eq:aep-vn}
\frac{1}{n}D_{\max}^{\sqrt{\eps}}(\rho^{\otimes n} \| \sigma^{\otimes n}) =  D( \rho \| \sigma) -\sqrt{\frac{V(\rho\|\sigma)}{n}}\Phi^{-1}(\eps) + O(\frac{\log n}{n}) \ ,
\end{align}
where the $O(\cdot )$ hides constants that only depend on $\rho, \sigma$ and $\eps$.
\end{theorem}
Theorem \ref{AEP for states sf} generalizes the result of Khatri et al. in \cite{wilde-2019} where the authors prove the second order estimate for $D_{\max}^{\eps}(\rho\|\sigma)$ on bounded linear operators on an infinite-dimensional Hilbert space. We obtain the same result in a general von Neumann algebra encapsulating scenarios like continuous variable systems and quantum field theoretic settings, thus applicable in a wider class of quantum systems.

In order to prove this result, we combine the asymptotic equipartition theorem for quantum hypothesis entropy $D_H^\eps$ in von Neumann algebras established by Pautrat and Wang~\cite{ke-Li} together with a generalization of the known relation between  $D_H^\eps$ and $D^{\eps}_{\max}$ established by Anshu \emph{et al} \cite{Anshu} and Khatri \emph{et al}  \cite{wilde-2019} to states on general von Neumann algebras.

%, we will make use of the following result of \cite{ke-Li} about the second order estimates for quantum hypothesis entropy $D_H^\eps$ in general von Neumann algebras.
Let us first state the AEP for hypothesis testing.
 \begin{theorem}[Pautrat-Wang, 2021] \label{Pautrat-Wang}
 Let $\cM$ be a von Neumann algebra and $\rho$ be a state on $\cM$ and $\sigma$ be a normal positive linear functional on $\cM$. Assume that $D(\rho \| \sigma) < \infty, V(\rho \| \sigma) < \infty, T(\rho \| \sigma) < \infty$. Then for any $\eps \in (0,1)$ and $n \geq 1$,
\begin{align}
\label{eq:aep-vn}
\frac{1}{n}D_{H}^{\eps}(\rho^{\otimes n} \| \sigma^{\otimes n}) =  D( \rho \| \sigma) +\sqrt{\frac{V(\rho\|\sigma)}{n}}\Phi^{-1}(\eps) + O(\frac{\log n}{n}) \ ,
\end{align}
where the $O(\cdot)$ hides constants that only depend on $\rho$ and $\sigma$.
 \end{theorem}
\begin{rem}{\rm The article \cite{ke-Li} is devoted to generalizing a key lemma by Ke Li \cite{ke-Li-1} to the framework of general von Neumann algebras. With this key lemma, the above theorem follows as in \cite{datta2016second}.}
\end{rem}
%To prove the AEP for $D^{\eps}_{\max}$ (Theorem \ref{AEP for states sf}) we establish a relation between  $D_{H}^{\eps}(\rho\|\sigma)$ and $D_{\max}^{\eps}(\rho\|\sigma)$ and then use the result of Pautrat-Wang (\ref{Pautrat-Wang}).

To obtain the relation between $D_{H}^{\eps}$ and $D_{\max}^{\eps}$ we follow the same proof as \cite{Anshu}. First, the following expression can be obtained with the same proof as in the finite dimensional case.
%where the authors established the second order AEP for $D_{\max}^{\eps}$ for $B(\Ha)$ of separable $\cH$.
%We need their lemmas (see also \cite{Anshu}) for general von Neumann algebras.

\begin{lemma}\label{alt. max} Let $\rho\in S(\M)$ be a state and $\sigma\in \M_*^+$ be a positive linear functional. For any $\eps\in (0,1)$ and $D_{\max}^{\eps}(\rho\|\sigma)<\infty$,
\[D_{\max}^{\eps}(\rho\|\sigma)=\log \sup_{X\geq 0}\inf_{{\rho'}: d(\rho, {\rho'})\leq \eps}\{ {\rho'}(X): \sigma(X)\leq 1\},\]
where the supremum is over all positive operators $X\in \M^+$ and the infimum is taken over all normal positive linear functionals ${\rho'}$ that lie in the $\eps$-ball of the functional $\rho$ with respect to the purified distance $d(\cdot,\cdot)$.
\end{lemma}
\begin{proof}
%Let us denote the set in the right-hand side as $\widehat{D}_{\max}^{\eps}(\rho\|\sigma)$. The direction $D_{\max}^{\eps}(\rho\|\sigma)\geq \widehat{D}_{\max}^{\eps}(\rho\|\sigma)$ follows immediately. Indeed,  if for $\lambda\geq 0$ we have $\lambda\sigma\geq {\rho'}$, then for any $Q\geq 0$ with $\sigma(Q)\leq 1$, we have
%\[\lambda\geq \lambda \sigma(Q)\geq {\rho'}(Q).\]
%Now using the fact that $\sup\inf f \leq \inf \sup f$, for any bounded real function $f$, we get that
%\[D_{\max}^{\eps}(\rho\|\sigma)\geq \widehat{D}_{\max}^{\eps}(\rho\|\sigma).\]
Following the proof of \cite{Anshu}, it is simple to see, exactly as in the finite dimensional case that
\begin{align}\label{eq:equality}\inf \{\lambda\geq 0 | \rho\leq \lambda \sigma\}=\sup_{X, X\geq 0}\inf\{\lambda | \rho(X)\leq \lambda \sigma(X) \}.\end{align}
%The direction $`\geq'$ is easy. Indeed, if $\lambda\sigma\geq {\rho}$, then for any $X\geq 0$,
%\[ \lambda \sigma(X)\geq {\rho}(X).\]
%For the other direction, we notice that if $$\inf \{\lambda\geq 0 | \rho\leq \lambda \sigma\}=\lambda^*,$$ then take $P$ to be the support projection of the positive linear functional $\lambda^*\sigma-\rho$, then it follows that $(\lambda^*\sigma-\rho)(1-P)=0$. Now letting $X=1-P$, it follows that for any $\lambda>0$ we have
%\[\lambda \sigma (X)-\rho(X)=(\lambda-\lambda^*) \sigma(X)+(\lambda^*\sigma-\rho)(X)=(\lambda-\lambda^*) \sigma(X).\]
%This show that $\inf\{\lambda | \rho(X)\leq \lambda \sigma(X) \}=\lambda^*$ and hence we get \eqref{eq:equality}
%\[\inf \{\lambda\geq 0 | \rho\leq \lambda \sigma\}=\sup_{X, X\geq 0}\inf\{\lambda | \rho(X)\leq \lambda \sigma(X) \}.\]

Normalizing $X$ with $\sigma(X)\leq 1$, we see
\[2^{D_{\max}(\rho\|\sigma)}=\sup\{\rho(X)| X\geq 0, \sigma(X)\leq 1\}.\]
The statement of the lemma now follows from swapping the infimum and the supremum by Sion's minimax theorem, as $\{X \geq 0 : \sigma(X)\leq 1\}$ is a convex set in $\M$ and
 $\{{\rho'}\in  \S_{\leq}(\M)\ | \ d({\rho'}, \rho)\leq \eps\} $ is a compact and convex in $\M_*$ (This is seen from the joint upper semi-continuity and
concavity of the fidelity function).
\end{proof}

\begin{lemma} \label{max and hyp rel}
Let $\rho\in S(\M)$ be a state and $\sigma\in (\M_*)_+$ be a positive linear functional. We have for any $\eps, \delta\in (0,1)$ with $\eps+\delta<1$, and $D_{\max}^{\sqrt{\eps}}(\rho\|\sigma)<\infty,$
\begin{align}\label{eq:dhdmax}&D_{\max}^{\sqrt{\eps}}(\rho\|\sigma)\leq D_{H}^{1-\eps}(\rho\|\sigma) +\log (\frac{1}{1-\eps})\pl, \\ &D_{H}^{1-\eps-\delta}(\rho\|\sigma)\leq  D_{\max}^{\sqrt{\eps}}(\rho\|\sigma)+\log(\frac{4(1-\epsilon)}{\delta^2})\pl. \label{eq:dhdmax2}
\end{align}
\end{lemma}
\begin{proof}
We first prove~\eqref{eq:dhdmax} for a finite von Neumann algebra $\M$, which proceeds similarly as \cite{Anshu} and \cite{wilde-2019}.
Note that since the von Neumann algebra is tracial, we can think of $\rho, \sigma$ as elements in $L_1(\M)$ rather than functionals in $\M_*$.
By Lemma \ref{alt. max} if we fix an element $X\geq 0$ in $\M$ such that $\sigma(X)=\tau(X\sigma)\leq 1$, it suffices to find a $\widetilde{\rho}$ with $d(\rho, \widetilde{\rho})\leq \sqrt{\eps}$ such that
\[\widetilde{\rho}(X)\leq 2^{D_{H}^{1-\eps}(\rho\|\sigma) }/ (1-\eps).\]
To this end, we consider the commutative subalgebra $\A$ generated by $X$. Since $(\M, \tau)$ is a finite von Neumann algebra, there exists a trace-preserving conditional expectation
$\E:\M\rightarrow \A$ onto $\A$ given by
\[\tau(xy)=\tau(x\E(y))\pl, \pl  \forall\pl x\in \A\pl, \pl y\in \M\pl.\] This is a CPTP map satisfying the bimodule property
\[\E(axb)=a\E(x)b, \ \text{for} \ a, b\in\A,x\in \M.\]
By the data processing inequality for the entropy $D_H^\eps(\cdot\|\cdot)$,
\begin{equation}\label{eq-inf}
D_{H}^{1-\eps}(\rho\|\sigma)\geq D_{H}^{1-\eps}(\E(\rho)\|\E(\sigma))\geq D_{s}^{1-\eps}(\E(\rho)\|\E(\sigma))=:K,
\end{equation}
where the last quantity is called $\eps$-information spectrum relative entropy for a density $\gamma,$ and a  positive element $\omega$ as defined in section \ref{entropies and channels} as
\[D_{s}^{\eps}(\gamma\|\omega)=\sup \{\lambda\in \R | \tau(\gamma \{\gamma\leq 2^\lambda \omega \})\leq \eps\},\]
where $ \{ \gamma\leq 2^\lambda \omega\}$ denotes the projection onto the positive part of  $2^\lambda \omega -\gamma$. Note that the last inequality of equation \eqref{eq-inf} follows from choosing $Q=1- \{\gamma\leq 2^\lambda \omega \}$ in the definition of $D_{H}^{1-\eps}(\gamma\|\omega)$, as was shown in the Lemma 12 in \cite{toma-haya}.

Now for any $\eta>0$, we define $P$ to be the projection on the positive part of $(2^{K+\eta} \E(\sigma)-\E(\rho))$ and $Q=1-P$. From the definition of $D_{s}^{1-\eps}(\E(\rho)\|\E(\sigma))$, it follows that $\tau(P\E(\rho))>1-\eps$.  Note that $P,Q\in \A$ and hence
\[\tau(Q\rho)=\tau(\E(Q)\rho)=\tau(Q\E(\rho))=1-\tau(P\E(\rho))\leq \eps.\]
Now defining $\widetilde{\rho}=\frac{P\rho P}{\tau(P\rho)}$, one can verify that
\[d(\rho, \widetilde{\rho})\leq \sqrt{\eps}.\]
Now we obtain
\begin{align*}\tau(X\widetilde{\rho})&=\frac{\tau(X P\rho P)}{1-\tau(P\rho)}
=\frac{\tau(\E(X P\rho P))}{1-\tau(P\rho)}
=\frac{\tau(X P\E(\rho) P)}{1-\tau(P\rho)}
\leq \frac{\tau(2^{K+\eta}X P\E(\sigma) P)}{1-\eps}\\
&= \frac{2^{K+\eta}}{1-\eps}\tau(X \E(P\sigma P) )
= \frac{2^{K+\eta}}{1-\eps}\tau(X P\sigma P )
\leq \frac{2^{K+\eta}}{1-\eps}\tau(X \sigma )
\leq \frac{2^{D_{H}^{1-\eps}(\rho\|\sigma)+\eta} }{1-\eps}.
\end{align*}
Here we have repeatedly used the fact that $X, P\in \A$ and hence \[\E(XP\rho P)=XP\E(\rho)P\pl .
%=\E(X)\E(P\rho P)=X\E(P)\E(\rho)\E(P)
\]
After taking logarithm and infimum over all $\widetilde{\rho}$, the above calculations yields
\[D_{\max}^{\sqrt{\eps}}(\rho\|\sigma)\leq D_{H}^{1-\eps}(\rho\|\sigma) +\eta+ \log_2 (\frac{1}{1-\eps}),\]
As this holds for every $\eta>0$, the desired assertion follows with the assumption that $\M$ is finite.
The case for general von Neumann algebras follows from the Haagerup reduction method \ref{Haagerup reduction}.

The second inequality of the lemma follows directly from the argument given in \cite{Anshu}. Indeed, choose the optimal $\widetilde{\rho}$ for $D_{\max}^{\sqrt{\eps}}(\rho\|\sigma)$ and optimal $0\leq X\leq 1$ for $D_{H}^{1-\eps-\delta}(\rho\|\sigma)$. Then by data processing inequality for fidelity and applied to the measurement $\{X, 1-X\}$, it holds that
\[\log (\sqrt{1-\eps}-\sqrt{1-\eps-\delta})^2\leq D_{\max}^{\sqrt{\eps}}(\rho\|\sigma)-D_{H}^{1-\eps-\delta}(\rho\|\sigma).\]
Furthermore, using the relation $\sqrt{1-\eps}-\sqrt{1-\eps-\delta}\geq \frac{\delta}{2\sqrt{1-\eps}}$.
\end{proof}

%\noindent \textbf{Proof of Theorem \ref{AEP for states sf}:}
\begin{proof}[Proof of Theorem \ref{AEP for states sf}:]
Using Lemma \ref{max and hyp rel} and Theorem \ref{Pautrat-Wang} the proof follows easily. Indeed,
exploiting the first inequality in the previous lemma we get
\[\frac{1}{n}D_{\max}^{\sqrt{\eps}}(\rho^{\otimes n}\|\sigma^{\otimes n})\leq \frac{1}{n}D_{H}^{1-\eps}(\rho^{\otimes n}\|\sigma^{\otimes n})+\frac{1}{n}\log_2(1/1-\eps).\]
Now use Theorem \ref{Pautrat-Wang} to obtain
\[\frac{1}{n}D_{\max}^{\sqrt{\eps}}(\rho^{\otimes n}\|\sigma^{\otimes n})\leq D( \rho \| \sigma) -\sqrt{\frac{V(\rho\|\sigma)}{n}}\Phi^{-1}(\eps) + O(\frac{\log n}{n}).\]
Here we use the fact $\Phi^{-1}(1-\eps)=-\Phi^{-1}(\eps)$.

For the other direction we use the second inequality in Lemma \ref{max and hyp rel}
\[\frac{1}{n}D_{\max}^{\sqrt{\eps}}(\rho^{\otimes n}\|\sigma^{\otimes n})\geq \frac{1}{n}D_{H}^{1-\eps-\delta}(\rho^{\otimes n}\|\sigma^{\otimes n})- \frac{1}{n}\log_2(\frac{4(1-\epsilon)}{\delta^2}).\]
Then again using Theorem \ref{Pautrat-Wang} we obtain

\[\frac{1}{n}D_{\max}^{\sqrt{\eps}}(\rho^{\otimes n}\|\sigma^{\otimes n})\geq D( \rho \| \sigma) +\sqrt{\frac{V(\rho\|\sigma)}{n}}\Phi^{-1}(1-\eps-\delta) + O(\frac{\log n}{n}) - \frac{1}{n}\log_2(\frac{4(1-\epsilon)}{\delta^2}).\]
Now note that the function $\Phi^{-1}$ is continuously differentiable and choosing $\delta=1/\sqrt{n}$, and estimating by Taylor's theorem we obtain
\[\frac{1}{n}D_{\max}^{\sqrt{\eps}}(\rho^{\otimes n}\|\sigma^{\otimes n})\geq D( \rho \| \sigma) -\sqrt{\frac{V(\rho\|\sigma)}{n}}\Phi^{-1}(\eps) + O(\frac{\log n}{n}).\]
\end{proof}

We end this section by providing the following proposition which asserts that the finiteness of $D(\rho\|\sigma), V(\rho\|\sigma)$ and $T(\rho\|\sigma)$ once we assume that $D_{\max} (\rho\|\sigma)<\infty$.

\begin{prop} \label{lemma:appl}
Let $D_{\max} (\rho\|\sigma)<\infty$ for $\rho\in \S(\M)$ and $\sigma\in \M_*^+$. Then it follows that $D(\rho\|\sigma), V(\rho\|\sigma)$ and $T(\rho\|\sigma)<\infty$. Moreover,
\[V(\rho\|\sigma)\leq (\log (3+   D_{\max}(\rho\|\sigma)^{1/2} )^2, T(\rho||\sigma) \le \left(\log  (D_{\max}(\rho||\sigma)+2D_{\max}(\rho||\sigma)^{1/2}+6)\right)^3\]

%and moreover it holds that
%\[  V(\rho\|\sigma)\leq [D_{\max} (\rho\|\sigma)]^2.\]

\end{prop}
\begin{proof}
We first prove the finiteness of $D(\rho\|\sigma)$.
Recall from subsection \ref{subs-ent. in gen vN} the definition of the relative entropy, the variance and the quantity $T$ are defined as follows
\[D(\rho\|\sigma)=\lan \xi_\rho| \log\Delta(\rho,\sigma) | \xi_\rho \ran \pl , V(\rho\|\sigma)=\lan \xi_\rho| (\log\Delta(\rho,\sigma))^2 | \xi_\rho \ran \pl , T(\rho\|\sigma)=\lan \xi_\rho| (\log\Delta(\rho,\sigma))^3 | \xi_\rho \ran\]
where $\Delta(\rho,\sigma)$ is the relative modular operator and $\xi_\rho$ is a vector implementing $\rho$ via  $\rho(x)=\bra{\xi_\rho}x\ket{\xi_\rho}, \forall x\in \M$ (for instance, there is a unique one in the positive cone of $L_2(\cM)$).

Now let $\rho\leq \lambda \sigma$ for some $\lambda>0$. Then it follows that (see \cite{Araki}) $$D(\rho\|\sigma)\leq D(\rho\|\lambda^{-1}\rho)=\log \lambda.$$

%We have $\Delta(\rho,\sigma)\leq \lambda\Delta(\rho,\rho) $. Indeed this can be seen from the linearity of the relative modular operator for the first input (see \cite{hiai-book})
%\[ \Delta(\rho,\sigma)+\Delta(\omega,\sigma)=\Delta(\rho+\omega,\sigma)\pl,  \Delta(\lambda\rho,\sigma)=\lambda\Delta(\rho,\sigma)\]
%and the inverse relation on the standard form $L_2(\cM)$,
%\[ \Delta(\rho,\sigma)^{-1}=J\Delta(\sigma,\rho)J\pl.\]
%If there is a $\lambda>0$ such that $\rho\leq \lambda \sigma$, then $\omega=\lambda \sigma-\rho\in \cM_*^+ $, we have
%\[ \Delta(\rho,\rho)\le \Delta(\lambda\sigma,\rho)=\lambda \Delta(\sigma,\rho)  \Rightarrow   \Delta(\rho,\sigma)\le  \lambda\Delta(\rho,\rho)\pl.\]
%By operator monotonicity of $x\to \log x$, we have
%\[ D(\rho\|\sigma)=\lan \xi_\rho| \log\Delta(\rho,\sigma) | \xi_\rho \ran\le \lan \xi_\rho| \log\lambda| \xi_\rho \ran+\lan \lan\xi_\rho|\log\Delta(\rho,\rho) | \xi_\rho\ran=\log \lambda\]
\noindent Hence $D(\rho\|\sigma)\leq D_{\max}(\rho\|\sigma)<\infty$.

The finiteness of $V(\rho\|\sigma)$ follows from the argument of Lemma III.4 in \cite{DF}. Here we argue for $T(\rho\|\sigma)$.
Observe that for $v\in (0,1)$ and $t>0$,
\begin{align*}
\ln^3 t =\frac{1}{8v^3}\ln^3 t^{2v}\le \frac{1}{8v^3}\left(\ln (t^v+t^{-v}+1)^2\right)^3.
\end{align*}
Note that $x\mapsto \ln^3 x$ is concave on the interval $[e^2,\infty]$, and the spectrum of $(X^v+X^{-v}+1)^2$ of any positive operator $X$ is in $[9,\infty )$. Take the short notation $X:=\Delta(\rho/\sigma)$ and $\bra{\xi}=\bra{\xi_\rho} $ and note that $T(\rho||\sigma)=\frac{1}{(\ln 2)^3} \bra{\xi} (\ln X)^3  \ket{\xi}$ . We have by Jensen inequality that
\begin{align*}
\bra{\xi}\ln^3 \Delta \ket{\xi}\le \frac{1}{8v^3}\bra{\xi}\left(\ln (X^v+X^{-v}+1)^2\right)^3  \ket{\xi}
\le \frac{1}{8v^3}\left(\ln \bra{\xi}(X^v+X^{-v}+1)^2\ket{\xi} \right)^3 \pl.
\end{align*}
For simplicity, take $v=\frac{1}{2}$
\begin{align*} \bra{\xi}(X^{1/2}+X^{-1/2}+1)^2\ket{\xi}= \bra{\xi}X+X^{-1}+3+2X^{\frac{1}{2}}+2X^{-\frac{1}{2}}\ket{\xi}\le \lambda+2\lambda^{1/2}+6
\end{align*}
where for each term we use
\begin{align*}\bra{\xi}X\ket{\xi}=&\bra{\xi_\rho} \Delta(\rho,\sigma)\ket{\xi_\rho}\le  \bra{\xi_\rho} \lambda\Delta(\rho,\rho)\ket{\xi_\rho}
=\lambda \\
\bra{\xi}X^{-1}\ket{\xi}= & \bra{\xi_\rho} \Delta(\sigma,\rho)\ket{\xi_\rho}=\bra{\xi_{\sigma}}\xi_{\sigma}\rangle=1\\
\bra{\xi}X^{1/2}\ket{\xi}\le & \bra{\xi}X\ket{\xi}^{1/2}\le  \lambda^{1/2} \\
\bra{\xi}X^{-1/2}\ket{\xi}\le  & \bra{\xi}X^{-1}\ket{\xi}^{1/2}=1
\end{align*}
This proves that
\[ T(\rho||\sigma) \le [\log  (D_{\max}(\rho||\sigma)+2D_{\max}(\rho||\sigma)^{1/2}+6)]^3\]
and the argument for $V(\rho||\sigma)$ is similar.
\end{proof}

\section{$D_{\max}$ AEP for channels}\label{AEP for channels}
In this section we prove an Asymptotic Equipartition Property for quantum channels. It comes with a cost that we lose the second order expansion as it is unclear even in finite dimensions. We say $\Psi:\M_*\to \N_*$ is a \textbf{replacer channel} if
$$\Psi(a)= a(1) \sigma,$$
%the adjoint map $\Psi_*:\N\rightarrow \M$ is given by
%\[\Psi_*(x)=\sigma(x)1_\M.\]
%for some state $\sigma$ acting on $\N$. If $\M$ is semi-finite with trace $\tau$, then
%$$\Psi(a)=\tau(a) \sigma,$$
where $\sigma \in \N_*$.
\subsection{Channel Divergences:}
For two quantum channels $\Phi,\Psi : \M_*\to \N_*$, we define the $\al$-sandwiched R\'enyi divergence as
\[{D}_{\al}(\Phi\|\Psi)=\sup_{\rho\in \S(\M)}{D}_{\al}(\Phi(\rho)\|\Psi(\rho)).\]
Note that this quantity is $+\infty$ if there exists $\rho\in \S(\M)$ such that $s(\Phi(\rho))\nleq s(\Psi(\rho))$.

In this section we will only concern channel divergence $D_\al$ with $\al> 1$. The regularized channel divergence is defined as
\[{D}_{\al}^{\text{reg}}(\Phi\|\Psi):=\lim_{n\to\infty} \frac{1}{n}{D}_{\al}(\Phi^{\otimes n} \|\Psi^{\otimes n}). \]
Although in infinite dimensions, the above quantities can often be infinite, there are still densely many pairs $(\Phi,\Psi)$ such that ${D}_{\al}(\Phi\|\Psi)$ and ${D}_{\al}^{\text{reg}}(\Phi\|\Psi)$ are finite. Indeed, if for a constant $M>0$, we have $\Phi\le_{CP} M\Psi$ as completely positive maps, i.e. $M\Psi- \Phi$ is completely positive, then ${D}_{\al}(\Phi\|\Psi)\le {D}_{\al}^{\text{reg}}(\Phi\|\Psi)\le \log M$. Note that the channel divergence is super-additive, ${D}_{\al}(\Phi_1\ten \Phi_2\|\Psi_1\ten \Psi_2)\ge {D}_{\al}(\Phi_1\|\Psi_1)+{D}_{\al}(\Phi_2\|\Psi_2)$, by choosing tensor product states. The above limit $\lim_{n}$ can be replaced by $\sup_n$, hence always exists (allowing value as $+\infty$).

Following \cite{CMW}, we also introduce the channel divergence with reference system
\begin{align}\mathbb{D}^{\text{ref}}(\Phi\|\Psi)=\sup_{\rho\in \S(\cR\overline{\ten} \M)}\mathbb{D}(\id_\cR\otimes \Phi(\rho)\|\id_\cR\otimes \Psi(\rho)),\label{eq:R}\end{align}
 where the supremum is taken over all the bipartite states $\rho$ acting on the von Neumann algebra $\cR\overline{\ten}\M$, and $\cR$ is an arbitrary von Neumann algebra as a reference system. By Haagerup reduction, it suffices to consider $\cR$ being finite. We will use this definition for $\mathbb{D}=D,D_\al, D_{\max},$ and $D_{\max}^\eps$.
% It is immediate the channel divergence without any reference system. Using the chain rule \ref{thm:chain rule} for this definition we immediately get for bipartite $\rho,\sigma\in  \S(\cR\M)$
 %\[{D}_{\al}(id_\cR\otimes \Phi(\rho)\|id_\cR\otimes \Psi(\si))\leq {D}_{\al}(\rho\|\sigma)+{D}_{\al}^{\text{ref, reg}}(\Phi\|\Psi),\]
%where  ${D}_{\al}^{\text{ref, \ reg}}(\Phi\|\Psi)$ is defined as
%\[{D}_{\al}^{\text{ref, reg}}(\Phi\|\Psi)=\sup_n \frac{1}{n} \TD(\Phi)^{\otimes n}\| \Psi^{\otimes n}).\]

\begin{theorem} [$D_{\max}$ AEP for channels]
\label{thm:aep-channel-sec}
Let $\Psi_{\sigma}:\cM_*\to \cN_*$ be a replacer channel with output state $\sigma\in \S(\N)$. Suppose $\Phi:\cM_*\to \cN_*$ is a quantum channel such that ${D}^{\textup{ref}}_{\al}(\Phi\|\Psi_\sigma)< \infty$ for some $\alpha>1$. Then for any $\eps \in (0,1)$,
\begin{align}
\label{eq:aep-channel}
D_{\max}^{\eps, \textup{ref}}(\Phi^{\otimes n} \| \Psi_{\sigma}^{\otimes n}) = n D^{\textup{ref}}( \Phi \| \Psi_{
\sigma}) + O(\sqrt{n}) \ ,
\end{align}
where the $O(\cdot)$ hides constants that only depend on $\Psi$ and $\sigma$.
\end{theorem}
In the sequel we will often simply write $\Psi$ for $\Psi_\sigma$ as a replacer channel with no confusion. We first need the following additivity result
\begin{theorem}\label{rmv-reg}
Let $\Phi$ be an arbitrary channel and $\Psi$ be a replacer map. Then we have for $n\geq 1$
\[D^{\textup{ref}}_\al(\Phi^{\otimes n}\|\Psi^{\otimes n})=n D^{\textup{ref}}_\al(\Phi\|\Psi).\]
\end{theorem}
For the proof, following \cite{amortized}, we need to intorduce the amortized channel divergence for $\Phi, \Psi$ as follows
\[D^A_{\al}(\Phi\|\Psi)=\sup_{\rho,\si\in \S(\cR\M)} D_\al(\id_\cR\otimes \Phi(\rho)\|\id_\cR\otimes \Psi(\sigma))-D_\al(\rho\|\sigma).\]
In \cite{amortized}, Proposition 41, in the finite-dimensional case, the authors showed that for a channel $\Phi$ and the replacer channel $\Psi$, we have $D^A_{\al}(\Phi\|\Psi)=D^{\textup{ref}}_\al(\Phi\|\Psi)$. Furthermore, Theorem 5.4 in \cite{fawzi2021defining} asserts that for any two completely positive maps $\Phi, \Psi$ we have $D^A_{\al}(\Phi\|\Psi)=\lim_{n\to \infty}\frac{1}{n}D^{\textup{ref}}_\al(\Phi^{\otimes n}\|\Psi^{\otimes n})$. It follows from these two results that the assertion of the  theorem \ref{rmv-reg} in finite dimensions. We will extend this in our general setting of von Neumann algebras.

The proof of Theorem \ref{rmv-reg} uses the theory of non-commutative vector-valued $L_p$ spaces which we briefly recall here (see \cite{pisier} for more detail).  Let $1\leq p,q\leq \infty$ and $\frac{1}{r}=|\frac{1}{p}-\frac{1}{q}|$.  Recall that for two finite von Neumann algebras $\M, \N$, the space $L_p(\M, L_q(\N))$ is the norm completion of $\M\overline{\otimes}\N$ with respect to the following norm: for $p\le q$,
\[\|X\|_{L_p(\M, L_q(\N))}:=\inf_{X=(a\otimes 1)Y (b\otimes 1)} \|a\|_{L_{2r}(\M)} \|Y\|_{L_q(\M\bar{\otimes}\N)} \|b\|_{L_{2r}(\M)},\]
where the infimum is taken over all factorization $X=(a\otimes 1)Y (b\otimes 1)$, for $a,b\in \M$ and $Y\in \M\overline{\otimes} \N$; for $p\ge  q$
\[\|X\|_{L_p(\M, L_q(\N))}:=\sup_{\|a\|_{L_{2r}(\M)}=\|b\|_{L_{2r}(\M)}=1} \|(a\otimes 1)X (b\otimes 1)\|_{L_q(\M\bar{\otimes}\N)},\]
where the supremum is over all $a,b\in \M$ with unit $L_{2r}$-norm. For $X$ positive, it is sufficient to consider $a=b>0$ in above optimization. In particular, when $q=1$, noting that $\frac{1}{r}+\frac{1}{p}=1$ ($r$ is the conjugate of $p$ and denoted as $p'$) the $L_p(\N,L_1(\M))$-norm for positive $X$ is given by
\begin{align*} \norm{X}{L_p(\M,L_1(\N))}=\sup_{\norm{a}{L_{2p'}(\M)}=1} \norm{ (a\ten 1)X(a^*\ten 1)}{1} =&\sup_{\norm{a}{L_{2p'}(\M)}=1}\tau_\M\ten \tau_\N((a^*a\ten 1)X)\\ =&\sup_{\norm{a}{L_{2p'}(\M)}=1}\tau_\M(a^*aX_\M )=\norm{X_\M}{L_p(\M)}\pl,\end{align*}
where $X_\M=\id \ten \tau_\N( X)$ is the reduced operator of $X$ on $\M$.

Given a linear map $\Phi: L_1(\M)\rightarrow L_p(\N)$ with $p\ge 1$, we have by
Pisier's lemma \cite[Lemma 1.7]{pisier} that the completely bounded norm can be equivalently given by
\begin{align} \norm{\Phi:  L_1(\M)\rightarrow L_p(\N)}{cb}=&\sup_{n} \norm{ \id_{M_n}\ten \Phi:  L_1(M_n, L_1(\M))\rightarrow L_1(M_n,L_p(\N))}.{} \nonumber
\end{align}
The next proposition shows that the supremum over matrix algebras $M_n$ can be further enlarged to over all finite von Neumann algebra $\cR$. This justifies the fact that in the definition of channel divergence \eqref{eq:R}, we can allow the reference system to be any von Neumann algebra.
\begin{prop}
Let $p\in [1,\infty]$ and $\Phi: L_1(\M)\rightarrow L_p(\N)$ be a linear map. Then
\begin{align}
\norm{\Phi:  L_1(\M)\rightarrow L_p(\N)}{cb}=\sup_{\cR} \norm{ \id_{\cR}\ten \Phi:  L_1(\cR, L_1(\M))\rightarrow L_1(\cR,L_p(\N))}{}  \label{eq:cb}\end{align}
where the supremum is over any finite von Neumann algebra.
\end{prop}
\begin{proof}It suffices to show that for any $\cR$,
\[\norm{ \id_{\cR}\ten \Phi:  L_1(\cR, L_1(\M))\rightarrow L_1(\cR,L_p(\N))}{} \le \norm{\Phi:  L_1(\M)\rightarrow L_p(\N)}{cb}.\]
Recall that
\[ L_1(\cR, L_1(\M))\cong L_1(\cR\overline{\ten}\M)\cong L_1(\cR)\widehat{\ten}L_1(\M)\]
where $\widehat{\ten}$ is the operator projective tensor product. Then for any $\cR$,
\[ \norm{ \id_{\cR}\ten \Phi:  L_1(\cR, L_1(\M))\rightarrow L_1(\cR)\widehat{\ten}L_p(\N)}{} \le  \norm{\Phi:  L_1(\M)\rightarrow L_p(\N)}{cb}\pl.\]
It remains to show that for any $\cR$ and $p\in[1,\infty]$,
\[\id: L_1(\cR)\widehat{\ten}L_p(\N)\to L_1(\cR,L_p(\N))\pl \] is contraction.
We argue the case for $p<\infty$ and $p=\infty$ is similar. Let $\frac{1}{p}+\frac{1}{p'}=1$. By duality, we know
$L_1(\cR)\widehat{\ten}L_p(\N)^*=CB(L_1(\cR), L_{p'}(\N^{op}))$ where $\N^{op}$ is the opposite algebra of $\N$ \cite[Theorem 4.1]{pisier2003introduction}, and $L_\infty(\cR,L_{p'}(\N^{op}))$ is a weak$^*$ dense subspace of the dual $L_1(\cR,L_p(\N))^*$ (\cite[Proposition 4.5]{junge2010mixed}). Then it is equivalent to show
\[ \id^*: L_\infty(\cR,L_{p'}(\N^{op}))\to CB(L_1(\cR), L_{p'}(\N^{op}))\pl, \id^*(x)=T_x\]
is a contraction, where for $x\in \cR\overline{\ten} \N^{op}$,  $T_x$ is a linear map from $L_1(\cR)$ to $L_{p'}(\N^{op})$
\[T_x(y)=\tau_{\cR}(x(y\ten 1))\pl, \pl y\in L_1(\cR)\pl.\]
The argument to show $$\norm{T_x:L_1(\cR)\to L_{p'}(\N^{op})}{cb}\le  \norm{x}{L_\infty(\cR,L_{p'}(\N^{op}))}$$
is similar to \cite[lemma 3.4]{ganesan2020quantum}. Indeed, we first note that
$\operatorname{id}_{M_n}\ten T_x= T_{\phi\ten x}$ where $\phi=\sum_{i,j} e_{ij}\ten e_{ij}\in M_n\ten M_n$ is the Choi matrix for $\operatorname{id}:M_n\to M_n$. Given $$\norm{a}{S_2^n(L_2(\cR))}=\norm{b}{S_2^n(L_2(\cR))}=1,$$we can write
$a=\sum_{k}\mu_k\omega_k\ten a_k$ such that $\mu_k$ (resp. $a_k$) orthogonal in $S_2^n$ (resp. $L_2(\R)$) and $\norm{a_k}{2}=1, \sum_{k}\norm{\omega_k}{2}^2=1$ and similarly for $b=\sum_{l}\nu_l\si_l\ten b_l$. Then
\begin{align*}\operatorname{id}_{M_n}\ten T_x(ab)=T_{\phi\ten x}(ab)&=\sum_{k,l}\Big(\tr\ten \operatorname{id}_{M_n}((\omega_k\si_l\ten 1)\phi)\Big)\ten \Big(\tau_\cR\ten \id_\N ((a_kb_l\ten 1) x)\Big)\\
&=\sum_{k,l}\Big(\tr\ten \operatorname{id}_{M_n}(\omega_k\cdot \phi \cdot \si_l)\Big)\ten \Big(\tau_\cR\ten \id_\N (a_k\cdot x \cdot b_l)\Big)\\
&=\tr\ten id_{M_n}\ten\tau_\cR\ten \id_\N \Big(\sum_{k,l} (\omega_k\cdot \phi \cdot \si_l)\ten (a_k\cdot x \cdot b_l)\Big)\pl,
\end{align*}
where we use the short notation $a\cdot x \cdot b:=(a\ten 1)x(b\ten 1)$.
Using bracket notation,
\[\phi=\ket{h}\bra{h}\pl, \ket{h}=\sum_{i=1}^n\ket{i}\ket{i}\]
where $\{\ket{i}\}$ is the standard basis in $l_2^n$. We have
\[\norm{\omega_k\cdot \phi \cdot \si_l}{1}=\norm{\omega_k\ten 1 \ket{h}}{l_2}\norm{\si_l^*\ten 1 \ket{h}}{l_2}
=\norm{\omega_k}{2}\norm{\si_l}{2}
\pl.\]
Here $\norm{\cdot}{l_2}$ is the vector norm and \[\norm{\omega_k\ten 1 \ket{h}}{l_2}^2=\bra{h}\omega_k^*\omega_k\ten 1 \ket{h}=\tr(\omega_k^*\omega_k)=\norm{\omega_k}{2}\pl.\]
Therefore,
\begin{align*}&\norm{\sum_{k,l} (\omega_k\cdot \phi \cdot \si_l)\ten (a_k\cdot x \cdot b_l)}{L_1(M_{n^2}(\cR),L_{p'}(\N^{op}))} \\ \le &\sum_{k,l}\norm{\omega_k\cdot \phi \cdot \si_l}{S_1}\norm{a_k\cdot x \cdot b_l}{L_1(\cR,L_{p'}(\N^{op}))}
\\ \le &\sum_{k,l} \norm{\omega_k}{2}\norm{\si_l}{2} \norm{a_k\cdot x \cdot b_l}{L_1(\cR,L_{p'}(\N^{op}))}.
\end{align*}
Note that $\norm{a_k}{2}=\norm{b_l}{2}=1$. Assume the polar decomposition $a_k=u|a_k|$ and $b_l=|b_l|v$. We have by the definition of the norms of $L_1(\cR,L_{p'}(\N^{op})$ and $L_\infty(\cR,L_{p'}(\N^{op})$
\begin{align*}\norm{a_k\cdot x \cdot b_l}{L_1(\cR,L_{p'}(\N^{op})}\le & \norm{u|a_k|^{\frac{1}{p}}}{2p}\norm{|a_k|^{\frac{1}{p'}}\cdot x \cdot |b_l|^{\frac{1}{p'}}}{L_p'(\cR\overline{\ten}\N^{op})}\norm{|b_l|^{\frac{1}{p}}v}{2p}
\\ \le & \norm{ x }{L_\infty(\cR,L_{p'}(\N^{op}))}\end{align*}
for any $l,k$. Therefore,
\begin{align*}
\norm{\operatorname{id}_{M_n}\ten T_x(ab)}{S_1^n(L_p'(\cN^{op}))}\le& \norm{\sum_{k,l} (\omega_k\cdot \phi \cdot \si_l)\ten (a_k\cdot x \cdot b_l)}{L_1(M_{n^2}(\cR), L_p'(\cN^{op}))}\\
\le &\sum_{k,l} \norm{\omega_k}{2}\norm{\si_l}{2} \norm{a_k\cdot x \cdot b_l}{L_1(\cR,L_{p'}(\N^{op})}
\\
\le &(\sum_{k,l} \norm{\omega_k}{2}\norm{\si_l}{2}) \norm{ x }{L_\infty(\cR,L_{p'}(\N^{op}))}\le \norm{ x }{L_\infty(\cR,L_{p'}(\N^{op}))}.
\end{align*}
Then by Pisier's lemma \cite[Lemma 1.7]{pisier},
\begin{align*} \norm{T_x:L_1(\cR)\to L_{p'}(\N^{op})}{cb}=&\sup_{n}\norm{\id_{M_n}\ten T_x:S_1^n(L_1(\cR))\to S_1^n(L_{p'}(\N^{op}))}{}\\ \le  &\norm{ x }{L_\infty(\cR,L_{p'}(\N^{op}))}\pl.\end{align*}
This verifies $\id^*: L_\infty(\cR,L_{p'}(\N^{op}))\to CB(L_1(\cR), L_{p'}(\N^{op}))$ is a contraction and finishes the proof.
\end{proof}
% If $\Phi$ is completely positive, it is also proved in \cite[Theorem 12]{junge-mult} that the above norm are attained at positive element. (This assertion was proved for finite-dimensional matrix algebras, but the Cauchy-Schwarz type argument applies to all semi-finite von Neumann algebras.)

\begin{prop}
Let $\M$ be a von Neumann algebra with a normal faithful state $\phi$ and let $\Psi:\M_*\to \N_*$ be a replacer map associated with a normal state $\sigma$ on $\N$. Then for any channel $\Phi:\M_*\rightarrow \N_*$,
\[D^A_{\al}(\Phi\|\Psi)\leq \TD(\Phi\|\Psi).\]
\end{prop}
\begin{proof}
We first follow the idea of \cite[Proposition 41]{amortized} with necessary adaptations to prove the case where $\M$ is a finite von Neumann algebra. Suppose $\alpha'$ is such that  $1=\frac{1}{\al}+\frac{1}{\al'}$. Given states $\rho_{\cR\cM},\omega_{\cR\cM}\in S(\mathcal{R}\overline{\ten}\cM )$, where $\cR$ is a finite von Neumann algebra, we write $\rho_{\cR},\omega_{\cR} $ for the corresponding reduced operators on $\cR$. We have
\begin{align*}
&D_\al(\id_\cR\otimes \Phi(\rho_{\cR\M})\|\id_\cR\otimes \Psi(\omega_{\cR\M}))-D_\al(\rho_{\cR\M}\|\omega_{\cR\M})\\
=&D_\al(\id_\cR\otimes \Phi(\rho_{\cR\M})\|\omega_{\cR}\otimes\sigma_\N)-D_\al(\rho_{\cR\M}\|\omega_{\cR\M})\\
=&\al' \log \left\Vert (\omega_{\cR}\otimes\sigma_\N)^{-\frac{1}{2\al'}}\id_\cR\otimes \Phi(\rho_{\cR\M}) (\omega_{\cR}\otimes\sigma_\N)^{-\frac{1}{2\al'}}\right\Vert_\al-\al' \log \left\Vert \omega_{\cR\M}^{-\frac{1}{2\al'}}\rho_{\cR\M} \omega_{\cR\M}^{-\frac{1}{2\al'}}\right\Vert_\al\\
=&\al' \log \left\|\Theta_{\sigma_\N^{-\frac{1}{\al'}}} \circ \Phi(X_{\cR\M})\right\|_\al
\end{align*}
where $\Theta_\omega(\cdot)=\omega^{1/2}(\cdot)\omega^{1/2}$ and
\[X_{\cR\M}=\frac{\omega_\cR^{-\frac{1}{2\al'}} \rho_{\cR\M} \omega_\cR^{-\frac{1}{2\al'}}}{\|\omega_\cR^{-\frac{1}{2\al'}} \rho_{\cR\M} \omega_\cR^{-\frac{1}{2\al'}}\|_\al}.\]
Taking the supremum over $\rho, \omega$, we have
\begin{align*}
D^A_{\al}(\Phi\|\Psi) =     \sup_{X_{\cR\M}\geq 0, \|X_\cR\|_\al\leq 1} \frac{\al}{\al-1} \log \left\Vert\Theta_{\sigma_\N^{-\frac{1}{\al'}}} \circ \Phi(X_{\cR\M})\right\Vert_\al\pl,
\end{align*}
Note that given $X$ positive, $\|X_\cR\|_\al\leq1$ is equivalent to $\|X\|_{L_\al(\cR, L_1(\M))}\leq 1$.
Now using \cite[Theorem 12]{junge-mult}
\[ \sup_{X_{\cR\M}\geq 0, \|X_\cR\|_\al\leq 1} \al' \log \| \Theta_{\sigma_\N^{-\frac{1}{\al'}}} \circ \Phi(X_{\cR\M})\|_\al=  \al' \log\| \Theta_{\sigma_\N^{-\frac{1}{\al'}}} \circ \Phi: L_1(\M)\to L_\al(\N)\|_{cb},\]
where the CB norm from $L_1(\M)$ to $L_\al(\N)$ is given in equation (\ref{eq:cb}).
Thus, writing $T= \Theta_{\sigma_\N^{-\frac{1}{\al'}}} \circ \Phi$, we obtain
\begin{align*}
\exp({\frac{1}{\al'} D^A_{\al}(\Phi\|\Psi)})&= \|T:L_1(\M)\to L_\al(\N)\|_{cb}\\
&=\sup_{\cR} \| T: L_1(\cR\bar{\otimes}\M)\rightarrow L_1(\cR, L_\al(\N))\|\\
&=\sup_{\rho\in \S(\cR\M)} \|T(\rho_{\cR\M})\|_{L_1(\cR, L_\al(\N))}\\
&\overset{*}{=} \sup_{\rho\in \S(\cR\M)}\inf_{\omega_\cR\in \S(\cR)}\|\omega_\cR^{-\frac{1}{2\al'}} T( \rho_{\cR\M} )\omega_\cR^{-\frac{1}{2\al'}} \|_\al\\
&\leq   \sup_{\rho\in \S(\cR\M)}\|\rho_\cR^{(1-\al)/2\al} T( \rho_{\cR\M} )\rho_\cR^{(1-\al)/2\al} \|_\al\\
&=\exp(\frac{1}{\al'} \TD(\Phi\|\Psi)).
\end{align*}
Here the equation with $*$ follows from the formula of $||X||_{L_{p}(\M, L_{q}(\N))}$ given in the discussion of the non-commutative $L_p$ -space, for $p\leq q$ in the case where $p=1$ and $q=\alpha$.
This proves the assertion for finite von Neumann algebras.
Now we use Haagerup reduction for general the case.
From the definition of $D^A_{\al}(\Phi\|\Psi)$, for any $\epsilon>0$, we have two states $\rho_{\cR\M}, \omega_{\cR\M}\in \S(\cR\M)$ such that
\[D^A_{\al}(\Phi\|\Psi) -\epsilon\leq D_\al (\id_{\cR} \otimes\Phi (\rho_{\cR\M})\|  \id_{\cR} \otimes\Psi (\omega_{\cR\M}))- D_\al (\rho_{\cR\M} \|\omega_{\cR\M}). \]
Note that by Haagerup's reduction Proposition \ref{lemma:ultra}, we have a sequence of finite von Neumann algebras $(\cR\bar{\otimes}\M)_n$ approximating the algebra $\cR\bar{\otimes}\M$ and using the martingale convergence of $D_\al$ we get
\[D^A_{\al}(\Phi\|\Psi) -\epsilon\leq \lim\inf_{m,n} D_\al ((\id_{\cR} \otimes\Phi)_m (\rho_{\cR\M, n})\|  (\id_{\cR} \otimes\Psi)_m (\omega_{\cR\M, n}))- \lim_n D_\al (\rho_{\cR\M, n} \|\omega_{\cR\M, n}), \]
where the subscript $n$ refers to the corresponding approximating states and $m$ refers the restriction of $\id_{\cR} \otimes\Phi$ on the corresponding approximation for $\N$ (see Lemma \ref{lemma:ultra1}). The quantity in the right hand side is in the finite von Neumann algebra and hence we can use previous case and Lemma \ref{lemma:ultra1}
 to get
\[ D^A_{\al}(\Phi\|\Psi) -\epsilon\leq \lim\inf_{m, n} \TD(\Phi_{m, n}\|\Psi_{m, n}) \leq \TD(\Phi\|\Psi),\]
% where $\Phi_{m, n}$ denotes the restriction of $(\Phi\circ E_{\cR\M,m})$ on the space ${\cR\M,m}$,  $(\Phi\circ E_{\cR\M,n})|_{\cR\M, m}$ (see subsection \ref{chain rule in gen. von}).

As $\epsilon$ was arbitrary, the proof is complete.
\end{proof}

%\noindent \textbf{Proof of Theorem \ref{rmv-reg}:}
\begin{proof}[Proof of Theorem \ref{rmv-reg}]
 Using the proof of Theorem 5.4 in \cite{fawzi2021defining} we have for any channels $\Phi, \Psi$
\[  D^{\text{ref}}_\al(\Phi^{\otimes n}\|\Psi^{\otimes n})\leq n D^A_{\al}(\Phi\|\Psi).\]
Now as $\Psi$ is a replacer channel we further have by the previous proposition
\[ D^A_{\al}(\Phi\|\Psi)\leq D^{\text{ref}}_\al(\Phi\|\Psi). \]
Then it follows that for any channel $\Phi$ and a replacer channel $\Psi$ we have
\[D^{\text{ref}}_\al(\Phi^{\otimes n}\|\Psi^{\otimes n})\leq  nD^{\text{ref}}_\al(\Phi\|\Psi).\]
The other direction follows easily from the super-additivity of the channel divergence.
\end{proof}

Finally, for the proof of Theorem \ref{thm:aep-channel-sec},  we need the following generalization of \cite[Lemma 15]{tomamichel} (see also \cite[Lemma 5]{datta-renner} and \cite[Lemma 4]{datta-hyp}) in general von Neumann algebras. Note that here we are using the original definitions of the divergences (without the reference system).
\begin{lemma}\label{lemma-1}
i) Let $\rho,\sigma,\psi\in \M_*^+$ such that $\rho\leq \sigma+\psi$, $\rho(1)=1, \psi(1)<1$. Then there exists a $\widetilde{\rho}\in \M_*^+$ with $\tilde{\rho}(1)=1$ such that
\[\tilde{\rho}\leq (1-\psi(1))^{-1}\sigma \ \ \ \ and \ \ \  F(\rho,\tilde{\rho})\geq 1-\psi(1).\]
ii) Given $\rho,\sigma\in \S(\M)$ and $\lambda >0$, denote $\phi_\lambda:=(\rho-\lambda \sigma)_+$ as the positive part of $\rho-\lambda \sigma$. Then $$D_{\max}^{\eps(\lambda)}(\rho\|\sigma)\leq \log \frac{\lambda}{\sqrt{1-\eps(\lambda)^2}}, \ \text{where} \ \eps(\lambda)=\sqrt{\phi_\lambda(1)(2-\phi_\lambda(1))}.$$
iii) Let $\rho,\sigma\in \S(\M)$. For any $\eps \in (0, 1) $ and $1<\al\le 2$,
\[ D_{\max}^\eps(\rho \|\si)\le D_\al(\rho\|\si)+ \frac{1}{\al-1} \log \frac{2}{\eps^2} -\log ( \sqrt{1-\eps^2}).\]
\end{lemma}

\begin{proof}i) Using Sakai's non-commutative Radon-Nikodym theorem \cite{sakai} for $\sigma\leq \sigma+\psi$, we get an element $c\in \M_+, 0\leq c\leq 1$ such that $\sigma(x)=(\sigma+\psi)(cxc)$, for all $x\in \M$. Assume $\rho(c^2)\neq 0$. Define $$\tilde{\rho}(x)=\rho(c^2)^{-1}\rho(cxc).$$ Clearly $\tilde{\rho}(1)=1$ is unital and for all positive $x$ we have
\[\tilde{\rho}(x)\leq \rho(c^2)^{-1}(\sigma+\psi)(cxc)=\rho(c^2)^{-1}\sigma(x).\]
Now observe that $$1-\rho(c^2)=\rho(1-c^2)\leq (\sigma+\psi)(1-c^2)=(\sigma+\psi)(1)-(\sigma+\psi)(c^2)=\psi(1).$$
So $\rho(c^2)^{-1}\leq (1-\psi(1))^{-1}$ and hence  we get
\[\widetilde{\rho}\leq (1-\psi(1))^{-1}\sigma.\]
Note that, it is impossible to have $\rho(c^2)=0$. Since otherwise, from the earlier calculations, we get $1-\rho(c^2)\leq \psi(1)$. This would imply that $1\leq \psi(1)$, which contradicts the assumption $\psi(1)<1$. Now note that if $(\pi, \xi)$ is a vector representation for $\rho(x)=\lan \xi|\pi(x) |\xi\ran $, then ${(\pi, \rho(c^2)^{-1/2}c\xi)}$ (identifying $c$ with $\pi(c)$) is a vector representation $\tilde{\rho}$. Hence
\[F(\rho,\tilde{\rho})\geq \rho(c^2)^{-1/2}|\langle \xi|c|\xi\rangle|= \rho(c^2)^{-1/2}\langle \xi|c|\xi \rangle=\rho(c^2)^{-1/2}\rho(c)\geq  \rho(c)  \pl,\]
where we use the fact $\rho(c^2)\leq \rho(c)\leq 1$  as $0\le c\leq 1$. Thus,
\[F(\rho,\widetilde{\rho})\geq \rho(c).\]
We compute the infidelity
\begin{align*}
1-F(\rho,\widetilde{\rho})&\leq 1-\rho(c)= \rho(1-c)\leq (\sigma+\psi)(1-c)
= (\sigma+\psi)(1)-(\sigma+\psi)(c)\\
&\leq (\sigma+\psi)(1)-(\sigma+\psi)(c^2)=(\sigma+\psi)(1)-\sigma(1)=\psi(1)\pl,
\end{align*}
where we used the inequality $(\sigma+\psi)(c^2)\leq (\sigma+\psi)(c)$ as $c^2\le c$. This proves i).

For ii), note that $\rho=\lambda \sigma+\rho-\lambda \sigma\leq \lambda \sigma+ \phi_\lambda$. Also note that $\phi_\lambda(1)<\rho(1)=1$ as $\lambda>0$. Using the proof of i) we find a state $\rho'(x)=\rho(c^2)^{-1}\rho(cxc)$ such that $$\rho'\leq (1-\phi_\lambda(1))^{-1}\lambda \sigma,$$
and the fidelity
\[ F(\rho,\rho')\ge \rho(c)\ge \rho(c^2)\ge 1-\phi_\lambda(1)\pl.\]
So the purified distance
$$d(\rho, \rho')=\sqrt{1- F(\rho, \rho')^2}\leq \eps(\lambda).$$
By definition, we have
\[ D_{\max}^{\eps(\lambda)}(\rho\|\sigma)\le D_{\max}(\rho'\|\sigma)\leq \log \frac{\lambda}{\sqrt{1-\eps(\lambda)^2}}.\]
For iii), we first argue for the finite case. Let $e$ be the support of $X=(\rho-\la\si)_+$ and $\A$ be the commutative subalgebra generated by $X$. Clearly $e\in \A$. Look at the inclusion $i: \A\rightarrow \M$.
 Using the duality, we get a restriction map $E:\M_{*}\rightarrow \A_{*}$ defined as
$$E(\gamma)=\gamma|_{\A}.$$
Denote $\widetilde{\rho}=E(\rho)$ and $\widetilde{\sigma}=E(\sigma)$ as the restriction state on the subalgebra $\A$, and, with slight abuse of the notation, we view $\widetilde{\rho}$ and $\widetilde{\sigma}$ as the probability density function w.r.t to the trace $\tau$.  Then it follows that
\[e (\widetilde{\rho}-\lambda \widetilde{\sigma})e \geq 0\ ,   \ e\widetilde{\rho}e\geq \lambda e \widetilde{\sigma}e.\]
Since the function $x\mapsto x^{\al-1}$ is operator monotone for $1<\al\le 2$, we have by the commutativity of  $\widetilde{\rho}$ and $\widetilde{\sigma}$
\[ \lambda^{1-\al} e (\widetilde{\rho}^{\al-1})  (\widetilde{\sigma}^{1-\al}) e\geq e.\]
Note also that $X\leq e\rho e$ and hence $X\leq e \widetilde{\rho} e$.
\begin{align*}
\tau(X)&\leq \lambda^{1-\al}\tau(X e (\widetilde{\rho}^{\al-1})  (\widetilde{\sigma}^{1-\al}) e )
\leq \lambda^{1-\al} \tau(e \widetilde{\rho}  e (\widetilde{\rho}^{\al-1})e  (\widetilde{\sigma}^{1-\al}) e)
\leq\lambda^{1-\al} \tau (\widetilde{\rho}^\al \widetilde{\sigma}^{1-\al})\\
&= \lambda^{1-\al}\exp[ (\al-1) \widetilde{D_\al} (\widetilde\rho\|\widetilde\sigma)]
= \lambda^{1-\al} \exp[(\al-1) {D_\al} (\widetilde\rho\|\widetilde\sigma)]
\leq \lambda^{1-\al} \exp[ (\al-1) {D_\al} (\rho\|\sigma)].
\end{align*}
In the above we used the fact that Petz-R\'enyi relative entropy $\widetilde{D}_\al$ and Sandwiched R\'enyi relative entropy ${D_\al}$ coincides in the commutative case. The last inequality follows from data processing inequality. Now given $\eps$, we choose $\lambda$ such that $\tau(X)=1- \sqrt{1-\eps^2} $, for $X=(\rho-\lambda\sigma)_{+}$. We get from ii) and the above calculation
\[D_{\max}^\eps(\rho \|\si)\le \log \frac{\lambda}{\sqrt{1-\eps^2}} \leq D_\al(\rho\|\si)+ \frac{1}{\al-1} \log ({1- \sqrt{1-\eps^2}})^{-1}-\log (\sqrt{1-\eps^2}).\]
Finally using $(1-\sqrt{1-\eps^2})^{-1}\leq \frac{2}{\eps^2}$, we get the required inequality.

The case of general von Neumann algebra follows from Haagerup reduction Lemma \ref{lemma:ultra}.
\end{proof}
For the proof of Theorem \ref{thm:aep-channel-sec}, we will also need the following lemma on the continuity on $\al\to \tilde{D}_\al$ for channel divergence.
\begin{lemma}\label{lemma:conti}
i) Let $\rho, \sigma \in S(\mathcal{M})$ be two normal states on a von Neumann algebra $\mathcal{M}$ and $\gamma \in (0, 1]$.
Define \[c_\gamma (\rho \|\sigma) := \frac{1}{\gamma} \log \left( 2^{\gamma \tilde{D}_{1+\gamma}(\rho \|\sigma)} + 2^{-\gamma\tilde{D}_{1-\gamma}(\rho \|\sigma)} + 1 \right).\] Then for all $\gamma \in (0, 1]$ and $0 < \delta \leq \frac{\gamma}{2}$,
\[
\tilde{D}_{1+\delta}(\rho \|\sigma) \leq D(\rho \|\sigma) + \delta (c_\gamma(\rho \|\sigma))^2.
\]
ii) Let $\Phi,\Psi:\mathcal{M}_*\to \mathcal{N}_*$ be two quantum channels and $\gamma \in (0, 1]$. Define \[c_\gamma (\Phi \|\Psi) := \frac{1}{\gamma} \log \left( 2^{\gamma \tilde{D}_{1+\gamma}^{{\rm ref}}(\Phi \|\Psi)} + 2^{-\gamma\tilde{D}_{1-\gamma}^{{\rm ref}}(\Phi \|\Psi_\sigma)} + 1 \right).\]
Then for all $\gamma \in (0, 1]$ and $0 < \delta \leq \frac{\gamma}{2}$,
\begin{align}
\tilde{D}_{1+\delta}^{{\rm ref}}(\Phi\|\Psi) \leq D^{{\rm ref}}(\Phi \|\Psi) + \delta (c_\gamma(\Phi \|\Psi))^2. \label{eq:conti}
\end{align}
\end{lemma}
\begin{proof}
The continuity i) for two states is identical to \cite[Lemma 4.13]{bergh2023infinite} for $\mathcal{M}=B(\mathcal{H})$ using relative modular operator. The case ii) for two quantum channels follows from taking supremum of i) over all states $\id_{\mathcal{R}}\ten \Phi(\rho)$ and $\id_{\mathcal{R}}\ten \Psi(\rho)$. Indeed, given $\eps>0$, let $\rho\in S(\mathcal{R}\ten \mathcal{M})$ be a normal state such that
\begin{align*} \tilde{D}_{1+\delta}^{{\rm ref}}(\Phi\|\Psi)\le \tilde{D}_{1+\delta}(\id_{\mathcal{R}}\ten \Phi(\rho)\|\id_{\mathcal{R}}\ten \Psi(\rho))+\eps.
\end{align*}
Using the continuity in i) and the fact $c_\gamma (\id_{\mathcal{R}}\ten \Phi(\rho)\|\id_{\mathcal{R}}\ten \Psi(\rho))\le c_\gamma (\Phi \|\Psi)$, we have
\begin{align*} \tilde{D}_{1+\delta}^{{\rm ref}}(\Phi\|\Psi)-\eps &\le \tilde{D}_{1+\delta}(\id_{\mathcal{R}}\ten \Phi(\rho)\|\id_{\mathcal{R}}\ten \Psi(\rho)) \\
&\le  D(\id_{\mathcal{R}}\ten \Phi(\rho)\|\id_{\mathcal{R}}\ten \Psi(\rho)) + \delta (c_\gamma(\id_{\mathcal{R}}\ten \Phi(\rho)\|\id_{\mathcal{R}}\ten \Psi(\rho)))^2
\\
&\le D^{{\rm ref}}(\Phi \|\Psi) + \delta (c_\gamma(\Phi \|\Psi))^2.
\end{align*}
Since $\eps$ is arbitrary, the proof is complete.
\end{proof}

We are now ready to the main theorem of this section.
%\noindent \textbf{Proof of Theorem \ref{thm:aep-channel}:}
\begin{proof}[Proof of Theorem \ref{thm:aep-channel-sec}]
Assume that for some $\alpha > 1$, we have $D_\alpha^{\text{ref}}(\Phi\|\Psi_\sigma)<\infty$. Using the bound in equation \eqref{eq:bound},
for any $\delta\le \gamma=1-\frac{1}{\alpha}$, we have
\[\tilde{D}_{1+\delta}^{\text{ref}}(\Phi\|\Psi_\sigma) \le \tilde{D}_{2-\frac{1}{\alpha}}^{\text{ref}}(\Phi\|\Psi_\sigma) \le D_\alpha^{\text{ref}}(\Phi\|\Psi_\sigma)<\infty,\]
and \[ c_{1-\frac{1}{\alpha}} (\Phi \|\Psi_\sigma)<\infty.\]

For the upper bound, using Lemma \ref{lemma-1}  we have
\begin{align*}
D_{\max}^{\eps, \text{ref}} (\Phi^{\otimes n} \|\Psi_\sigma^{\otimes n})&\leq D_{1+\delta}^{\text{ref}}(\Phi^{\otimes n}\|\Psi_\sigma^{\otimes n}) + \frac{1}{\delta} \log \frac{2}{\eps^2} -\log (\sqrt{1-\eps^2}).\\
&=n D_{1+\delta}^{\text{ref}}(\Phi\|\Psi_\sigma)+ \frac{1}{\delta} \log\frac{2}{\eps^2} -\log (\sqrt{1-\eps^2})\\
&\le n \tilde{D}_{1+\delta}^{\text{ref}}(\Phi\|\Psi_\sigma)+ \frac{1}{\delta} \log\frac{2}{\eps^2} -\log (\sqrt{1-\eps^2})\\
& \le n D(\Phi\|\Psi_\sigma)+ n\delta (c_{1-\frac{1}{\al}}(\Phi \|\Psi_\sigma))^2+ \frac{1}{\delta} \log\frac{2}{\eps^2} -\log (\sqrt{1-\eps^2}).
\end{align*}
Here, the first equality follows from Theorem \ref{rmv-reg}, the first inequality uses $D_{1+\delta}\le \tilde{D}_{1+\delta}$, and the second inequality uses the continuity from lemma \ref{lemma:conti}. Choosing $\delta=O(1/\sqrt{n})$ yields the upper bound
\[D_{\max}^{\eps, \text{ref}} (\Phi^{\otimes n} \|\Psi_\sigma^{\otimes n})\leq n D^{\text{ref}}(\Phi\|\Psi) + O(\sqrt{n}).\]
where all the constant terms are included in the expression $O(\sqrt{n})$.

The lower bound follows from the AEP for states described in Theorem \ref{eq:aep-states}. Indeed, from the definition of the channel divergence, for a small $\eta>0$, we can find a
von Neumann algebra $\cR$ and state $\rho\in \S(\cR\overline{\ten} \M)$ such that $D^{\text{ref}}(\Phi\|\Psi)-\eta \leq {D}(\id_\cR\otimes \Phi(\rho)\|\id_\cR\otimes \Psi_{\sigma}(\rho))$. Note that $id_\cR\otimes \Psi_{\sigma}(\rho)=\rho_{\cR}\otimes \sigma$, for some state $\rho_{\cR}$. Now we let $x=\id_\cR\otimes \Phi(\rho), y= \rho_{\cR}\otimes \sigma$. Then we get
\begin{align*}
n{D}(\id_\cR\otimes \Phi(\rho)\|\id_\cR\otimes \Psi_{\sigma}(\rho))&=n{D}(x\|y)\\
&=  D_{\max}^{\epsilon}(x^{\otimes n}\| y^{\otimes n})+ O({\sqrt{n}})\\
&\leq D_{\max}^{\eps, \text{ref}} (\Phi^{\otimes n} \|\Psi_\sigma^{\otimes n}) + O({\sqrt{n}})
\end{align*}
where the equality  follows from the AEP of states Theorem \ref{eq:aep-states}, where we hide the second order terms in the expression $O(\sqrt{n})$. And the inequality follows from the definition of the channel divergence. So putting all together we
have
\[D_{\max}^{\eps, \text{ref}} (\Phi^{\otimes n} \|\Psi_\sigma^{\otimes n})\geq n D^{\text{ref}}(\Phi\|\Psi) + O(\sqrt{n})-n\eta.\]
As $\eta$ was arbitrary, we get the required result.
\end{proof}

\section{Divergence Chain Rule and Sequential Applications of Channels}\label{REAT}

\subsection{Semi-finite case:}

Our main technical result in this article is the following chain rule, which we state here in the semi-finite case. The fully general version is provided as Theorem \ref{thm:chain rule1-sec} in the next subsection.
\begin{theorem}\label{thm:chain rule}
Let $\M$ and $\N$ two semi-finite von Neumann algebras and let
$\Phi,\Psi : L_1(\M)\to L_1(\N)$ be two quantum channels. Then for any $\alpha\in (1,\infty]$, $\rho,\sigma\in \S(\M)$ we have
\[{D}_{\al}(\Phi(\rho)\|\Psi(\si))\leq {D}_{\al}(\rho\|\sigma)+{D}_{\al}^{\textup{reg}}(\Phi\|\Psi).\]
\end{theorem}
The above theorem is known as a chain rule for quantum channels, which was first proved by H. Fawzi and O. Fawzi in \cite{fawzi2021defining} in finite dimensions. Here, to establish the semi-finite case, we follow a recent simple proof by Berta and Tomamichel \cite{berta-toma} and adapt ideas from the seminal work of Hayashi and Tomamichel \cite{toma-haya}.

Denote $\text{spec} ( {\sigma})$ as the spectrum set of $\sigma$ and $|\text{spec}({\sigma})|$ as the cardinality of the set. Namely, $|\text{spec}({\sigma})|=k$ if $\sigma=\sum_{i=1}^k \mu_i F_{i}$ has finite spectrum, and $+\infty$ otherwise. As a first step, we extend a useful lemma to infinite dimensions. It is commonly known as Hayashi's pinching map \cite{hayashi}. Given a density operator $\sigma\in \S(\M)$, we denote $$\N_\sigma=\{\sigma\}':=\{a\in \M: a\sigma=\sigma a\}=\{a\in \M: ae^{i\sigma t}=e^{i\sigma t}a \text{ for all } t\in \mathbb{R}\}$$ as the commutant (or centralizer) of $\sigma$. One can show that the trace $\tau|_{\N_\sigma}$ restricted on the subalgebra $\N_\sigma$ is again semi-finite. Indeed, let $e_{\eps}:=1_{[\eps,\infty)}(\sigma)$ be the spectral projection of $\sigma$ on the unit interval $[\eps, \infty)$, for $\eps>0$ and $e_{\{0\}}:= 1_{0}(\sigma)$. If $y\in \N_\sigma$ and $y\ge 0$, and $y\neq 0$, we set $\tilde{y}_{\eps}:=y^{\frac{1}{2}}e_{\eps} y^{\frac{1}{2}}$. Then $\tilde{y}_{\eps}\leq y$ and $\tau(\tilde{y}_{\eps})=\tau(ye_{\eps})\le \norm{y}{}\norm{e_{\eps}}{1}<\infty$. If $\tilde{y}_{\eps}=0$ for all $\eps>0$, then  $y^{\frac{1}{2}}(1-e_0 )y^{\frac{1}{2}}=0$, which shows that $y=e_0 y e_0\in e_0\M e_0$. Hence there exists a $\tilde{y}\in e_0\M e_0$ such that
$0\leq \tilde{y}\leq y$ and $\tau(\tilde{y})<\infty$. Then there exists a trace preserving conditional expectation $\E_{\sigma}:\M\to \N_\sigma$ onto $\N_\sigma$ such that
\[\tau(xy)=\tau(\E_\sigma(x)y)\pl,  \pl  \forall \pl x\in \M, y\in \N_\sigma\pl. \]
$\E_{\sigma}$ can be viewed as generalization of pinching map for the state $\sigma$. In particular, if $\sigma=\sum_{i=1}^k \mu_i E_{i}$ is of finite spectrum with $E_i\ge 0, \sum_{i}E_i=1$ is a finite family of mutually orthogonal projections, then
\[ \E_{\sigma}(x)=\sum_{i=1}^k E_i x E_i\pl, \pl \forall x\in \M.\]
\begin{lemma}\label{lemma:order}
For $\sigma=\sum_{i=1}^k \mu_i E_{i}$ with finite spectrum, $\rho\le k\E_{\sigma}(\rho)$ for any $\rho\in \S(\M)$.
\end{lemma}
\begin{proof}Denote $1$ as the identity element of $\M$.
Note that the $k\times k$ matrix \[\left(\begin{array}{cccc}1&1&\cdots&1\\ 1&1&\cdots&1\\ \vdots& &\ddots& \\ 1& &\cdots&1\end{array}\right),\] satisfies the following operator inequality

\[\left(\begin{array}{cccc}1&1&\cdots&1\\ 1&1&\cdots&1\\ \vdots& &\ddots& \\ 1& &\cdots&1\end{array}\right)\le k\left(\begin{array}{cccc}1& & & \\  &1& &\\ & &\ddots& \\ & & &1\end{array}\right).\] Then
\begin{align*} \rho=&\sum_{i,j=1}^k E_i\rho E_j=\left(E_1,E_2,\cdots, E_k\right)\left(\begin{array}{cccc}\rho&\rho&\cdots&\rho\\ \rho&\rho&\cdots&\rho\\ \vdots& &\ddots& \\ \rho& &\cdots&\rho\end{array}\right)\left(\begin{array}{c} E_1 \\ E_2 \\ \vdots \\ E_k\end{array}\right)
\\ \le &
k \left(E_1,E_2,\cdots, E_k\right)\left(\begin{array}{cccc}\rho& & & \\  &\rho& &\\ & &\ddots& \\ & & &\rho\end{array}\right) \left(\begin{array}{c} E_1 \\ E_2 \\ \vdots \\ E_k\end{array}\right)=k\sum_{i}^k E_i\rho E_i=k\E_{\sigma}(\rho) \qedhere
\end{align*}
\end{proof}

Under the assumption of finite spectrum, we have the following estimate as in finite dimensions.

\begin{lemma}\label{lemma:pinching-spectrum}
Let $\sigma$ be a positive element with finite spectrum. Let $\Phi,\Psi:L_1(\M)\rightarrow L_1(\N)$ be two positive maps. Then for any $\al\in (1,\infty)$ and $\rho\in \S(\M)$,
\[D_\al(\Phi(\rho)\|\Psi(\sigma))\le D_\al(\Phi(\E_\sigma(\rho))\|\Psi(\sigma)) +\frac{\al}{\al-1}\log |\text{spec}(\sigma)|.\]
\end{lemma}
\begin{proof}
Using the pinching inequality in Lemma \ref{lemma:order}, we have for any $\rho$ and positive $\Phi$ , $$\Phi(\rho)\le | \text{spec}(\sigma)| \Phi(\E_{\sigma}(\rho)).$$
Then conjugating by $\Psi(\sigma)^{\frac{1-\al}{2\al}}$ both sides gives us
\[\Psi(\sigma)^{\frac{1-\al}{2\al}}\Phi(\rho)\Psi(\sigma)^{\frac{1-\al}{2\al}}\le |\text{spec}(\sigma)| [\Psi(\sigma)^{\frac{1-\al}{2\al}}\Phi(\E_{\sigma}(\rho))\Psi(\sigma)^{\frac{1-\al}{2\al}}].\]
Now the claim follows from the norm inequality $\norm{x}{\al}\le \norm{y}{\al}$ when $x\le y$.
\end{proof}

The next result is about measured R\'enyi relative entropies. Recall that a positive operator valued measure (POVM) in a von Neumann algebra $\M$ is a positive linear map $\mathfrak{M}^*:L_\infty(X)\to \M$ from some general measure space $X$ such that $\mathfrak{M}^*(1)=1$.
Then the measurement map in the predual is the map $\mathfrak{M}:L_1(\M)\to L_1(X) $ which sends density operators to probability density functions. When $X=\Omega=\{1,\cdots, k\}$ is a finite set, this means the measurement has finite outcome and $\mathfrak{M}=(M_j)_{1\le j\le k}$ is given by a finite family of positive operators in $\M$ such that $\sum_j M_j=1$. For each $\rho\in \S(\M)$, $\mathfrak{M}(\rho)=(\tr(\rho M_j))_{1\le j\le k}$ gives the discrete probability distribution on $\Omega$.

\begin{defi}\label{def:measure}
Let $\rho,\sigma\in \S(\M)$. Then for $\alpha\in [1/2,1)\cup (1, \infty)$, the measured R\'enyi divergence $D_\al^{M}(\rho\|\sigma)$ is defined as follows
\[D_\al^{M}(\rho\|\sigma)=\sup\{ D_\al(\mathfrak{M}(\rho)\|\mathfrak{M}(\sigma)): \mathfrak{M} \text{ is  a  POVM  in } \M \}. \]
\end{defi}
It follows from \cite[Proposition 5.2]{hiai-book} that  one can replace the supremum of POVMs over general measure spaces $L_\infty(X)$ by only the finite outcome POVMs over discrete $\Omega$. In the sequel we will use this equivalent simpler definition.

%\end{defi}

In finite dimensions, it was proved by Ogawa-Mosonyi \cite{mosnyi-ogawa}  that the regularized measured R\'enyi relative entropy coincides with sandwiched R\'enyi relative entropy. This was extended to approximately finite dimensional von Neumann algebra by Hiai and Mosonyi \cite{hiai-mosonyi}.
Below, we extend this theorem to semi-finite cases.

\begin{theorem} \label{thm:hiai-mosonyi}
 Let $\M$ be a semi-finite von Neumann algebra. Then for two states $\rho,\sigma\in \S(\M)$ with $s(\rho)\leq s(\sigma)$, and $\al\in [1,\infty)$,  we have
\[\lim_{n\to\infty}\frac{1}{n}D_\al^M(\rho^{\otimes n}\|\sigma^{\otimes n})= {D}_{\al}(\rho\|\sigma).\]
\end{theorem}
\begin{proof}
The $\leq$ direction  follows from data processing inequality. Indeed, any measurement operator $\mathfrak{M}:L_1(\M)\to L_1(\Omega) $ can be thought of a channel and
$D_\al(\mathfrak{M}(\rho)\|\mathfrak{M}(\sigma))\leq D_\al(\rho\|\sigma)$, by the data processing inequality. And hence $D_\al^{M}(\rho\|\sigma)\leq D_\al(\rho\|\sigma)$. The $\leq$ direction now follows from the fact that the sandwiched R\'enyi relative entropy is additive under tensor product.

% fact that the measured R\'enyi relative entropy is in some sense the smaller than the standard divergences. Indeed, following Prop. 5.5 in \cite{hiai-book} we find
%$D_\al^{M}(\rho\|\sigma)\leq D_\al(\rho\|\sigma)$.
 %\[D_\al^M(\rho^{\otimes n}\|\sigma^{\otimes n})\le  {D}_{\al}(\rho^{\ten n}\|\sigma^{\ten n})\pl. \]
The proof for the other direction is divided into several steps.\\
{\bf Step 1}. \textbf{$\sigma$ has finite spectrum.}
This proof is identical to finite-dimensional case \cite[Theorem III. 7]{mosnyi-ogawa} . Assume that $\sigma=\sum_{i=1}^k \mu_i E_i$ with $k=|\text{spec}(\sigma)|$. Using Lemma \ref{lemma:pinching-spectrum} (with $\Phi=\Psi=\id$),
\begin{align} D_\al(\rho\|\sigma)\le D_\al(\E_\sigma(\rho)\|\sigma)+  \frac{\al}{\al-1}\log k\le D_\al^M(\rho\|\sigma)+ \frac{\al}{\al-1}\log k \label{eq:oneshot}\end{align}
where we use the fact $\E_\sigma(\rho)$ and $\sigma$ commute. Using the fact $|\text{spec}(\sigma^{\otimes n})|\le (n+1)^{k-1}$, we have
\begin{align*}
 {D}_{\al}(\rho\|\sigma)&=\lim_{n\to\infty} \frac{1}{n}{D}_{\al}(\rho^{\otimes n}\|\sigma^{\otimes n})\\
 &\leq \lim_{n\to\infty}\frac{1}{n} {D}_{\al}^M(\rho^{\otimes n}\|\sigma^{\otimes n}) + \frac{1}{n} \frac{\al}{\al-1}\log |\text{spec}(\sigma^{\otimes n})| \\
 &\leq  \lim_{n\to\infty}\frac{1}{n} {D}_{\al}^M(\rho^{\otimes n}\|\sigma^{\otimes n})+\frac{(k-1)}{n}\log(n+1)\\
 &\leq \lim_{n\to\infty}\frac{1}{n}D_\al^M(\rho^{\otimes n}\|\sigma^{\otimes n})\pl.
\end{align*}
{\bf Step 2}. \textbf{ $\sigma$ has bounded spectrum.}
For this proof, we use the spectrum truncation method of Hayashi and Tomamichel \cite{toma-haya}. Assume $m s(\sigma)\leq \sigma\leq M s(\sigma)$. Let $\theta=\theta(\sigma)=\log M-\log m$. For a positive integer $l$, we define the step function
\begin{align*}
f_l(y)=\begin{cases}
  m, & \mbox{if } y=m \\
  2^{\frac{\theta i}{l}}m , & \mbox{if }  y\in (2^{\frac{\theta i}{l}}m,2^{\frac{\theta (i+1)}{l}}m], 0\leq i\leq l.
\end{cases}
\end{align*}
Define
\[\sigma':=f_l(\sigma)=m E_\sigma(\{m\})+\sum_{i=0}^l 2^{\frac{\theta i}{l}}m E_\sigma\Big((2^{\frac{\theta i}{l}}m,2^{\frac{\theta (i+1)}{l}}m]\Big)\]
where $E_\sigma$ is the spectral projection of $\sigma$. Since $2^{-\theta/l}y \leq f_l(y)\leq y$, we have
\[2^{-\theta/l} \sigma\leq \sigma'\leq \sigma,\]
and $|\spec(\sigma')|=l+1$. Using \eqref{eq:oneshot},
\begin{align}\label{eq:boundedspectrum}
D_\al(\rho\|\sigma) &\leq D_\al(\rho\|\sigma')\nonumber\\
&\leq D_\al^{M} (\rho\|\sigma')+\frac{\al}{\al-1} \log (l+1) \nonumber\\
&\leq D_\al^{M} (\rho\|\sigma) +\frac{\al}{\al-1} \log (l+1)+ \frac{\theta}{l}\nonumber\\
&\leq D_\al^{M} (\rho\|\sigma) +\frac{\al}{\al-1} \log (2\lceil \theta +1\rceil)
\end{align}
where we choose $l=\lceil \theta\rceil$ (our $\log$ is base 2). Note that $\theta(\sigma^{\ten n})=\log M^n-\log m^n=n\theta$. Applying \eqref{eq:boundedspectrum} to $\rho^{\ten n}$ and $\sigma^{\ten n}$,
\begin{align*}
 {D}_{\al}(\rho\|\sigma)&=\lim_{n\to\infty} \frac{1}{n}{D}_{\al}(\rho^{\otimes n}\|\sigma^{\otimes n})\\
 &\leq \lim_{n\to\infty}\frac{1}{n} {D}_{\al}^M(\rho^{\otimes n}\|\sigma^{\otimes n}) + \frac{1}{n} \frac{\al}{\al-1}\log (2\lceil \theta(\sigma^{\ten n})+1\rceil) \\
 &\leq  \lim_{n\to\infty}\frac{1}{n} {D}_{\al}^M(\rho^{\otimes n}\|\sigma^{\otimes n})+\frac{1}{n}\log(2n\lceil\theta(\sigma)+1\rceil)\\
 &\leq \lim_{n\to\infty}\frac{1}{n}D_\al^M(\rho^{\otimes n}\|\sigma^{\otimes n})\pl.
\end{align*}
{\bf Step 3.} \textbf{General $\sigma\in \S(\M)$}.
Here we use approximation. For $k\in \mathbb{N}$,
let $e_k:=\mathbbm{1}_{[1/k,k]}(\sigma)$ be the spectral projection of $\sigma$ supported on the interval $[1/k,k]$ and $e_k'$ be any nonzero projection of finite trace orthogonal to $e_k$. Define the quantum channel
\[\Phi_k(x)=e_k x e_k +\frac{\tau(x(1-e_k))}{\tau(e_k')} e_k'\]
It holds that for $\rho_k:=\Phi_k(\rho)\to\rho$ and $\sigma_k:=\Phi_k(\sigma)\to\sigma$ in $L_1$. Indeed,
\begin{align*} \norm{\rho_k-\rho}{1}\le &\tau(\rho (1-e_k))+\norm{(1-e_k)\rho e_k}{1}+\norm{\rho (1-e_k)}{1}\\
\le &\tau(\rho (1-e_k))+\norm{(1-e_k)\sqrt{\rho}}{2}\norm{\sqrt{\rho} e_k}{1}+ \norm{(1-e_k)\sqrt{\rho}}{2}\norm{\sqrt{\rho}}{1}\\
\le &\tau(\rho (1-e_k))+\norm{(1-e_k)\sqrt{\rho}}{2}+ \norm{(1-e_k)\sqrt{\rho}}{2}\\
= &\tau(\rho (1-e_k))+\tau(\rho (1-e_k))^{1/2}+ \tau(\rho (1-e_k))^{1/2}
\end{align*}
which converges to $0$ because $e_k\to s(\sigma)$ weakly and $s(\rho)\le s(\sigma)$. The argument for $\sigma_k$ is similar.
Then by lower semi-continuity \cite[Proposition 3.10]{jencova-1} and data processing inequality of $D_\al$,
\[ {D}_{\al}(\rho\|\sigma)\le \liminf_k {D}_{\al}(\rho_k\|\sigma_k)\le \limsup_k {D}_{\al}(\rho_k\|\sigma_k)\le {D}_{\al}(\rho\|\sigma)\pl.\]
Note that each $\sigma_k$ has bounded spectrum.
\begin{align*}
{D}_{\al}(\rho\|\sigma) &=\lim_k {D}_{\al}(\rho_k\|\sigma_k)\\
&=\lim_k \lim_n \frac{1}{n} {D}_{\al}^M(\rho_k^{\otimes n}\|\sigma_k^{\otimes n})\\
&\leq  \lim_n \frac{1}{n}  {D}_{\al}^M(\rho^{\otimes n}\|\sigma^{\otimes n})
\end{align*}
where we use, once again, the data processing inequality of ${D}_{\al}^M$ for $\rho_k^{\otimes n}=\Phi_k^{\ten n}(\rho^{\otimes n} )$ and $\sigma_k^{\otimes n}=\Phi_k^{\ten n}(\sigma^{\otimes n} )$. That finishes the proof.
\end{proof}
We shall now discuss the proof of chain rule Theorem \ref{thm:chain rule}. We follow the recent simple proof in \cite{berta-toma} (albeit in finite dimensions). We note the following definition of Motsumoto's maximal $f$-divergence  ( \cite{matsumoto-1}) for $D_\al$:
\begin{equation}\label{max-motsumoto}
D_{\al}^{\max}(\rho\|\sigma):=\inf_{\Lambda, P, Q}D_\al(P\|Q),
\end{equation}
where $D_\al(P\|Q)$ is the classical R\'enyi $\al$-divergence, and the infimum is over all possible state preparation channels
\[\Lambda: L_1(\Omega,\mu)\to L_1(\M)\]
such that $\rho=\Lambda(P)$ and  $\sigma=\Lambda(Q)$ for two probability distribution on $(\Omega,\mu)$. See  Hiai's book \cite{hiai-book} for more details about how this divergence is related to the geometric relative entropy for operator convex functions.

%generalization from Hiai's book \cite{hiai-book}.
%\begin{theorem} [Theorem 4.17 of \cite{hiai-book}] \label{thm:reverse}
%Let $\M$ be a general von Neumann algebra and let $\rho,\sigma\in \S(\M)$ be two states. Then
%\[\hat{D}_\al(\rho\|\sigma)=\min_{\Lambda, P, Q}D_\al(P\|Q),\]
%where $D_\al(P\|Q)$ is the classical R\'enyi $\al$-divergence, and the supremum is over all possible state preparation channel
%\[\Lambda: L_1(\Omega,\mu)\to L_1(\M)\]
%such that $\rho=\Lambda(P)$ and  $\sigma=\Lambda(Q)$ for two probability distribution on $(\Omega,\mu)$.
%\end{theorem}

%The above theorem in finite dimensions, was proved by Matsumoto \cite{matsumoto-1,matsumoto-2}, in which he gave explicit construction of $\Lambda$ with finite $\Omega$.
 Note that the finiteness of $\Omega$ is used in the proof of chain rule in \cite{berta-toma}.
Nevertheless, in infinite dimensions, the optimal preparation channel can be from a continuous classical system.
Because of this, for proving Theorem \ref{thm:chain rule}, we start with the baby case that the pair $(\rho,\sigma)$ admits an optimal preparation channel with finite input, and then gradually extends to the general pair of states $(\rho,\sigma)$.
We outline the proof of Theorem \ref{thm:chain rule} below.
\begin{enumerate}
\item {\bf Step 1.} At the outset, we prove the key inequality
\begin{align}{D}_{\al}(\Phi(\rho)\|\Psi(\sigma))\leq {D}_{\al}^M(\rho\|\sigma)+{D}_{\al}(\Phi\|\Psi)+\frac{\al}{\al-1}\log |spec(\sigma)|.\label{eq:key}\end{align}
for $\sigma$ with finite spectrum and $\rho$ such that $\sigma^{-1/2}\rho\sigma^{-1/2}$ also has finite spectrum. The latter corresponds to the finiteness of $\Omega$ in the definition of $D_{\al}^{\max}(\rho\|\sigma)$ , so that we can run the proof from \cite{berta-toma} for such special $\rho$ and $\sigma$. Note that the chain rule follows from the regularization of the above inequality \eqref{eq:key}.
\item {\bf Step 2.} For the second step, we extend \eqref{eq:key} for $\sigma$ still with finite spectrum, but relaxing the assumption of $\rho$ such that $\sigma^{-1/2}\rho\sigma^{-1/2}$ is of bounded spectrum. This is obtained using the spectrum truncation technique from Hayashi and Tomamichel's work \cite[Theorem 14]{toma-haya} in a careful way.
\item {\bf Step 3.} We extend \eqref{eq:key} for $\sigma$ with finite spectrum and general state $\rho\in \S(\M)$. Here we use approximation with bounded and invertible operators from the previous two steps. At this stage, we obtain chain rule with restriction that $\sigma$ has finite spectrum.
\item {\bf Step 4.}  Now we extend to $\si$ with bounded spectrum. Here, we used again the spectrum truncation technique from \cite{toma-haya} and our Theorem \ref{thm:hiai-mosonyi} in semi-finite case.
\item {\bf Step 5.}  Finally we approximate general $\sigma$ by bounded and invertible states from Step 4. to finish the proof.
\end{enumerate}

\begin{proof}[Proof of Theorem \ref{thm:chain rule}]
{\bf Step 1. Both $\sigma$ and $\sigma^{-1/2}\rho\sigma^{-1/2}$ have finite spectrum.}
We recall Matsumoto's construction (see  Proposition 3 in  \cite{berta-toma} ). Assume $\sigma^{-\frac{1}{2}}\rho\sigma^{-\frac{1}{2}}=\sum_{i=1}^n \mu_i F_i$ for some finite family of orthogonal projections $\{F_i\}_{i=1}^n$. Define the map $\Lambda:l_1^n\rightarrow L_1(\M)$
\[\Lambda(y)=\sum_i y_i\frac{\sigma^{1/2}F_i\sigma^{1/2}}{\tau(\sigma F_i)},\]
where $y=(y_i)_{i=1}^n$ is a finite sequence. Note that $\tau(\sigma F_i)\le \mu_i^{-1}\tau(\sigma \sigma^{-1/2}\rho\sigma^{-1/2})=\mu_i^{-1}\tau(\rho)=\mu_i^{-1}$ is of finite trace. $\Lambda$ is then a completely positive and trace preserving map, hence a preparation channel.
From the definition, we have \[\rho=\sum_i p_i\Lambda(|i\rangle\langle i| ) \  \text{and} \ \sigma=\sum_i q_i \Lambda(|i\rangle\langle i| ),\] for the probability distribution $p=(p_i)=(\mu_i \tau(\sigma F_i)), q=(q_i)=(\tau(\sigma F_i))$.
Hence we have \[\Phi(\rho)=\sum_i p_i\Phi(\Lambda(|i\rangle\langle i| ))\pl , \pl \Psi(\sigma)=\sum_i q_i \Psi(\Lambda(|i\rangle\langle i|) ).\]
Let us denote the element $\Lambda_i=\Lambda(|i\rangle\langle i| )$. It is positive and trace 1, hence a density operator in $\S(\M)$. Now by the data-processing inequality we have
\begin{align}
{D}_{\al}(\Phi(\rho)\|\Psi(\sigma))&\leq {D}_{\al}\big(\sum_i p_i |i\rangle\langle i| \otimes \Phi(\Lambda_i)\|\sum_i q_i |i\rangle\langle i| \otimes \Psi(\Lambda_i)\big) \nonumber\\
&=\frac{1}{\al-1} \log (\sum_i p_i^\al q_i^{1-\al} 2^{(\al-1)D_\al(\Phi(\Lambda_i)\|\Psi(\Lambda_i)))}) \nonumber\\
&\le D_\al(p\|q)+\max_i D_\al(\Phi(\Lambda_i)\|\Psi(\Lambda_i)) \nonumber\\
&\le {D}_\al^{\max}(\rho\|\sigma)+ {D}_{\al}(\Phi\|\Psi)\pl.\label{eq:chainhat}
\end{align}
Here, we used the definition of the maximal divergence in equation (\ref{max-motsumoto}) with optimal $\Lambda$ and $p,q$ on the space $\Omega=\{1,\cdots, n\}$.
Now take the pinching operator $\E_\sigma$ and note that since $\E_\si(\rho)$ and $\sigma$ commute, the max $f$-divergence and the measured R\'enyi divergences coincide,
$${D}_\al^{\max}(\E_\sigma(\rho)\|\sigma)=D_\al^M(\E_\sigma(\rho)\|\sigma)=D_\al(\E_\sigma(\rho)\|\sigma)\pl. $$
Thus, from the calculations above \eqref{eq:chainhat} we have
\begin{align*}
{D}_{\al}(\Phi(\E_\sigma(\rho))\|\Psi(\sigma))&\leq {D}_\al^{\max}(\E_\sigma(\rho)\|\sigma)+{D}_{\al}(\Phi\|\Psi)\\&=D_\al^M(\E_\sigma(\rho)\|\sigma)+{D}_{\al}(\Phi\|\Psi)\\ &\leq D_\al^M(\rho\|\sigma)+{D}_{\al}(\Phi\|\Psi).
\end{align*}
The last inequality follows from the data processing inequality of $D_\al^M$. Now applying Lemma \ref{lemma:pinching-spectrum}, we get
\begin{align}\label{eq:key1}{D}_{\al}(\Phi(\rho)\|\Psi(\sigma))\leq {D}_{\al}^M(\rho\|\sigma)+{D}_{\al}(\Phi\|\Psi)+\frac{\al}{\al-1}\log |spec(\sigma)|.\end{align}

{\bf Step 2. $\sigma$ has finite spectrum and $\sigma^{-1/2}\rho\sigma^{-1/2}$ has bounded spectrum.}
Let $\sigma=\sum_k \lambda_k E_k$ be the spectral decomposition with $\{E_j\}_{j=1}^n$ being orthogonal decomposition and $\sum_{j=1}^nE_j=s(\sigma)$. Suppose $\rho$ is a density operator such that $s(\rho)\le s(\sigma)$ and  \begin{align}ms(\sigma)\le \sigma^{-1/2}\rho\sigma^{-1/2}\le Ms(\sigma)\label{eq:1}\end{align} for some $M>m>0$ ($\sigma^{-1}$ is the generalized inverse on its support $s(\sigma)$).

We apply the spectrum truncation from \cite{toma-haya} to $\sigma^{-1/2}\rho\sigma^{-1/2}$.  First we observe that \eqref{eq:1} imples $m\sigma\leq \rho\leq M \sigma$ and thus we can, instead, assume
$m s(\sigma)\leq \rho\leq M s(\sigma)$ for some positive constants $M>m>0$. Consequently $ms(\sigma)\leq \E_\si(\rho)\leq M s(\sigma)$, that is, $\E_\si(\rho)$ is bounded and invertible on its support $s(\sigma)$ with the spectral decomposition $\E_\si(\rho)= \int_{c}^C y dF_y$, where $F_y$ is the corresponding spectral measure. Let $\theta=\theta(\rho)=\log M-\log m$. For an positive integer $l$, we define the step function
\begin{align*}
f_l(y)=\begin{cases}
  m, & \mbox{if } y=m \\
  2^{\frac{\theta i}{l}}m , & \mbox{if }  y\in (2^{\frac{\theta i}{l}}m,2^{\frac{\theta (i+1)}{l}}m], 0\leq i \leq l
\end{cases}
\end{align*}
and the functional calculus
\[w_l:=f_l(\E_\si(\rho))=m F(\{m\})+\sum_{i=0}^l 2^{\frac{\theta i}{l}}m F((2^{\frac{\theta i}{l}}m,2^{\frac{\theta (i+1)}{l}}m]).\]
Since $2^{-\theta/l}y \leq f_l(y)\leq y$, we have
\[2^{-\theta/l}\E_\si(\rho) \leq w_l\leq \E_\si(\rho),\]
and by construction $w_l$ is a sub-state with simple spectrum. Note that from the relation
\[2^{-\theta/l}\E_\si(\rho) \leq w_l\leq \E_\si(\rho)\] we get for any quantum channel $\Phi$
\[2^{-\theta/l}\Phi(\E_\si(\rho)) \leq \Phi(w_l)\leq \Phi(\E_\si(\rho)),\] and hence
\[D_\al(\Phi(\E_\si(\rho))\|\Psi(\sigma))\leq   D_\al(\Phi(w_l)\|\Psi(\sigma)) +  \frac{\theta}{l}.\]
Now since both $w_l$ and $\sigma$ has finite spectrum and they commute, $\sigma^{-1/2}w_l\sigma^{-1/2}$ has finite spectrum. Using the same argument in Step 1. we get
\begin{align*}
D_\al(\Phi(\E_\si(\rho))\|\Psi(\sigma))&\leq  D_\al(\Phi(w_l))\|\Psi(\sigma)) +  \frac{\theta}{l}\\
& \leq [D_\al^{M} (w_l \| \sigma)+ D_\al (\Phi\|\Psi)]+  \frac{\theta}{l}\\
& \leq [D_\al^{M} (\E_\si(\rho)\|\sigma)+ D_\al (\Phi\|\Psi)]+  \frac{\theta}{l}.
\end{align*}
The last inequality follows because $ w_l\leq \E_\si(\rho)$.
Now letting $l\to\infty$ we have
$$D_\al(\Phi(\E_\si(\rho))\|\Psi(\sigma))\leq D_\al^{M} (\E_\si(\rho)\|\sigma)+ D_\al (\Phi\|\Psi).$$
The inequality \eqref{eq:key1} follows from the same argument of Step 1.

{\bf Step 3. $\sigma$ has finite spectrum and $\rho$ is a state with $s(\rho)\le s(\sigma).$}
Consider $\E_\si(\rho)=\sum_k E_k \rho E_k$. We truncate the spectrum of each $\rho_k:=E_k \rho E_k$.
Let us denote $e_k:=s(\rho_k)$ as the support of $\rho_k$ and $e_{n,k}={e}_{[1/n,n]} (\rho_k)$ as the spectral projection of $\rho_k$ on the interval $[1/n,n]$, for $n\in \mathbb{N}$. So $s(\E_\sigma(\rho))=:e=\sum_{k}e_k$ and also $e_{n, k}\to e_k$ weakly as $n\to\infty$.

For each $n\in \mathbb{N}$,  define the following quantum channel
\[T_n(x)=\sum_k e_{n,k}  x e_{n,k}+ \sum_k  \frac{\tau(x (E_k-e_{n,k}))}{\tau(E_k-e_{n,k})} (E_k-e_{n,k}). \]
It should be noted that $\tau(E_k)<\infty$, for any $k$.
Note that $T_n(E_k)=E_k$, hence $T_n(\sigma)=\sigma$. Moreover, $T_n(\E_\sigma(\rho))\to\E_\si(\rho)$ in $L_1$ and each $T_n(\E_\sigma(\rho))$ has bounded spectrum and commutes with $\sigma$ . Note that $T_n(\E_\sigma(\rho))=T_n(\rho)=\E_\sigma(T_n(\rho))$. Now using Lemma \ref{lemma:pinching-spectrum} and Step 2, we have
\begin{align*}
D_\al(\Phi(\rho)\|\Psi(\sigma))&\leq \frac{\al}{\al-1}\log |\spec(\sigma)| + D_\al(\Phi(\E_\si(\rho))\|\Psi(\sigma))\\
&\leq  \frac{\al}{\al-1}\log |\spec(\sigma)| + \lim\inf_n  D_\al(\Phi\circ \E_\si \circ T_n(\rho)\|\Psi\circ T_n(\sigma))\\
&\leq  \frac{\al}{\al-1}\log |\spec(\sigma)|+  \lim\inf_n D_\al^{M} ( \E_\si\circ T_n(\rho)\|\sigma)+ D_\al (\Phi\|\Psi)\\
&\leq   \frac{\al}{\al-1}\log |\spec(\sigma)|+  D_\al^{M} (\rho\|\sigma)+ D_\al (\Phi\|\Psi).
\end{align*}
{\bf Step 4. $\sigma$ has bounded spectrum and $s(\rho)\le s(\sigma)$.}
We truncate the spectrum of $\sigma$ in this step. Assume that $b^{-1}1\leq \sigma\leq b 1$ for some $b>0$. Define $\theta=\theta(\sigma)=\log b-\log 1/b$. We proceed similarly as Step 2 of Theorem \ref{thm:hiai-mosonyi} to define the sub-state $\sigma'=f_l(\sigma)$ with finite spectrum and
\[2^{\frac{-\theta}{l}}\sigma\le \sigma'\le \sigma.\]
Using Step 3 and $|\spec(\sigma')|=l+1$, we get
\begin{align*}
D_\al(\Phi(\rho)\|\Psi(\sigma)) &\leq D_\al(\Phi(\rho)\|\Psi(\sigma'))\\
&\leq D_\al^{M} (\rho\|\sigma')+ D_\al (\Phi\|\Psi) +\frac{\al}{\al-1} \log (l+1) \\
&\leq D_\al^{M} (\rho\|\sigma)+ D_\al (\Phi\|\Psi) +\frac{\al}{\al-1} \log (l+1)+ \frac{\theta(\sigma)}{l}\\
&\leq D_\al^{M} (\rho\|\sigma)+ D_\al (\Phi\|\Psi) +\frac{\al}{\al-1} \log (2\log \lceil \theta(\sigma)+1\rceil)
\end{align*}
where in the last line we choose $l=\lceil \theta(\sigma)\rceil$. The third inequality follows by comparing $D_\al^{M} (\rho\|\sigma')$ and $D_\al^{M} (\rho\|\sigma)$ from the relation $2^{\frac{-\theta}{l}}\sigma\le \sigma'\le \sigma$.
Now, in the i.i.d. setting, $\theta(\sigma^{\otimes n})=n \theta(\sigma)$, hence
\begin{align*}
&\lim_{n\to\infty} \frac{1}{n}D_\al(\Phi^{\otimes n} (\rho^{\otimes n})\| \Psi^{\otimes n} (\sigma^{\otimes n} ))\\ \leq &\lim_{n\to\infty}  \frac{1}{n} D_\al^{M} (\rho^{\otimes n}\|\sigma^{\otimes n}) + \lim_{n\to\infty}  \frac{1}{n} D_\al (\Phi^{\otimes n}\|\Psi^{\otimes n})+ \lim_{n\to\infty}   \frac{\al}{\al -1}\frac{\log (2 n\lceil \theta(\sigma)+1\rceil )}{n}
\end{align*}
Therefore, by the additivity of sandwiched entropy under tensor power and Theorem \ref{thm:hiai-mosonyi},  we have
\[{D}_{\al}(\Phi(\rho)\|\Psi(\sigma))\leq {D}_{\al}(\rho\|\sigma)+{D}_{\al}^{\textup{reg}}(\Phi\|\Psi).\]
{\bf Step 5. General states $\sigma$ and $\rho$.}
For each $k\in \mathbb{N}$, consider the projection $e_k=\mathbbm{1}_{[\frac{1}{k}, k]} (\sigma)$, that is, the spectral projection of $\sigma$ supported on the interval $[\frac{1}{k}, k]$. Choose a projection $e_k'$ of finite trace in the orthogonal complement of $e_k$, that is, $e_k'e_k=0$. Define the quantum channel
$$\E_k(x)= e_kxe_k+ \frac{\tau(x(1-e_k))}{\tau(e'_k)}e_k'$$
 Note that $\E_k(\rho)\to\rho$ and $\E_k(\sigma)\to\sigma$ in $L_1$ (as in the proof of Theorem \ref{thm:hiai-mosonyi}). Since each $\E_k(\sigma)$ has bounded spectrum, we have
 \begin{align*}
D_{\al}(\Phi(\rho)\|\Psi(\sigma))&\leq \liminf_k D_{\al}(\Phi(\E_k(\rho))\|\Psi(\E_k(\sigma)))\\
&\leq \liminf_k D_{\al}(\E_k(\rho)\|\E_k(\sigma))+ {D}_{\al}^{\textup{reg}}(\Phi\|\Psi)\\
&\le D_\al(\rho\|\sigma)+  {D}_{\al}^{\textup{reg}}(\Phi\|\Psi)
\end{align*}
 where we used lower semi-continuity and data processing inequality of $D_\al$ respectively. This completes the proof.
\end{proof}

\subsection{For general von Neumann algebras}
We shall now apply the Haagerup's reduction to quantum channels. Let $(\M,\phi)$ and $(\N,\psi)$ be two von Neumann algebras equipped with normal faithful states $\phi$ and $\psi$ respectively. We write $\hat{\M}$ (resp. $\hat{\N}$) as the crossed product algebra and denote $E_{\M}$ (resp. $E_{\N}$) and $E_{\M_n}$ (resp. $E_{\N_n}$) as the conditional expectation defined in section \ref{Haagerup reduction}. We also denote $\iota_\M:\M\to \hat{\M}$ and $\iota_{\M_n}:\M_n\to \hat{\M}$ as the inclusion maps, and similarly for $\iota_\N, \iota_{\N_n}$.
Given a normal UCP map $\Phi^*:\N\to \M$, we define the natural extension map
 \[\hat{\Phi}^*=\iota_\M\circ\Phi^*\circ E_{\N}:\hat{\N}\to \hat{\M} \]
 and for $m,n\in \bN^+$, a sequence of approximation of it as follows
 \[ \Phi^*_{m,n}= E_{\M_m}\circ\hat{\Phi}^*\circ\iota_{\N_n}: \N_n\to \M_m.\]
The relation of the above UCP maps can be summarized in the following commuting diagram
\begin{center}
\begin{tikzpicture}
  \matrix (m) [matrix of math nodes,row sep=3em,column sep=4em,minimum width=2em]
  {
    {\color{white}\hat{l}}\N & \hat{\N} & \N_n{\color{white}\hat{l}} \\
     {\color{white}\hat{l}}\M & \hat{\M} & \M_m{\color{white}\hat{l}} \\};
  \path[-stealth]
    (m-1-1) edge node [left] {$\Phi^*$} (m-2-1)
    (m-1-2) edge node [above] {$E_\N$} (m-1-1)
    edge node [right] {$\hat{\Phi}^*$} (m-2-2)
    (m-1-3) edge node [above] {$\iota_{\N_n}$} (m-1-2)
    (m-1-3) edge node [right] {$\Phi_{m,n}^*$} (m-2-3)
    (m-2-1) edge node [below] {$\iota_\M$} (m-2-2)
  (m-2-2) edge node [below] {$E_{\M_m}$} (m-2-3);
\end{tikzpicture}

\end{center}
At predual level, each $\Phi_{m,n}: \M_{m,*}\to \N_{n,*}$ is a quantum channel between the finite von Neumann algebras $\M_m$ and $\N_n$.
\begin{lemma}
\label{lemma:ultra1}
Let $\Phi,\Psi:\M_*\to \N_*$ be two quantum channels. Then for any $m,n\in\bN^+$,
\[ D_\al(\Phi \| \Psi)=D_\al(\hat{\Phi}\|\hat{\Psi})\ge D_\al(\Phi_{m,n}\|\Psi_{m,n})\pl.\]
\end{lemma}
\begin{proof} Note that $E_\M\circ \iota_{\M}=\id_\M$ and $E_\N\circ \iota_{\N}=\id_\N$.
Then we have
\[\hat{\Phi}^*=\iota_{\M}\circ \Phi^*\circ E_\N\pl, \pl  \Phi^*=E_\M\circ \iota_{\M}\circ \Phi^*\circ E_\N\circ \iota_{\N}=E_\M\circ\hat{\Phi}^*\circ \iota_{\N}\pl.\]
By data processing,
\[D_\al(\hat{\Phi}\|\hat{\Psi})= D_\al(\Phi \| \Psi).\]
The inequality also follows from the relation $\Phi^*_{m,n}= E_{\M_n}\circ\hat{\Phi}^*\circ\iota_{\N_n}$ and data processing inequality.
\end{proof}
Now we are ready to extend the chain rule to the general von Neumann algebras by Haagerup reudction.
 \begin{theorem}\label{thm:chain rule1-sec}
Let $\M,\N$ be two von Neumann algebras and $\Phi,\Psi:\M_*\to \N_*$ be two quantum channels. Then for any $\alpha\in (1,\infty]$ and $\rho,\sigma\in \S(\M)$,
\begin{align*}{D}_{\al}(\Phi(\rho)\|\Psi(\si))\leq {D}_{\al}(\rho\|\sigma)+{D}_{\al}^{\textup{reg}}(\Phi\|\Psi).\end{align*}
\end{theorem}
\begin{proof} We assume that $\sigma$ and $\Psi(\sigma)$ have full support. Otherwise we can restrict our discussion on their support. Consider $\M,\M_n\subset \hat{\M}$ and $\N,\N_n\subset \hat{\N}$ as the Haagerup construction discussed above. For any $\rho\in \S(\M)$,
we have
\begin{align*} &\Phi_{m, n}(\rho_m)=\rho_m \circ E_{\M_m}\circ\hat{\Phi}^*\circ\iota_{\N_n}=\hat{\Phi}(\rho_m)|_{\N_n}\pl,\\ &\hat{\Phi}(\hat{\rho})=\rho\circ E_\M \circ\iota_{\M}\circ \Phi^*\circ E_\N=\Phi(\rho)\circ E_\N=\widehat{\Phi(\rho)}\pl.
\end{align*}
By the martingale convergence \cite[Theorem 3.1]{hiai-mosonyi}
\begin{align*} &\liminf_{m}\lim_{n}D_\al(\Phi_{m,n}(\rho_m)\|\Psi_{m,n}(\sigma_m))\\ =& \liminf_{m}\lim_{n}D_\al(\hat{\Phi}(\rho_m)|_{\N_n}\|\hat{\Psi}(\sigma_m)|_{\N_n})
\\ =&  \liminf_{m}D_\al(\hat{\Phi}(\rho_m)\|\hat{\Psi}(\sigma_m))
\\ \ge& D_\al(\hat{\Phi}(\hat{\rho})\|\hat{\Psi}(\hat{\sigma}))=D_\al(\widehat{\Phi(\rho)}\|\widehat{\Psi(\sigma)})=D_\al(\Phi(\rho)\|\Psi(\sigma)).
\end{align*}
The above last equality is due to Proposition \ref{lemma:ultra}. Then we have
 \begin{align*}D_\al(\Phi(\rho)\|\Psi(\sigma)) \le &\liminf_{m}\lim_{n}D_\al(\Phi_{m,n}(\rho_m)\|\Psi_{m,n}(\sigma_m))\\
\le &\liminf_{m}\lim_{n} D_\al(\rho_m\|\sigma_m)+ D_\al^{\textup{reg}}(\Phi_{m,n}\|\Psi_{m,n})\\
 \le  &D_\al(\rho\|\sigma)+ D_\al^{\textup{reg}}(\Phi\|\Psi)
\end{align*}
where the second inequality is the chain rule Theorem \ref{thm:chain rule} for channels on finite von Neumann algebra $\Phi_{m,n},\Psi_{m,n}:\M_m\to\N_n$, and the last inequality uses Proposition \ref{lemma:ultra} and the fact that $D_\al^{\textup{reg}}(\Phi_{m,n}\|\Psi_{m,n})\leq D_\al^{\textup{reg}}(\Phi\|\Psi) $ can be seen similarly to Lemma \ref{lemma:ultra1}.
\end{proof}

\subsection{Sequential Application of Channels.}

As a corollary, we obtain the relative entropy accumulation theorem for two sequences of channels.
\begin{theorem}
Let $\eps \in (0,1)$ and $1 < \alpha \leq 2$. Let $\M_1,\cdots,\M_n,\M_{n+1}$ be a sequence of von Neumann algebras and $\Phi_i,\Psi_i:(\M_i)_*\to (\M_{i+1})_*$ for $1\le i\le n$ be two sequences of quantum channels. Then
\[D_{\max}^\eps (\Phi_n\circ\cdots\circ\Phi_1 \| \Psi_n\circ\cdots\circ\Psi_1)\leq \sum_{i=1}^n D_\al^{\textup{reg}}(\Phi_i\|\Psi_i) +  \frac{1}{\al-1} \log \frac{2}{\eps^2} -\log ( \sqrt{1-\eps^2}).\]
\end{theorem}\begin{proof}
%We recall the following bound of $D_{\max}^\eps$ by $D_\al$ that
%for  $\rho\in S(\M)$, $\sigma$ any positive linear functional, $\epsilon >0$ and $1<\al\le 2$,
Using Lemma \ref{lemma-1},
For any $\eps \in (0, 1) $ and $1<\al\le 2$,

\begin{align} \label{eq:boundmax}
D_{\max}^\eps(\rho \|\si)\le D_\al(\rho\|\si)+ \frac{1}{\al-1} \log \frac{2}{\eps^2} -\log ( \sqrt{1-\eps^2}).
\end{align}
%This bound  is provided in Lemma \ref{lemma-1} in the previous section.
Applying \eqref{eq:boundmax} and the chain rule Theorem \ref{thm:chain rule1-sec} sequentially we obtain that for any state $\rho\in S(\M_1)$
\begin{align*}
&D_{\max}^\eps (\Phi_n\circ\cdots\circ\Phi_1(\rho)\| \Psi_n\circ\cdots\circ\Psi_1(\rho)) \\&\leq D_\al(\Phi_n\circ\cdots\circ\Phi_1(\rho)\| \Psi_n\circ\cdots\circ\Psi_1(\rho)) + \frac{1}{\al-1} \log \frac{2}{\eps^2} -\log ( \sqrt{1-\eps^2})\\
&\leq D_\al(\Phi_{n-1}\circ\cdots\circ\Phi_1(\rho)\| \Psi_{n-1}\circ\cdots\circ\Psi_1(\rho))
+ D_\al^{\text{reg}}(\Phi_n\|\Psi_n) + \frac{1}{\al-1} \log \frac{2}{\eps^2} -\log ( \sqrt{1-\eps^2})\\
&\leq D_\al(\rho\|\rho)+ \sum_{i=1}^n  D_\al^{\text{reg}}(\Phi_i\|\Psi_i) +  \frac{1}{\al-1} \log \frac{2}{\eps^2} -\log ( \sqrt{1-\eps^2}).
\end{align*}
%Choosing $\al=1+ O(1/\sqrt{n})$ and noting that $D_\al(\rho\|\rho)=0$ vanishes we get
%\[D_{max}^\eps (\E_n\circ\cdots\circ\E_1(\rho)\| \F_n\circ\cdots\circ\F_1(\rho))\leq \inf_{\al>1}\sum_{i=1}^nD_\al^{reg}(\E_i\|\F_i)+ O(\sqrt{n}).\]
The assertion follows from taking the supremum over all $\rho\in S(\M_1)$ and noting that $D_\al(\rho\|\rho)=0$.
\end{proof}

\section{Application: Chain rule for the quantum relative entropy}
In this section we discuss how our AEP results can be used in quantum information theoretic tasks in the context of von Neumann algebras. The state discrimination problem has been a widely studied topic (\cite{hiai-petz, ogawa-nag, szkola, Audenaert}) in finite dimensional quantum information theory. In infinite dimension, notably in the context of von Neumann algebras, the state discrimination problem was studied in detail by \cite{jaksic}. They investigated quantum hypothesis testing beyond the i.i.d case and extended this study in scenarios such as quantum lattice systems. Note that, their work only provides estimates of first order and hence the speed of convergence of the rate of decay to its limiting value can not be directly measured. The results of \cite{ke-Li} can be directly applied to obtain the speed of convergence for the i.i.d case to the setting of quantum hypothesis testing of quantum states modelled in abstract von Neumann algebras.

The results we obtain here can be applied to the case of quantum channel discrimination. Discrimination of two quantum channels is a natural extension of the state discrimination problem. Nevertheless, despite inherent mathematical links between the channel and state discrimination problems, discrimination of channels is more involved, partially due to the fact that in this case, entanglement with an ancillary system may matter. Consider two quantum channels $\Phi$ and $\Psi$. In the setting of asymptotic channel discrimination, there are two types of strategies: \textbf{Non-adaptive strategies (also known as parallel strategies)} where  we are given access to a “black-box” for $n$ uses of a channel $\G$, which is either $\Phi$ or $\Psi$, that can be used in parallel before performing a measurement. Based on the measurement outcome, we guess whether $\G=\Phi$ or $\G=\Psi$. In the other scenario known as \textbf{Adaptive strategies (also called sequential strategies)}, here we are also given “black-box” access to $n$ uses of a channel $\G$ but unlike the previous case,  after
each use of a channel we are allowed to perform an adaptive quantum channel on the appropriate system before we perform a measurement at the end. Finally we perform a measurement on the final state arising this way after $n$ rounds and guess if $\G = \Phi$ or $\G = \Psi$.

In finite dimensional systems, Cooney \emph{et al} (\cite{CMW}) proved the quantum Stein’s lemma between an arbitrary quantum channel and a replacer channel. This work led to the conclusion that at least in the asymptotic regime, a non-adaptive strategy is optimal in the setting of Stein’s lemma. Later Wang and Wilde \cite{wang-wilde} showed that  for non-adaptive strategies, where we fix the type-I error to be bounded by $\epsilon$, the asymptotic optimal rate of the type-II error exponent is given by $D^{\textup{ref, reg}}(\Phi \| \Psi)$, when $\epsilon$ goes to 0. Here $$D^{\textup{ref, reg}}(\Phi \| \Psi):=\lim_{n\to\infty} \frac{1}{n} D^{\textup{ref}}(\Phi^{\otimes n} \| \Psi ^{\otimes n})\pl .$$
%\OF{The setting of Stein's lemma and the asymmetric regime is the same thing, it's strange to give them different names in a sentence and the next one.}
For adaptive strategies, they showed the asymptotic optimal rate is given by $D^{A}(\Phi\|\Psi)$ (see section \ref{AEP for channels}). Interestingly, in \cite{fang} it was shown that
$D^{A}(\Phi\|\Psi)=D^{\textup{ref, reg}}(\Phi \| \Psi)$, implying that adaptive strategies do not offer an advantage in the setting of Stein’s lemma for quantum channels. This was done by establishing the following  chain rule for the quantum relative entropy
\begin{align*}{D}(\Phi(\rho) \| \Psi(\si))\leq {D}(\rho \| \sigma)+{D}^{\textup{ref, reg}}(\Phi\|\Psi).\end{align*}
Although the channel discrimination problem has not yet been investigated in the setting of infinite dimensional systems, with the techniques we developed in this paper, we are able to justify the above chain rule in general von Neumann algebras, which is likely to be useful in extending the main results of \cite{CMW} and \cite{wang-wilde} to infinite dimensions.

% here is enough to generalize the result of \cite{fang}. Indeed we give a brief proof sketch of the following theorem:

%Establishing the precise quantities that are relevant to quantify the optimal rate of errors are beyond the scope of this paper and we aim to pursue it in our future endeavour. This would require to generalize the main results of \cite{CMW} and \cite{wang-wilde}. It is much likely that the relevant quantities in finite dimensional systems would extend to be the same in infinite dimension as well. And hence assuming this, we can use our AEP for channels (Theorem \ref{thm:aep-channel}) to conclude that in the setting of general von Neumann algebras, the sequential protocol is not necessary to discriminate between an arbitrary quantum channel and a replace channel. We also take this opportunity to mention that the techniques  we developed here is enough to generalize the result of \cite{fang}. Indeed we give a brief proof sketch of the following theorem:
\begin{theorem}\label{thm:application}
Let $\M,\N$ be two von Neumann algebras and $\Phi,\Psi:\M_*\to \N_*$ be two quantum channels such that $D_{\max} ^{\textup{ref}}(\Phi\|\Psi)<\infty$. Then for $\rho,\sigma\in \S(\M)$, we have
\begin{align*}{D}(\Phi(\rho) \| \Psi(\si))\leq {D}(\rho \| \sigma)+{D}^{\textup{ref, reg}}(\Phi\|\Psi)\pl.\end{align*}
\end{theorem}

Before proving this theorem we need a lemma which is a variant of Lemma \ref{lemma-1}. The exact error terms appearing in the right hand side are slightly different and hence we are rewriting this lemma similarly as  Lemma 3.4 in \cite{fang} in von Neumann algebras:

\begin{lemma}\label{lemma:last}
For a normal state $\rho$ and a normal positive functional $\sigma$ on a von Neumann algebra $\M$, and any $\eps\in (0,1)$ with $n\geq 2 g(\eps)$ we have
\[\frac{1}{n}D_{\max}^{\eps}(\rho^{\otimes n}\|\sigma^{\otimes n})\leq D(\rho\|\sigma)+ \frac{4 \log (\mu)\sqrt{g(\eps)}}{\sqrt{n}},\]
where $g(\eps)=\log (\frac{2}{\eps^2(1-\eps^2)})$ and $\mu=\sqrt{2^{\tilde{D}_{3/2}(\rho||\si)}+2^{-\tilde{D}_{1/2}(\rho||\si)}+1}<\infty$, where $\tilde{D}(\cdot\|\cdot)$ is the Petz-R\'enyi relative entropy. Moreover, for any sequence $\eps_n\in(0, 1)$ such that $\lim_n \eps_n<1$, we have
\[\lim_{n\to\infty} \frac{1}{n} D_{\max}^{\eps_n} (\rho^{\otimes n}\|\sigma^{\otimes n})\geq D(\rho\|\sigma).\]
\end{lemma}
\begin{proof}This is similar to the upper bound of Theorem \ref{thm:aep-channel-sec} ,which follows the finite-dimensional cases first proved in \cite{toma-thesis}. Using Lemma \ref{lemma-1} iii) and \ref{lemma:conti} i) that for $\gamma<\frac{1}{4}$
\begin{align*} D_{\max}^\eps(\rho^{\ten n} \|\si^{\ten n})\le& nD_{1+\gamma}(\rho\|\si)+ \frac{1}{\gamma} \log \frac{2}{\eps^2 \sqrt{1-\eps^2}} \\
\le & n\tilde{D}_{1+\gamma}(\rho\|\si)+\frac{1}{\gamma} \log \frac{2}{\eps^2 \sqrt{1-\eps^2}}\\
\le & nD(\rho\|\si)+ n\gamma (c_{\frac{1}{2}}(\rho||\sigma)^2)+\frac{1}{\gamma} \log \frac{2}{\eps^2 \sqrt{1-\eps^2}}\\
\le & nD(\rho\|\si)+ 4n\gamma \log^2(\mu)+\frac{1}{\gamma} g(\eps)
\end{align*}where in the second inequality we used again that $D_\al(\rho||\si)\leq \tilde{D}_\al(\rho||\si)$. Now choosing $\gamma=\frac{1}{\lambda \sqrt{n}}$ and optimizing for the parameter $\lambda$, we have
\[\frac{1}{n}D_{\max}^{\eps}(\rho^{\otimes n}\|\sigma^{\otimes n})\leq D(\rho\|\sigma)+ \frac{4 \log (\mu)\sqrt{g(\eps)}}{\sqrt{n}}.\]
The last assertion of this lemma follows from a straight forward generalization of \cite[Lemma 2.1]{fawzi-renner} using the relation  between the hypothesis relative entropy and max-relative entropy. We skip the details for brevity.
\end{proof}

%\begin{proof}
%This follows from the relations between $D_{\max}$ and $D_H$ in Lemma \ref{max and hyp rel}, and  \cite[Theorem 5]{ke-Li-1} in finite-dimensional cases. The von Neumann algebra case can be obtained using the work of Pautrat-Wang ( \cite{ke-Li}).
%\end{proof}
\begin{proof}[Proof of Theorem \ref{thm:application}] %The proof presented in the article \cite{fang} holds for tracial von Neumann algebras, and the general case follows by Haagerup reduction as in Theorem \ref{thm:chain rule1-sec}.
% It suffices to argue for the tracial case and the general case follows by Haagerup reduction as in Theorem \ref{thm:chain rule1-sec}.
First of all we note that Proposition 3.2 in  \cite{fang} works in this general setting without any modification, which is stated as follow:
For $\epsilon ,\epsilon' \in (0,1], m\in \mathbb{N}, \rho\in \mathcal{S}(\M)$ and $\sigma$ any positive functional on $\M$, and $\Phi,\Psi:\M_*\rightarrow \N_*$ two quantum channels, we have
\[D_{\max}^{m\eps+\sqrt{m\eps}+\eps'}(\Phi(\rho)^{\otimes m}\| \Psi(\sigma)^{\otimes m})\leq m D_{\max}^\eps(\rho\|\sigma) +
 \sup_{\nu \in  \mathcal{S}(\M)} D_{\max}^{\eps'}({\Phi(\nu)}^{\otimes m}\| {\Psi(\nu)}^{\otimes m})- m\log (1-\eps).\]
Applying the above inequality for $D_{\max}$ for the maps $\Phi^{\otimes n}, \Psi^{\otimes n}$ and $\rho^{\otimes n}, \sigma^{\otimes n}$, one have
 \begin{align}\label{eq-1}
 & D_{\max}^{m\eps+\sqrt{m\eps}+\eps'} ((\Phi^{\otimes n}(\rho^{\otimes n}))^{\otimes m}\| (\Psi^{\otimes n}(\sigma^{\otimes n}))^{\otimes m} ) \nonumber \\
 &\leq m D_{\max}^\eps(\rho^{\otimes n}\|\sigma^{\otimes n})
+ \sup_{\nu \in  \mathcal{S}(\M)} D_{\max}^{\eps'}(({\Phi^{\otimes n}(\nu))}^{\otimes m}\| {(\Psi^{\otimes n}(\nu))}^{\otimes m})- m\log (1-\eps).
 \end{align}
 Then by the Lemma \ref{lemma:last} above,
 \begin{align*}
&\frac{1}{m} D_{\max}^{\eps'}({\Phi^{\otimes n}(\nu))}^{\otimes m}\| {(\Psi^{\otimes n}(\nu))}^{\otimes m})\leq D({\Phi^{\otimes n}(\nu))}\| {(\Psi^{\otimes n}(\nu))})+ \frac{4 \log (\mu')\sqrt{g(\eps')}}{\sqrt{m}}.
 \end{align*}

Here, we need to bound the term
$\mu'=\sqrt{2^{\tilde{D}_{3/2}(\Phi^{\otimes n}(\nu)\|\Psi^{\otimes n}(\nu))}}+\sqrt{2^{-\tilde{D}_{1/2}(\Phi^{\otimes n}(\nu))\|\Psi^{\otimes n}(\nu))}}+1$, for a state  $\nu$ acting on $\M^{\otimes n}$. This follows from the properties of Petz-R\'enyi relative entropy and from the assumption that $D_{\max} ^{\textup{ref}}(\Phi\|\Psi)<\infty$.

Indeed, for the first term we use the fact that $\tilde{D}_{2-1/\alpha} (\cdot\| \cdot)\leq D_\alpha  (\cdot\| \cdot)$ for $\alpha\in (1, \infty)$ (see \cite{hiai-book}). Then
$$\tilde{D}_{3/2}(\Phi^{\otimes n}(\nu)\|\Psi^{\otimes n}(\nu))\leq {D}_{2}(\Phi^{\otimes n}(\nu)\|\Psi^{\otimes n}(\nu))\leq nD_{\max} ^{\textup{ref}}(\Phi\|\Psi)<\infty.$$
We used the fact that $D_{\max} ^{\textup{ref}}(\Phi ^{\otimes n}\|\Psi^{\otimes n})=n D_{\max} ^{\textup{ref}}(\Phi\|\Psi)$.

For the second term,  we note that $-\tilde{D}_{1/2}(\Phi^{\otimes n}(\nu)\|\Psi^{\otimes n}(\nu))\leq 0$, as the divergences are always non negative for a state  $\nu$.

%\[-\tilde{D}_{1/2}(\Phi^{\otimes n}(\nu)\|\Psi^{\otimes n}(\nu))\leq - {D}_{1/2}(\Phi^{\otimes n}(\nu)\|\Psi^{\otimes n}(\nu)) \leq -nD_{\max} ^{\textup{ref}}(\Phi\|\Psi).\]
%\OF{There seems to be a problem with the last equation... It doesn't go in the right direction. If we assume that everything is a state that the divergences should be nonnegative and so this term should be $\leq 0$.}

The remaining steps of the proof follows similarly as in \cite[Theorem 3.5]{fang}.
\end{proof}

\section*{Conflict of interests}
All authors declare that they have no conflicts of interest.
\section*{Acknowledgement}

We would like to thank Peter Brown, Fr\'ed\'eric Dupuis, Anthony Leverrier and Uta Meyer for multiple discussions on entropy accumulation beyond the finite-dimensional case as well as Mario Berta and Marco Tomamichel for discussions about asymptotic equipartition in von Neumann algebras. MR likes to thank Chris Schafhauser for insightful discussions. OF and MR acknowledge funding from the European Research Council (ERC Grant AlgoQIP, Agreement No. 851716), from a government grant managed by the Agence Nationale de la Recherche under the Plan France 2030 with the reference ANR-22-PETQ-0009 and from the QuantERA II Programme within the European Union’s Horizon 2020 research and innovation programme under Grant Agreement No 101017733. MR is supported by the Marie Sk\l{}odowska- Curie Fellowship from
the European Union’s Horizon Research and Innovation programme, grant Agreement No.
HORIZON-MSCA-2022-PF-01 (Project number: 101108117). LG was partially supported by NSF grant DMS-2154903.

\bibliography{AEP.bib}

\begin{thebibliography}{10}

\bibitem{Alberti-Uhlmann}
Peter~M. {Alberti} and Armin {Uhlmann}.
\newblock On bures-distance and $*$-algebraic transition probability between
  inner derived positive linear forms over $\mathrm{W^*}$-algebra.
\newblock {\em Acta Applicandae Mathematicae}, 60:1--37, 2000.

\bibitem{Anshu}
Anurag {Anshu}, Mario {Berta}, Rahul {Jain}, and Marco {Tomamichel}.
\newblock A minimax approach to one-shot entropy inequalities.
\newblock {\em Journal of Mathematical Physics}, 60(12), 2019.

\bibitem{Araki}
Huzihiro Araki.
\newblock {Relative entropy for states of von Neumann algebras II}.
\newblock {\em Publ. Res. Inst. Math. Sci}, 13, 1977.

\bibitem{DIQKD}
Rotem {Arnon-Friedman}, Fr\'ed\'eric {Dupuis}, Omar {Fawzi}, Renato {Renner},
  and Thomas {Viddic}.
\newblock Practical device independent quantum cryptography via entropy
  accumulation.
\newblock {\em Nat. Comm.}, 9(459), 2018.

\bibitem{Audenaert}
K.~M.~R. {Audenaert}, J.~{Calsamiglia}, R.~{Muñoz-Tapia}, E.~{Bagan}, Ll.
  {Masanes}, A.~{Acin}, and F.~{Verstraete}.
\newblock Discriminating states: The quantum chernoff bound.
\newblock {\em Phys. Rev. Lett.}, 98(16), 2007.

\bibitem{bergh2023infinite}
Bjarne Bergh, Jan Kochanowski, Robert Salzmann, and Nilanjana Datta.
\newblock Infinite dimensional asymmetric quantum channel discrimination.
\newblock {\em arXiv preprint arXiv:2308.12959}, 2023.

\bibitem{berta-smooth}
Mario Berta, Fabian Furrer, and Volkher~B. Scholz.
\newblock The smooth entropy formalism for von {N}eumann algebras.
\newblock {\em J. Math. Phys.}, 57(1):015213, 25, 2016.

\bibitem{bst}
Mario Berta, Volkher~B. Scholz, and Marco Tomamichel.
\newblock R\'{e}nyi divergences as weighted non-commutative vector-valued
  {$L_p$}-spaces.
\newblock {\em Annales Henri Poincar\'{e}}, 19(6), 2018.

\bibitem{berta-toma}
Mario Berta and Marco Tomamichel.
\newblock Chain rules for quantum channels.
\newblock {\em arXiv preprint arXiv:2204.11153}, 2022.

\bibitem{CMW}
Tom Cooney, Milan Mosonyi, and Mark~M. Wilde.
\newblock Strong converse exponents for a quantum channel discrimination
  problem and quantum-feedback-assisted communication.
\newblock {\em Communications in Mathematical Physics}, 344, 2016.

\bibitem{thoma-cover}
Thomas~M. {Cover} and Joy~A. {Thomas}.
\newblock {\em Elements of Information Theory, Second Edition Elements of
  Information Theory, Second Edition}.
\newblock 2005.

\bibitem{datta-hyp}
Nilanjana Datta, Milan Mosonyi, Min-Hsiu Hsieh, and Fernando
  G.~S.~L.~{Brandao}.
\newblock A smooth entropy approach to quantum hypothesis testing and the
  classical capacity of quantum channels.
\newblock {\em IEEE Transactions on Information Theory,}, 59(12), 2013.

\bibitem{datta2016second}
Nilanjana Datta, Yan Pautrat, and Cambyse Rouzé.
\newblock Second-order asymptotics for quantum hypothesis testing in settings
  beyond iid—quantum lattice systems and more.
\newblock {\em Journal of Mathematical Physics}, 57(6):062207, 2016.

\bibitem{datta-renner}
Nilanjana Datta and Renato Renner.
\newblock Smooth entropies and the quantum information spectrum.
\newblock {\em IEEE Trans. Inform. Theory}, 55(6):2807--2815, 2009.

\bibitem{junge-mult}
Igor {Devetak}, Marius {Junge}, Christoper {King}, and Mary~Beth {Ruskai}.
\newblock Multiplicativity of completely bounded p-norms implies a new
  additivity result.
\newblock {\em Communications in Mathematical Physics}, 266(1), 2006.

\bibitem{DF}
Frédéric {Dupuis} and Omar {Fawzi}.
\newblock Entropy accumulation with improved second-order term.
\newblock {\em IEEE Transactions on Information Theory}, 65, 2019.

\bibitem{DFR}
Frédéric Dupuis, Omar Fawzi, and Renato Renner.
\newblock Entropy accumulation.
\newblock {\em Communications in Mathematical Physics,}, 379, 2020.

\bibitem{fang}
Kun Fang, Omar Fawzi, Renato Renner, and David Sutter.
\newblock A chain rule for the quantum relative entropy.
\newblock {\em Phys. Rev. Lett.}, 124, 2020.

\bibitem{FHSW}
Thomas Faulkner, Stefan Hollands, Brian Swingle, and Yixu Wang.
\newblock Approximate recovery and relative entropy $\mathrm{I}$: General von
  neumann subalgebras.
\newblock {\em Communications in Mathematical Physics}, 389(1):349--397, 2022.

\bibitem{fawzi2021defining}
Hamza Fawzi and Omar Fawzi.
\newblock Defining quantum divergences via convex optimization.
\newblock {\em Quantum}, 5:387, 2021.

\bibitem{fawzi-renner}
Omar {Fawzi} and Renato {Renner}.
\newblock Quantum conditional mutual information and approximate markov chains.
\newblock {\em Commun. Math. Phys.}, 340, 2015.

\bibitem{furrer}
Fabian Furrer, Johan Aberg, and Renato Renner.
\newblock Min- and max-entropy in infinite dimensions.
\newblock {\em Comm. Math. Phys.}, 306(1):165--186, 2011.

\bibitem{FLO}
Keiichiro {Furuya}, Nima {Lashkari}, and Shoy {Ouseph}.
\newblock {Real-space RG, error correction and Petz map}.
\newblock {\em Journal of High Energy Physics}, 2022(1), 2022.

\bibitem{ganesan2020quantum}
Priyanga Ganesan, Li~Gao, Satish~K Pandey, and Sarah Plosker.
\newblock Quantum majorization on semi-finite von neumann algebras.
\newblock {\em Journal of Functional Analysis}, 279(7):108650, 2020.

\bibitem{gao2022complete}
Li~Gao, Marius Junge, Nicholas LaRacuente, and Haojian Li.
\newblock Complete order and relative entropy decay rates.
\newblock {\em arXiv preprint arXiv:2209.11684}, 2022.

\bibitem{gao2021recoverability}
Li~Gao and Mark~M Wilde.
\newblock Recoverability for optimized quantum f-divergences.
\newblock {\em Journal of Physics A: Mathematical and Theoretical}, 54(28),
  2021.

\bibitem{haagerup2010reduction}
Uffe Haagerup, Marius Junge, and Quanhua Xu.
\newblock A reduction method for noncommutative $l_p$-spaces and applications.
\newblock {\em Transactions of the American Mathematical Society}, 362, 2010.

\bibitem{hayashi}
Masahito Hayashi.
\newblock Optimal sequence of quantum measurements in the sense of
  \text{Stein's lemma} in quantum hypothesis testing.
\newblock {\em J. Phys. A: Math. Gen.}, 35, 2002.

\bibitem{hiai-book}
Fumio Hiai.
\newblock {\em Quantum $f$-divergences in von Neumann algebras---reversibility
  of quantum operations}.
\newblock Mathematical Physics Studies,. Springer, Singapore, 2021.

\bibitem{hiai-mosonyi}
Fumio {Hiai} and Milán {Mosonyi}.
\newblock Quantum rényi divergences and the strong converse exponent of state
  discrimination in operator algebras.
\newblock {\em Annales Henri Poincar\'e}, 24:1681–1724, 2023.

\bibitem{hiai-petz}
Fumio {Hiai} and D\'enes Petz.
\newblock The proper formula for relative entropy and its asymptotics in
  quantum probability.
\newblock {\em Communications in Mathematical Physics}, 143(1), 1991.

\bibitem{HS}
Stefan {Hollands} and Ko~{Sanders}.
\newblock Entanglement measures and their properties in quantum field theory.
\newblock {\em SpringerBriefs in Mathematical Physics}, 34, 2018.

\bibitem{jaksic}
V.~{Jak{\v{s}}ić}, Y.~{Ogata}, C.~A. {Pillet}, and R.~{Seiringer}.
\newblock Quantum hypothesis testing and non-equilibrium statistical mechanics.
\newblock {\em Reviews in Mathematical Physics}, 24(6), 2012.

\bibitem{jencova-1}
Anna {Jen{\v{c}}ová}.
\newblock Rényi relative entropies and noncommutative $\mathrm{L_p}$-spaces.
\newblock {\em Annales Henri Poincaré}, 19(8), 2018.

\bibitem{jencova-2}
Anna {Jen{\v{c}}ová}.
\newblock Rényi relative entropies and noncommutative $\mathrm{L_p}$-spaces
  $\mathrm{II}$.
\newblock {\em Annales Henri Poincaré}, 22, 2021.

\bibitem{junge2010mixed}
Marius Junge and Javier Parcet.
\newblock {\em Mixed-norm inequalities and operator space Lp embedding theory},
  volume 203.
\newblock Memoirs of American Mathematical Society, 2010.

\bibitem{kadison-pnas}
Richard {Kadison}.
\newblock On representations of finite type.
\newblock {\em Proc. Natl. Acad. Sci. USA}, 95(23), 1998.

\bibitem{wilde-2019}
Sumeet {Khatri}, Eneet {Kaur}, Saikat {Guha}, and Mark {Wilde}.
\newblock Second-order coding rates for key distillation in quantum key
  distribution.
\newblock {\em arXiv e-prints,}, 2019.

\bibitem{Kosaki}
Hideki Kosaki.
\newblock {Relative entropy of states: a variational expression}.
\newblock {\em J. Oper. Theory}, 16, 1986.

\bibitem{ke-Li-1}
Ke~{Li}.
\newblock Second-order asymptotics for quantum hypothesis testing.
\newblock {\em Annals of Statistics}, 42(1), 2014.

\bibitem{matsumoto-1}
Keiji {Matsumoto}.
\newblock A new quantum version of f-divergence.
\newblock {\em Reality and Measurement in Algebraic Quantum Theory, Springer
  Proceedings in Mathematics and Statistics}, 34:229--273, 2018.

\bibitem{gen-EAT}
Tony {Metger}, Omar {Fawzi}, David {Sutter}, and Renato {Renner}.
\newblock Generalised entropy accumulation.
\newblock {\em IEEE 63rd Annual Symposium on Foundations of Computer Science
  (FOCS)}, pages 844--850, 2022.

\bibitem{mosnyi-ogawa}
Milan {Mosonyi} and Tomohiro {Ogawa}.
\newblock Quantum hypothesis testing and the operational interpretation of the
  quantum rényi relative entropies.
\newblock {\em Commun. Math. Phys.}, 334:1617–1648, 2015.

\bibitem{nelson}
E.~Nelson.
\newblock Notes on non-commutative integration.
\newblock {\em J. Funct. Anal.}, 15:103--116, 1974.

\bibitem{szkola}
Michael Nussbaum and Arleta Szkoła.
\newblock The chernoff lower bound for symmetric quantum hypothesis testing.
\newblock {\em The Annals of Statistics}, 37:1040–57, 2009.

\bibitem{ogawa-nag}
Tomohiro Ogawa and Hiroshi Nagaoka.
\newblock \text{Strong converse and Stein’s lemma in quantum hypothesis
  testing}.
\newblock {\em IEEE Transactions on Information Theory}, 46(7), 2000.

\bibitem{petz-ohaya}
Masanori {Ohya} and D\'enes {Petz}.
\newblock {\em Quantum Entropy and Its Use}.
\newblock Springer, 1993.

\bibitem{ke-Li}
Yan {Pautrat} and Simeng {Wang}.
\newblock \text{Ke Li's} lemma for quantum hypothesis testing in general von
  neumann algebras.
\newblock {\em Annales Henri Poincar\'e}, 24:2323--2339, 2023.

\bibitem{pisier}
Gilles Pisier.
\newblock Noncommutative vector valued $\mathrm{L_p}$-spaces and completely
  p-summing maps.
\newblock {\em Ast\'erisque}, 247, 1998.

\bibitem{pisier2003introduction}
Gilles Pisier.
\newblock {\em Introduction to operator space theory}.
\newblock Number 294. Cambridge University Press, 2003.

\bibitem{renner2008security}
Renato {Renner}.
\newblock Security of quantum key distribution.
\newblock {\em Ph.D. dissertation, Swiss Federal Institute of Technology},
  2005.

\bibitem{sakai}
Shoichiro Sakai.
\newblock A \text{Radon-Nikodym} theorem in $\mathrm{W^*}$-algebras.
\newblock {\em Bulletin of the American Mathematical Society}, 71(1), 1965.

\bibitem{schumacher1995quantum}
Benjamin {Schumacher}.
\newblock Quantum coding.
\newblock {\em Phys. Rev. A}, 51, 1995.

\bibitem{takesaki-1}
Masamichi Takesaki.
\newblock {\em Theory of Operator Algebras I}.
\newblock Springer New York, 1979.

\bibitem{terp}
M.~Terp.
\newblock $\mathrm{L_p}$ spaces associated with von neumann algebras.
\newblock {\em Notes, Copenhagen Univ.}, 1981.

\bibitem{toma-thesis}
Marco Tomamichel.
\newblock A framework for non-asymptotic quantum information theory a framework
  for non-asymptotic quantum information theory (thesis).
\newblock {\em arXiv:1203.2142}, 2012.

\bibitem{tomamichel-book}
Marco Tomamichel.
\newblock {\em Quantum Information Processing with Finite Resources --
  Mathematical Foundations}, volume~5.
\newblock SpringerBriefs in Mathematical Physics, 2016.

\bibitem{tomamichel}
Marco Tomamichel, Roger Colbeck, and Renato Renner.
\newblock A fully quantum asymptotic equipartition property.
\newblock {\em IEEE Trans. Inform. Theory}, 55(12):5840--5847, 2009.

\bibitem{toma-haya}
Marco Tomamichel and Masahito Hayashi.
\newblock A hierarchy of information quantities for finite block length
  analysis of quantum tasks.
\newblock {\em IEEE Transactions on Information Theory}, 59:7693--7710, 2013.

\bibitem{uhlmann}
A.~Uhlmann.
\newblock The ``transition probability'' in the state space of a
  {$\sp*$}-algebra.
\newblock {\em Rep. Mathematical Phys.}, 9(2):273--279, 1976.

\bibitem{wang-wilde}
Xin {Wang} and Mark {Wilde}.
\newblock Resource theory of asymmetric distinguishability for quantum
  channels.
\newblock {\em Phys. Rev. Research}, 1, 2019.

\bibitem{amortized}
Mark~M. {Wilde}, Mario {Berta}, Christoph {Hirche}, and Eneet {Kaur}.
\newblock Amortized channel divergence for asymptotic quantum channel
  discrimination.
\newblock {\em Letters in Mathematical Physics}, 110(8), 2020.

\bibitem{witten}
Edward {Witten}.
\newblock Notes on some entanglement properties of quantum field theory.
\newblock {\em arXiv e-prints, arXiv:1803.04993}, 2018.

\bibitem{yngvason2005role}
Jakob Yngvason.
\newblock The role of type $\mathrm{III}$ factors in quantum field theory.
\newblock {\em Reports on Mathematical Physics}, 55(1):135--147, 2005.

\end{thebibliography}
\bibliographystyle{plain}
\end{document}